\documentclass[11pt]{elsarticle}

\usepackage{lineno}
\usepackage{hyperref}
\hypersetup{colorlinks=true, citecolor=blue, linkcolor=red, urlcolor=blue}
\usepackage{mathtools}
\usepackage{tyson}
\usepackage{fullpage}
\usepackage{amsmath}
\usepackage{amsfonts}
\usepackage{amssymb}
\usepackage{amsthm}
\usepackage{mathrsfs}

\usepackage{tikz}
\usetikzlibrary{arrows,backgrounds,calc,fit,decorations.pathreplacing,decorations.markings,shapes.geometric}
\tikzset{every fit/.append style=text badly centered}
\usepackage[all]{xy}
\usepackage{dashbox}
\newcommand\dboxed[1]{\dbox{\ensuremath{#1}}}
\usepackage{fancybox,framed}

\newtheorem{theorem}{Theorem}[section]

\newtheorem{corollary}[theorem]{Corollary}
\newtheorem{example}[theorem]{Example}
\newtheorem{definition}[theorem]{Definition}

\newtheorem{lemma}[theorem]{Lemma}

\theoremstyle{definition}
\newtheorem*{remark}{Remark}

\newcommand{\ii}{\mathfrak{i}}
\newcommand{\Holant}{\operatorname{Holant}}

\newcommand{\holant}[2]{\ensuremath{\Holant\left(#1 \mid #2\right)}}

\newcommand{\CSP}{\operatorname{\#CSP}}

\makeatletter
\newcommand{\rmnum}[1]{\romannumeral #1}
\newcommand{\Rmnum}[1]{\expandafter\@slowromancap\romannumeral #1@}

\numberwithin{equation}{section}

\date{May 22, 2020}

\makeatletter
\def\ps@pprintTitle{%
    \let\@oddhead\@empty
    \let\@evenhead\@empty
    \def\@oddfoot{\footnotesize\emph
         {To appear in Information and Computation ({\sl\url{https://doi.org/10.1016/j.ic.2020.104589}})  \hfill May 22, 2020}}%
    \let\@evenfoot\@oddfoot
    }
\makeatother

\journal{Information and Computation}

\begin{document}

\begin{frontmatter}

\title{{\bf Beyond \#CSP: A Dichotomy for  Counting Weighted
Eulerian Orientations with   ARS}}

%\author{Jin-Yi Cai\fnref{myfootnote}}

%\author{Zhiguo \fnref{myfootnote}}

%\author{Shuai Shao\fnref{myfootnote}}

\author[mymainaddress]{Jin-Yi Cai\fnref{fn1}}
\ead{jyc@cs.wisc.edu}

\author[mysecondaryaddress]{Zhiguo Fu\fnref{fn2}}
\ead{fuzg432@nenu.edu.cn}

\author[mymainaddress]{Shuai Shao\fnref{fn1}\corref{mycorrespondingauthor}}
\ead{sh@cs.wisc.edu}

\cortext[mycorrespondingauthor]{Corresponding author}
\fntext[fn1]{Supported by NSF CCF-1714275.} \fntext[fn2]{Supported by NSFC-61872076, 
 Natural Science Foundation of Jilin Province 20200201161JC and Fundamental Research Funds for Central Universities.}

\address[mymainaddress]{Department of Computer Sciences, University of Wisconsin-Madison, Madison, USA}
\address[mysecondaryaddress]{School of Information Science and Technology \& KLAS, Northeast Normal University, Changchun, China}

\begin{abstract}
\small{We define and explore a notion of unique prime factorization
for constraint functions,
and use this as a new tool
to  prove a complexity classification for counting weighted  \emph{Eulerian orientation} problems with \emph{arrow reversal symmetry} ({\sc ars}). 
We prove that all such problems are either polynomial-time computable or \#P-hard.
We show that the class of  weighted  \emph{Eulerian orientation} problems subsumes
all  weighted counting constraint satisfaction problems (\#CSP) on Boolean variables. More significantly,
we establish a novel connection between \#CSP
and counting weighted Eulerian orientation problems that is global in nature. This
connection is based on a structural determination of all half-weighted affine linear subspaces
over $\mathbb{Z}_2$, which is proved using M\"obius inversion.}
\end{abstract}

\begin{keyword}
\small{Eulerian orientation, Holant problem,  \#CSP, Unique prime factorization.}
\end{keyword}

\end{frontmatter}

%\linenumbers
%\setcounter{footnote}{0} 
\section{Introduction}\label{sec:intro}
A most basic fact in mathematics is that (rational) integers have a  
unique prime factorization (UPF). When we consider algebraic extensions,
certain primes remain so, while others can be further factored.
For example, for algebraic integers in the extension $
\mathbb{Q}[\sqrt{-2}] \supset \mathbb{Q}$, the rational prime $3$ is no longer a prime,
as $3 = (1 + {\mathfrak i}\sqrt{2})(1 - {\mathfrak i}\sqrt{2})$,
 but $5$ remains a prime.
Also for algebraic integers, sometimes
UPF still holds, but in general it fails.
For example, for algebraic integers in $\mathbb{Q}[\sqrt{-5}]$,
$6 = 2 \cdot 3 = (1 + {\mathfrak i}\sqrt{5})(1 - {\mathfrak i}\sqrt{5})$ 
are two distinct
factorizations into irreducible elements, and thus UPF fails.
%%%JYC don't change to UFD
Unique prime factorization  is only restored in Kummer's theory of prime ideals.
Perhaps more important than what is in the details, in many mathematical investigations
the \emph{idea} of UPF has been a powerful guiding principle.
Is there some analogous notion that is useful for the complexity classification
program of counting problems?
%Perhaps equally well-known to any student of Calculus, 
%if $f(x, y, \ldots)$ is a
%function  satisfying some mild smoothness condition,
%then we have $\frac{\partial^2 f}{\partial x \partial y}
%= \frac{\partial^2 f}{\partial y \partial x}$, i.e., the two partial derivative
%operations commute. Again, we ask,
%is there some analogous notion of this that is useful for the classification
%program of counting problems?

In this paper we show that for a class (which is in fact \emph{broader} than all weighted  counting CSP)
called
counting weighted Eulerian orientation problems,
 we \emph{can} develop an analogous notion for the purpose of  complexity classification of counting problems.  
We prove a
UPF for constraint functions,
(and for the relevant constraint functions there is
also a need to extend the scope where the factorization takes place).
Next we 
develop a merging operation on  constraint functions.
%such that, again,
%, again under mild conditions, 
%the two partial derivative
%operations commute. 
%In our setting for constraint functions,
%this notion is very elementary to define.
The main technical challenge turns out to be the
%the analysis of
interplay of these merging operations and  
%a special
%operation and 
the divisibility relation in the unique prime factorization.

%But we got ahead of ourselves. First 
Let us define the
counting problems that we wish to classify.
Let $G$ be an undirected Eulerian graph, 
i.e., every vertex has even degree. 
An \emph{Eulerian orientation} of $G$ is an orientation of its edges such that at each vertex
the number of incoming edges is equal to the number of outgoing edges.
Mihail and Winkler showed that counting the number of Eulerian orientations of an undirected Eulerian graph is \#P-complete \cite{MW}.
In this paper, we consider counting weighted Eulerian orientations  (\#EO problems), formulated as 
a partition function defined by constraint functions %, also called signatures 
placed at each vertex that represent
weightings of various local Eulerian configurations. There are a host of problems
in statistical physics and combinatorics that can be formulated as computing this partition function.

This partition function is defined as follows.
Suppose $G$ is given, and
each vertex $v$ is associated with a weight function $f_v$ taking values in
$\mathbb{C}$. 
The 
%ordered list of 
incident edges  to  $v$  are totally ordered and correspond to input variables
to $f_v$.
Each edge has two ends, and an  orientation of the edge is  denoted by assigning 
0 to the head and 1 to the tail.  
Thus, locally at every vertex $v$,  a
$0$ represents an incoming edge and a $1$ represents an outgoing edge. (A loop has both ends at the same vertex.)
An Eulerian orientation corresponds to
an assignment to each edge ($01$ or $10$) where the numbers of 0's and 1's at
each $v$ are equal. Then
a vertex $v$ contributes a weight by the local constraint function
$f_v$ evaluated according to the local assignment.
 The weight of an Eulerian orientation is the product of weights over all vertices. 
 If every constraint function  $f_v$  is only nonzero   when the numbers of input 0's and 1's are equal, i.e., the support of $f_v$ is on half weighted inputs,
 then only Eulerian orientations contribute nonzero values.
 The partition function of \#EO is the sum of weights over all %possible 
 Eulerian orientations. This is a sum-of-product computation.

%and can be expressed as Holant problems \cite{Cai-Lu-Xia}
%\footnote{
%Indeed, under the holographic transformation $Z = 
%\frac{1}{sqrt{2}} 
%\left[\begin{smallmatrix}
%1 & 1 \\
%{\mathfrak i} & -{\mathfrak i}
%\end{smallmatrix}\right]$, this is transformd to a standard
%Holant problem where {\sc ars} translates to precisely real-valued
%constraint functions.
%}. 
  %By expressing 
  
  %Formal definitions will be given in section \ref. 

The significance of  this partition function 
%of Eulerian orientations 
is evidenced by its appearance in several different
fields. 
In statistical physics, the partition function of the so-called \emph{ice-type
model} \cite{Pauling, Slater, rys1963uber} is the partition function of Eulerian orientations of some underlying Eulerian graph. 
When it is restricted on the square lattice, it is the
classical \emph{six-vertex model} \cite{Pauling}.
Literally thousands of papers have been written in the literature
dealing with the six-vertex model, mainly on the square lattice, but also on other
graphs~\cite{Baxter}. 
%%% TO Shuai!
 It is perhaps one of the three most intensely studied models
in statistical physics, along with that of ferromagnetic
Ising and monomer-dimer.
The ``exact solution'' %integrability
of the (unweighted) six-vertex model with periodic boundary conditions by Lieb is
an important milestone in statistical physics~\cite{Lie67a,Lie67b,Lieb}.
  In combinatorics, the resolution of the \emph{Alternating
Sign Matrix} conjecture is linked to the classical {six-vertex model} with the domain wall 
boundary condition~%\cite{Mills-Robbins-Rumsey,Zeilberger,Kuperberg,Korepin,Bressoud}.
\cite{Korepin,Mills-Robbins-Rumsey,Zeilberger,Kuperberg,Bressoud}.
% see the book by David Bressoud,Proofs and Confirmations
 Las Vergnas also observed that the evaluation of the Tutte polynomial 
at the point $(3, 3)$ is related to the partition function of  Eulerian orientations
with a specific weight assignment~\cite{tutte}.
Cai, Fu and Xia proved a complexity classification of the partition function
of the six-vertex model on general 4-regular graphs \cite{cfx}.
%this is the special case  of the main theorem in this paper
%restricted to 4-regular graphs.
%%%  JYC not quite. b/c ARS

The \#EO problems have an intrinsic significance in the classification program for counting problems. 
At first glance, the \#EO framework may appear to be specialized as it requires all constraint functions to be supported on  half weighted inputs.
However, surprisingly, we show that it encompasses all counting constraint satisfaction problems (\#CSP)
on Boolean variables.
\begin{theorem}\label{csp-by-eo-intro}
Every complex-weighted {\rm\#CSP} problem  on Boolean variables can be expressed as an {\rm\#EO} problem.
\end{theorem}
A more detailed statement for this theorem is given 
as Theorem~\ref{thm:csp-by-eo}.
By this theorem, a complexity dichotomy
 for  \#EO problems would 
 %automatically give a 
 generalize the
 complexity dichotomy
 for $\CSP$ problems on the Boolean domain, 
 which is already known\footnote{In fact, the full complexity dichotomy
 for $\CSP$ problems on any finite domain is also known~\cite{bulatov-ccsp, dyer-richerby, bulatov2012csp, cai-chen-lu-nonnegative-csp, cai-chen-csp}.} \cite{creignou1996complexity, dyer2009complexity, bulatov2009complexity, Cai-Lu-Xia-csp}.

In physics, a symmetry called the \emph{arrow reversal symmetry} ({\sc{ars}})\footnote{On a square lattice, when there is no
external electric field, physical considerations imply that the model is unchanged
by reversing all arrows \cite{Baxter}.  This ‘zero field’
model includes the {\it ice}~\cite{Pauling}, {\it KDP}~\cite{Slater} and {\it F}~\cite{rys1963uber} models as special cases.
} is usually assumed. 
For complex-valued  local constraint functions $f$,
{\sc ars} requires that $f(\overline{\alpha})=\overline{f(\alpha)}$ 
for all $\alpha$,
where $\overline{f(\alpha)}$ denotes the complex conjugation of $f(\alpha)$,
and $\overline{\alpha}$  denotes the bit-wise complement of $\alpha$.
For real-valued functions, this is $f(\overline{\alpha})=f(\alpha)$.
%
%i.e., for a real-valued local constraint function $f_v$,
%$f_v(\bar{\alpha})=f_v(\alpha)$ where $\bar{\alpha}$ is the bit-wise complement of $\alpha$.
This means, if we flip the
orientations of all edges, the (real)
function value $f$ is unchanged (or for complex $f$
it is changed  to its complex conjugation).
We will see that it is not only natural
 but also necessary for the proof that we
 consider complex values.

In complexity theory, there is a more intrinsic reason for considering the arrow reversal symmetry.
Under the holographic transformation $Z =
\frac{1}{\sqrt{2}}
\left[\begin{smallmatrix}
1 & 1 \\
{\mathfrak i} & -{\mathfrak i}
\end{smallmatrix}\right]$, any \#EO problem %constrained by  {\sc ars} 
is transformed to a standard
Holant problem, and  complex-valued constraint functions with the {\sc ars} restriction translate  precisely to real-valued
constraint functions.
\begin{theorem}
A complex-valued constraint function $f$ satisfies  {\sc ars} if and only if after the holographic transformation $Z$, $f$ is transformed to a real-valued function. 
\end{theorem}
We will describe holographic transformation and prove this theorem in Section~\ref{apen-holo}.
By this theorem,  the classification of \#EO problems can serve as building blocks towards a classification of general Holant problems. 
In Section~\ref{sec-conclusion}, we  will give a  partial map (Figure~\ref{fig:structure}) for the complexity classification program for Holant problems on the Boolean domain, where we indicate how the framework of \#EO problems fits in this program.
 %While physicists usually focus on real values (especially in classical physics),
 %our results generalize to complex values, and in fact 
% in Section \ref{prelim}. 
%In the  setting of complex values, 
% {\sc ars} is stated as $f_v(\bar{\alpha})=\overline{f_v(\alpha)}$ where $\overline{f_v(\alpha)}$ denotes the complex conjugation of $f_v(\alpha)$.  
% That is, if we flip the orientations of all edges incident to vertex $v$, the function value $f_v$ becomes its complex conjugate.
 
 In this paper, we prove a complexity dichotomy theorem for \#EO problems with {\sc ars}.
\begin{theorem}\label{main-intro}
For any set $\mathcal{F}$  of complex-valued  constraint functions satisfying the {\sc ars} condition,
the weighted counting problem for Eulerian orientations 
specified by $\mathcal{F}$ is either computable in polynomial time, or \#P-hard.
The polynomial time tractability criterion on $\mathcal{F}$ is explicit.
\end{theorem}
 The explicit  tractability criterion  will be stated  
 in terms of certain constraint function sets to be defined in Section~\ref{prelim}, where we give some preliminaries.
 The formal statement of the dichotomy theorem (Theorem~\ref{main}) and a proof outline will be given in Section~\ref{sec-main}.
 %the partition function of Eulerian orientations 
%, with arbitrary constraint functions satisfying  {\sc ars}.
 %arrow reversal symmetry.
The most technical part of the proof is an induction that guarantees
a suitable interplay between the merging operation and 
%divisibility in 
 unique factorization.  However this inductive proof
only works when the arity (i.e., the number of 
input variables) of these constraint functions is sufficiently high
(arity $\geqslant 10$).  For lower arity cases, we prove it separately. 
We will give the inductive proof and handle lower arity cases in Section \ref{no-neq_4}.
In addition, we discover a novel connection between 
%counting constraint satisfaction problems (\#CSP)
\#CSP and \#EO problems, which will be introduced in Section \ref{neq_4}.
%counting weighted Eulerian orientation problems. 
We establish this connection by a  simulation of \#CSP  by \#EO 
that is global in nature.
%We will introduce this connection in Section \ref{sec: connection}
Owing to
the restriction that the support of every constraint function in  \#EO problems is on half weighted inputs, (provably) no local replacement reduction
can work. 
One crucial ingredient of the proof  in Section \ref{neq_4} is the determination of all half weighted affine
linear subspaces over $\mathbb{Z}_2$. This result should be of independent interest. 
This determination
uses techniques including M\"obius inversion.

%%set of constraint functions ${\cal F}$,
%the problem \#CSP(${\cal F}$) is:
%The input is a bipartite graph $G = (U, V, E)$, where %$U$ are
%variables, $V$ are labeled by constraint functions %from ${\cal F}$,
%and $E$ describes how the constraints are applied.
%% on the variables.
%The output is the sum, over all assignments to variables in $U$,
%of the product of constraint function evaluations in $V$.

%The main technical contribution of this paper is to introduce
%unique prime factorization for signatures as a method to prove complexity
%dichotomies. 
%Combined with merging operations, it is a powerful technique to analyze
%the complexity of signatures, and
%build inductive proofs. This method should be more widely
%applicable in the study of Holant problems. 
%The result of this paper can serve as building blocks towards a classification of real-valued Holant problems (with no symmetry assumptions on the signatures). 
%Under a suitable holographic transformation (see Appendix) these \#EO problems with {\sc ars} correspond
%to precisely a class of real valued Holant problems. 
%The techniques of this paper can handle all, not
%necessarily symmetric, local constraint functions on half weight support with {\sc ars}, where
%it seems that some most intricate cases of a full real valued Holant dichotomy lie. 

\section{Preliminaries}\label{prelim}
\subsection{Definitions and notations}
%Let $\mathbb{Z}_2^{2n}=\{(x_1, x_2, \ldots, x_{2n})\mid x_1, \ldots, x_{2n} \in \mathbb{Z}_2\}$.
A constraint function $f$, or a signature,   of arity $m$ is a map $\mathbb{Z}_2^{m} \rightarrow \mathbb{C}$.
In this paper we assume all signatures have positive arities, and since we study weighted Eulerian orientation problems, we may assume each signature $f$ has an even arity, say
$2n$, for some $n = n(f) >0$.
We use $f^\alpha$ to denote $f(\alpha)$.
%In this paper, we assume signatures satisfy the array reverse symmetry ({\sc ars}), i.e., $f(\alpha)=\overline{f}(\overline{\alpha})$ for all $\alpha \in \mathbb{Z}_2^{2n}$. 
%For simplicity, we may  somewhere.
The support  of a signature $f$ of
arity $2n$ is %the set of inputs on which $f$ is not zero. That is
$\mathscr{S}(f) = \{\alpha \in \mathbb{Z}_2^{2n} \mid f^\alpha \neq 0\}$.
 If $\mathscr{S}(f)=\emptyset$, i.e., $f$ is identically $0$, we say $f$ is a zero signature and denote it by $f\equiv 0$.
Otherwise, $f$ is a nonzero signature.
 %Otherwise, $f$ is a nonzero signature. 
%Let $\mathscr{H}_{2n}$
%denote the set of inputs of length $2n$, of which the Hamming weight is $n$, half of its length. 
%That is, 
We use ${\rm wt}({\alpha})$ to denote the Hamming weight of $\alpha\in \mathbb{Z}_2^{2n}$.
Let $\mathscr{H}_{2n}=\{\alpha \in \mathbb{Z}_2^{2n} \mid
 {\rm wt}(\alpha)= n \}$.
Note that if $\alpha$ belongs to   $\mathscr{H}_{2n}$ then its complement string
$\overline{\alpha}$ also belongs to 
$\mathscr{H}_{2n}$. 
A signature $f$ of arity $2n$ is 
 an \emph{Eulerian orientation} (EO) signature
 if $\mathscr{S}(f) \subseteq \mathscr{H}_{2n}$.
%In this paper, without further specification, we use $\mathcal{F}$ to denote a set  of EO signatures.
Let $\mathcal{F}$ be any fixed set of  EO signatures. 
An (EO-)signature grid
$\Omega=(G, \pi)$ over $\mathcal{F}$
 is a tuple, where $G = (V,E)$
is an Eulerian graph without isolated vertex (i.e., every vertex has positive even degree),
 $\pi$ labels each $v\in V$ with a signature
$f_v\in\mathcal{F}$ of arity ${\operatorname{deg}(v)}$,
and labels the incident edges
$E(v)$ at $v$ with input variables of $f_v$.
For any Eulerian graph $G$, let ${\rm EO}(G)$
be the set of all Eulerian orientations of $G$.
We view each edge as having two ends, and  an  orientation of the edge  is denoted by assigning 
0 to the head and 1 to the tail.  
%Thus, an edge with an orientation is assigned 10 from tail to head.
An Eulerian orientation corresponds to
an assignment to the ends of each edge where the numbers of 0's and 1's at
each $v$ are equal. Then
a vertex $v$ contributes a weight by the local constraint function
$f_v$ evaluated according to the local assignment.
Each $\sigma \in {\rm EO}(G)$ gives an evaluation $\prod_{v\in V}f_v(\sigma|_{E(v)})$, where $\sigma|_{E(v)}$ assigns
$0$ to an incoming edge and $1$ to an outgoing edge.

\begin{definition}[{\rm \#EO} problems]
A {\rm \#EO} problem {\rm \#EO}$(\mathcal{\mathcal{F}})$
 specified by a  set  $\mathcal{F}$ of {\rm EO}
signatures  is the following:
The input is an {\rm (EO-)}signature grid $\Omega=(G, \pi)$ over $\mathcal{F}$;
the output is
the partition function of $\Omega$,
$${\rm\#EO}_{\Omega}=\sum\limits_{\sigma \in {\rm EO}(G)}\prod_{v\in V}f_v(\sigma|_{E_{(v)}}).$$
%The problem {\rm \#EO}$(\mathcal{\mathcal{F}})$ is to compute ${\rm\#EO}_{\Omega}$ for instances $\Omega$ over $\mathcal{F}$.
\end{definition}
\begin{example}[Unweighted {\rm \#EO} problem]
Let $\mathcal{F}=\{f_2, f_4, \ldots f_{2n}, \ldots\},$ where $f_{2n}^\alpha = 1$ when ${\rm wt}(\alpha)=n$ and $f_{2n}^\alpha =0$ otherwise. 
Then {\rm  \#EO}$(\mathcal{\mathcal{F}})$  counts the number of Eulerian orientations.
\end{example}

\begin{example}[Six-vertex models] %Let $\mathcal{F}=\{f_{\rm six}\}$, where 
Let $f_{\rm six}$ be an  {\rm EO} signature of arity $4$, where $f_{\rm six}^{0011}=f_{\rm six}^{1100}=a$, $f_{\rm six}^{0101}=f_{\rm six}^{1010}=b$, $f_{\rm six}^{0110}=f_{\rm six}^{1001}=c$, (where $a, b, c \in \mathbb{R}^+$). Then {\rm \#EO}$(f_{\rm six})$ is the  classical six-vertex model satisfying  {\sc ars}
with real parameters $(a, b, c)$.
\end{example}

{\rm \#EO} problems can be viewed as special cases of Holant problems.
A (general) signature grid
$\Omega=(G, \pi)$ over a set of  arbitrary (not necessarily EO) signatures $\mathcal{F}$
 is a tuple as before, where $G$ is a graph,
 each vertex $v$ is assigned some $f_v \in \mathcal{F}$ of arity
$\deg(v)$, with incident edges $E(v)$ labeled as input variables.
% = (V,E)$.
% Every edge of $G$ is assigned a value $0$ or $1$ by an assignment $\sigma$.
We consider all 0-1 edge assignments $\sigma: E(G)\rightarrow\{0, 1\}$.
Each  $\sigma$ gives an evaluation
$\prod_{v\in V(G)}f_v(\sigma|_{E(v)})$, where $\sigma|_{E(v)}$
denotes the restriction of $\sigma$ to $E(v)$. 
%$\pi$ labels each $v\in V$ with a function
%$f_v\in\mathcal{F}$ of arity ${\operatorname{deg}(v)}$,
%and the incident edges
%$E(v)$ at $v$ with input variables of $f_v$.
% Each
%$f_v$ maps
%$\{0, 1\}^{\operatorname{deg}(v)}$ to $\mathbb{C}$.%By doing this, we
%Readers are referred to Section  1 of the book \cite{jcbook} for more details.

\begin{definition}[Holant problems]
%Fix a set $\mathcal{F}$ of  which are not necessarily EO signatures, a signature grid...
The input to the problem {\rm Holant}$(\mathcal{F})$
is a signature grid $\Omega=(G, \pi)$ over $\mathcal{F}$.
The output is  the partition function,
$$\Holant_{\Omega} =\sum\limits_{\sigma:E(G)\rightarrow\{0, 1\}}\prod_{v\in V(G)}f_v(\sigma|_{E_{(v)}}).$$
%The {\rm Holant} problem parameterized by the set $\mathcal{F}$ is denoted by {\rm Holant}$(\mathcal{F})$.
The bipartite Holant problems $\holant{\mathcal{F}}{\mathcal{G}}$
are {\rm Holant} problems
 over bipartite graphs $H = (U,V,E)$,
where each vertex in $U$ or $V$ is labeled by a signature in $\mathcal{F}$ or $\mathcal{G}$ respectively.
\end{definition}
We use $\leqslant_T$ (and $\equiv_T$) to denote polynomial-time Turing reductions (and
equivalences, respectively).
In this paper without loss of generality we assume all signatures have positive arity, i.e., any signature grid $G$ has no isolated vertices.
We use $\neq_2$ to denote the binary  {\sc Disequality} signature
with truth table $(f^{00}, f^{01}, f^{10}, f^{11})
= (0, 1, 1, 0)$. It can also be expressed by the matrix $$N=\left[\begin{matrix}
0 & 1 \\
 1 & 0
\end{matrix}\right]$$ with one variable indexing rows  and the other
indexing
 columns respectively.

\begin{lemma}\label{eo=holant}
{\rm \#EO}$({\mathcal{F}})\equiv_T \Holant(\neq_2 \mid \mathcal{F}).$ 
\end{lemma}

\begin{proof}
If  $\Omega=(G, \pi)$ is an instance of {\rm \#EO}$({\mathcal{F}})$,
we add a middle  vertex on each edge  of $G$ and label it  by $\neq_2$.
This defines  an instance $\Omega'$ of $\Holant(\neq_2 \mid \mathcal{F})$
with a  bipartite graph $H$ (which is the edge-vertex incidence graph of $G$), where every edge  of $G$ is
 broken into two. There is a 1-1 correspondence of the terms  in the partition functions
${\rm\#EO}_{\Omega}$ and 
$\Holant_{\Omega'}$. The process is obviously reversible.
%
%The problem \#EO$(\mathcal{\mathcal{F}})$ can be expressed by the Holant framework \cite{}. 
%Given an instance $\Omega=(G, \pi)$ of {\rm \#EO}$({\mathcal{F}})$ where $G=(E, V)$,
%Note that in an Eulerian orientation, edge assigns opposite values (incoming/outgoing) for its two endpoints. 
%We can add a vertex on each edge and label this vertex by $\neq_2$. Now, each edge is broken into two edges connected by $\neq_2$. We assign the new edges with $\{0, 1\}$. Then the 2 possible orientations of the original edge is expressed by the 2 possible assignments $(0, 1)$ or $(1, 0)$ of the two new edges. The new graph is a bipartite graph, then \#EO$(\mathcal{\mathcal{F}})$ is equivalent to Holant$(\neq_2 \mid \mathcal{F})$.
\end{proof}

%\begin{definition}[\#CSP]

%\end{definition}

Counting constraint satisfaction problems (\#CSP) can also be expressed as Holant problems (Lemma 1.2 in \cite{jcbook}).
We use $=_n$ to denote the {\sc Equality} signature of arity $n$,
which takes value $1$ on the all-0 or all-1 inputs,
and $0$ elsewhere.
% $ ${\rm wt}(\alpha)=0$ or $n$ and $=_n(\alpha) =0$ elsewhere.
Let $\mathcal{EQ}=\{=_1, =_2, \ldots, =_n, \ldots\}$ denote the set of all {\sc Equality} signatures.

\begin{lemma}%\cite{jcbook}
$\CSP(\mathcal{F})\equiv_T \holant{\mathcal{EQ}}{\mathcal{F}}$.
\end{lemma}
%\begin{proof}
%Refer to Lemma 1.2 of \cite{jcbook}. 
%\end{proof}

A signature $f$ of arity $n \geqslant 2$ can be expressed as 
%a $2^k \times 2^{n-k}$ matrix $M_{[k],[n-k]}(f)$, which lists the $2^n$ many values of $f$ with $k$ variables $[k]$ 
a $4 \times 2^{n-2}$ matrix $M_{(i j)}(f)$, 
%%% JYC I didn't remove this \qed.  i guess you meant as a sign to yourself???
which lists  the $2^n$ values of $f$ with variables $(x_i, x_j)
\in \{0, 1\}^2$ as row index and the assignments of the other
 $n-2$ variables as column index, both in lexicographic order.
%%% JYC: i feel it is ok not to talk about this ordering withing the
% n-2 varibles. We can take for granted this lexicographic order as the
% default that also order the n-2 varibles lexicographically.
%(%$[n-2]$ denotes a permutation of the other $n-2$ variables. %and we use $\overline{[n-2]}$ to denote the reversal permutation of $[n-2]$. 
%In this paper, there is no need to specify the permutation of those $n-2$ variables and we omit it.)
%This matrix consists of $4$ row vectors indexed by $(x_i, x_j)=(0, 0), (0, 1), (1, 0)$ 
%and $(1, 1)$. 
That is,
$$M_{(i j)}(f)=\begin{bmatrix}
f^{00,00\ldots0} & f^{00, 00\ldots1} & \ldots & f^{00, 11\ldots1}\\
f^{01,00\ldots0} & f^{01, 00\ldots1} & \ldots & f^{01, 11\ldots1}\\
f^{10,00\ldots0} & f^{10, 00\ldots1} & \ldots & f^{10, 11\ldots1}\\
f^{11,00\ldots0} & f^{11, 00\ldots1} & \ldots & f^{11, 11\ldots1}\\
\end{bmatrix}=\begin{bmatrix}
{\bf {f}}^{00}_{ij}\\
{\bf {f}}^{01}_{ij}\\
{\bf {f}}^{10}_{ij}\\
{\bf {f}}^{11}_{ij}\\
\end{bmatrix}.$$
We use $ {\bf {f}}^{ab}_{ij}$ %${\overrightarrow{f}}_{01{y}}$, 
%${\overrightarrow{f}}_{10{y}}$, and 
%${\overrightarrow{f}}_{11{y}}$ 
to denote the row vector indexed by $(x_i, x_j)=(a, b)$ in $M_{(i j)}(f)$, $\overline{{\bf {f}}^{ab}_{ij}}$ denotes
its conjugate, %of ${\bf {f}}^{ab}_{ij}$
and ${{\bf {f}}^{ab}_{ij}}^{\tt T}$ denotes its transpose.
%of ${\bf {f}}^{ab}_{ij}$. %where ${y} 
%\in \{0, 1\}^{n-2}$ are indexed by the other $n-2$ variables in the implied order.
%The entries of ${\bf {f}}^{ab}_{[ij],[n-2]}$ are indexed by variables $[n-2]$ in lexicographic order.
Consider the reversal vector  ${\bf {f}}^{ab}_{ij}N^{\otimes (n-2)}$
 of ${\bf {f}}^{ab}_{ij}$. 
%That is, ${\bf {f}}^{ab}_{ij}N^{\otimes (n-2)}$. 
We have
\begin{equation*}
    {\bf {f}}^{ab}_{ij}N^{\otimes (n-2)}=(f^{ab,11\ldots1}, f^{ab, 11\ldots0}, \ldots, f^{ab, 00\ldots0})\\=(\overline{f^{\bar a\bar b,00\ldots0}}, \overline{f^{\bar a\bar b,00\ldots1}}, \ldots, \overline{f^{\bar a\bar b, 11\ldots1}})=\overline{{\bf {f}}^{\bar a\bar b}_{ij}}.
\end{equation*}
The second equality holds by {\sc ars}. By taking transpose, we have $N^{\otimes (n-2)}{{\bf {f}}^{ab}_{ij}}^{\tt T}={\overline{{\bf {f}}^{\bar a\bar b}_{ij}}}^{\tt T}$.
%That is, ${\bf {f}}_{ab}^{[ij],[n-2]}M(\neq_2)^{\otimes n-2}$
When $(ij)$ is clear from the context, we omit the subscript $(ij)$.
%Furthermore, 

%We use ${\overrightarrow{f}}_{00{\bar{y}}}$ to denote the reverse vector
%of ${\overrightarrow{f}}_{00{{y}}}$. By  arrow reversal symmetry, we know that, e.g.,  ${\overrightarrow{f}}_{00{\bar{y}}}=\overline{\overrightarrow{f}_{11{y}}}$.

%We use $\neq_2$ to denote the binary  {\sc Disequality} function $(0,1, 1, 0)^T$.  
We generalize the notion of binary {\sc Disequality} to signatures of 
higher arities. 
A signature $f$ of arity $2n$ is called a  {\sc Disequality} signature of arity $2n$, denoted by $\neq_{2n}$, if $f =1$ when $x_1=x_2=\ldots=x_n\neq x_{n+1}=x_{n+2}=\ldots=x_{2n},$
and 0 otherwise. By permuting its variables
the  {\sc Disequality} signature  of arity $2n$ also defines
${2n \choose n}/2$ functions which we also call {\sc Disequality} signatures.
Alternatively we can define a {\sc Disequality} signature as follows. 
A signature $f$ of arity $2n$ is a  {\sc Disequality} signature
%, denoted by $\neq_{2n}$,
if $f^\alpha=f^{\bar \alpha}=1$ for some $\alpha$ with ${\rm wt}(\alpha)=n$, and $f^\beta=0$ elsewhere. Since there are ${2n \choose n}/2$ many possible ways to choose a pair $\{\alpha, \bar \alpha\}$, there are ${2n \choose n}/2$ many {\sc Disequality} signatures  of arity $2n$.
They are equivalent for the complexity of Holant problems; once we have one
we have them all, i.e., if $f$ is obtained from
$\ne_{2n}$ by relabeling its variables, then
for any set  $\mathcal{F}$ of EO signatures,
\[{\rm \#EO}(\mathcal{F}, f) \equiv_T {\rm \#EO}(\mathcal{F}, \ne_{2n}).\]
%$f^\alpha=f^{\overline \alpha}=1$ for some $\alpha$ with ${\rm wt}(\alpha)=n$, and for all  %$\beta 
%\not = \alpha$ or ${\overline \alpha}$, we have $f^\beta=0$.
%That is, $\neq_{2k}$ can be viewed as  two $=_k$ connected by a $\neq_2$.
%Note that unlike {\sc Equality} signatures, {\sc Disequality} signatures of arity $k\geqslant 4$ are not unique. 
%They may differ by a permutation of the variables, but one can see that the permutation does not matter. 
Let $\mathcal{DEQ}=\{\neq_2, \neq_4, \ldots, \neq_{2n}, \ldots \}$ denote the set of all {\sc Disequality} signatures.
%A signature $f$ of arity $2n$ is called a weighted {\sc Disequality} signature if $\overline{f^\alpha}= f^{\bar{\alpha}}\in \mathbb{C}$ for some $\alpha$ with ${\rm wt}(\alpha)=n$, and $f^\beta=0$ elsewhere. 
%We use $\mathcal{WDEQ}$ to denote the set of all weighted {\sc Disequality} signatures. 
%We will show that $\mathcal{WDEQ}$ essentially captures the tractable classes of \#EO$(\mathcal{F})$.

\subsection{Gadget construction}
One basic tool used throughout this paper is gadget construction.
An $\mathcal{F}$-gate (or simply a gadget) is similar to a signature grid $(G, \pi)$ for $\Holant(\mathcal{F})$ except that $G = (V,E,D)$ is a graph with internal edges $E$ and dangling edges $D$.
The dangling edges $D$ define input variables for the $\mathcal{F}$-gate.
We order the regular edges in $E$ from  $1$ to $m$, and order the dangling edges in $D$ from $1$ to $n$.
Then the  $\mathcal{F}$-gate  defines a function $f$
\[
f(y_1, \dotsc, y_n) = \sum_{\sigma: E \rightarrow\{0, 1\}} \prod_{v\in V}f_v(\hat{\sigma}\mid_{E(v)}),
\]
where $(y_1, \dotsc, y_n) \in \{0, 1\}^n$ is an assignment on the dangling edges, $\hat{\sigma}$ is the extension of  $\sigma$ on $E$ by the assignment $(y_1, \ldots, y_n)$, and $f_v$ is the signature assigned at each vertex $v \in V$.
%and $H(x_1, \dotsc, x_m, y_1, \dotsc, y_n)$ is the value of the signature grid on an assignment of all edges in $G$,
%which is the product of evaluations at all vertices in $V$.
This function $f$ is called the signature of the $\mathcal{F}$-gate.
There may be no internal edges in an $\mathcal{F}$-gate. In this case, $f$ is simply a tensor product of these signatures $f_v$, i.e., $f={\bigotimes}_{v\in V}f_v$ (with possibly a permutation of variables).
We say a signature $f$ is \emph{realizable} from a signature set $\mathcal{F}$ by gadget construction
if $f$ is the signature of an 
 $\mathcal{F}$-gate. 
 In the following, when we describe a  gadget construction we identify a vertex in
 the  graph   $G = (V,E,D)$ for the gadget with the  signature labeling that vertex.

If $f$ is realizable from a set $\mathcal{F}$,
then we can freely include $f$  in the set of signatures while preserving the complexity (Lemma 1.3 in \cite{jcbook}). 
Note that by Lemma \ref{eo=holant}, in the setting of \#EO($\mathcal{F}$) problems (i.e., for \holant{\neq}{\mathcal{F}}), every edge in a gadget construction is effectively labeled by $\neq_2$.
%we connect two signatures using $\neq_2$.
%However, 
%in order to use gadget construction as a tool to build polynomial reductions, one still need 

A basic gadget construction is \emph{merging}. Given a signature $f$ of arity $n$, we can connect two variables $x_i$ and $x_j$ of $f$ using $\neq_2$, and it gives a signature of arity $n-2$. We use $\partial_{(ij)}f$ to denote this signature and $\partial_{(ij)}f=f^{01}_{ij}+f^{10}_{ij}$, where $f^{ab}_{ij}$ denotes the signature obtained by setting $(x_i, x_j)=(a, b)\in \{0, 1\}^2$. We use
$f^{ab}_{ij}$ to denote a function, and ${\bf f}^{ab}_{ij}$ to denote a vector that lists the truth table of $f^{ab}_{ij}$ in a given order.
We may further merge variables $x_u$ and $x_v$ of $\partial_{(ij)}f$ for any $\{u, v\}$ disjoint with $\{i, j\}$, and we use $\partial_{(uv)(ij)}f=\partial_{(uv)}(\partial_{(ij)}f)$ to denote the realized signature. %Since we can merge $x_i$ and $x_j$ first and then $x_u$ and $x_v$, or $x_u$ and $x_v$ first and then $x_i$ and $x_j$.
Note that these two merging operations commute, $\partial_{(uv)(ij)}f=\partial_{(ij)(uv)}f$.  This is illustrated in the following commutative diagram.
\begin{equation*}
\begin{tikzpicture}[every node/.style={midway}]
  \matrix[column sep={8em,between origins}, row sep={4em}] at (0,0) {
    \node(f) {$f$}; & \node[align=left](f') {$\partial_{(ij)}f$}; \\
    \node(P) {$\partial_{(uv)}f$}; & \node[align=left](P') {\hspace{15ex}$\partial_{(uv)(ij)}f=\partial_{(ij)(uv)}f$};\\
  };
  \draw[<-] (P) -- (f) node[anchor=east]  {$\partial_{(uv)}$};
  %\draw[->] (R/I) -- (S) node[anchor=north]  {$\psi$};
  \draw[->] (f) -- (f') node[anchor=south] {$\partial_{(ij)}$};
  \draw[->] (f') -- (P') node[anchor=west] {$\partial_{(uv)}$};
  \draw[ shorten >=-15ex, shorten <=7.5ex, ->] (P) -- (P') node[anchor=north] {\hspace{21ex}$\partial_{(ij)}$};
\end{tikzpicture}
\end{equation*}
(We adopt the notation $\partial$ for the similarity
of the merging operation  with taking partial derivatives.
They both reduce the number of variables, they both are linear,
and under mild smoothness conditions
we know $\frac{\partial^2 f}{\partial x \partial y}
= \frac{\partial^2 f}{\partial y \partial x}$. 
The 
commutativity of $\partial_{(ij)}f$ and $\partial_{(uv)}f$ is a key property in our proof.)
%(We explain why we adapt the notation $\partial$ here. Let  $f$ be a
%function satisfying some mild smoothness conditions,
%then we know $\frac{\partial^2 f}{\partial x \partial y}
%= \frac{\partial^2 f}{\partial y \partial x}$, i.e., the two partial derivative
%operations commute.  The key property we will use for merging operations is the commutability.)
%All gadget constructions can be essentially realized by mergeing operations. 
If by merging any two variables of $f$, we can only realize the zero signature, then we show that $f$ itself is ``almost'' a zero signature.

\begin{lemma}\label{zero}
Let $f$ be a signature of arity $n\geqslant 3$. If for every pair  of 
distinct indices $\{i, j\}$, by merging variables $x_i$ and $x_j$ of $f$ we have $\partial_{(ij)}f\equiv0$, then $f^\alpha=0$ for any $\alpha$ with $0 < {\rm wt}(\alpha)< n$. In particular, if $f$ is an {\rm EO} signature, then $f\equiv 0$.
\end{lemma}

\begin{proof}
Suppose there exists some $\alpha$, where $0 <  {\rm wt}(\alpha) < n$, such that $f^{\alpha}\neq 0$.
Since $\alpha$ is not all-0 nor all-1 and $\alpha$ has length at least $3$, we can find three bits in some order such that on these three bits, $\alpha$ takes value $001$ or $110$.
Without loss of generality, we assume they are the first three bits of $\alpha$ and we denote $\alpha$ by $001\delta$ or $110\delta$ ($\delta$ maybe empty). 
We first consider the case that $\alpha=001\delta$.
%the proof for the other one is similar. 
Consider another two strings $\beta=010\delta$ and $\gamma=100\delta$. 
Note that if we merge variables $x_1$ and $x_2$ of $f$, we get $\partial_{(12)}f$, 
its entry $(\partial_{(12)}f)^{0\delta}$ on the input $0\delta$ 
(for bit positions 3 to $n$) is the sum of $f^{010\delta}$ and $f^{100\delta}$. 
Since $\partial_{(12)}f\equiv0$, we have $$f^{010\delta}+f^{100\delta}=0.$$
Similarly, by merging variables $x_1$ and $x_3$, we have $$f^{001\delta}+f^{100\delta}=0,$$ and by merging variables $x_2$ and $x_3$, we have $$f^{001\delta}+f^{010\delta}=0.$$
These three equations have only a trivial solution, $f^{001\delta}=f^{010\delta}=f^{100\delta}=0$. A contradiction.

If $\alpha=110\delta$, the proof is symmetric.
%let $\beta=101\delta$ and $\gamma=011\delta$.
%Similarly, by considering $\partial_{(12)}f, \partial_{(13)}f, %\partial_{(23)}f\equiv 0$, we get %$f^{110\delta}=f^{101\delta}=f^{011\delta}=0$. Contradiction.
%$\bar{\alpha}=001\bar{\delta}$. 
%We have $f^{\bar{\alpha}}=\bar{f}^{\alpha}\neq 0$. The same proof holds for $f^{\bar{\alpha}}$. 
\end{proof}

The following lemma  makes sure that a realized signature 
is suitable for \#EO problems.
%we have the following lemma.
\begin{lemma}\label{lemma:eo-gates}
Any signature realizable from a set $\mathcal{F}$ of {\rm EO} signatures satisfying   {\sc ars} is also an {\rm EO} signature  satisfying   {\sc ars}. 
\end{lemma}
\begin{proof}
By definition $\partial_{(ij)}f=f^{01}_{ij}+f^{10}_{ij}$.
Hence  for any EO signature satisfying   {\sc ars}, after merging any two variables, the realized signature is still  an EO signature  satisfying   {\sc ars}.
Then, suppose $f$ is  realized by a graph $G$ with dangling
edges and $n$ vertices labeled by signatures $f_1, f_2, \ldots, f_n \in \mathcal{F}$. We first cut all  internal edges in $G$ and  get 
the signature $f'=f_1\otimes f_2 \otimes \cdots \otimes f_n$.
Clearly $f'$ is an EO signature satisfying   {\sc ars}  since all $f_i$ are. Then, $f$ can be realized by \emph{merging} (with $\neq_2$) all cut edges of $f'$ in a sequence.
After each merging operation, the realized signature is an EO signature  satisfying   {\sc ars}, and hence $f$ is an EO signature  satisfying   {\sc ars}.
\end{proof}

%The first part of this lemma says that, if we have a directed graph with some dangling edges such that at each vertex $v$, the in-degree and the out-degree are the same,
%then among its dangling edges, the numbers of incoming and outgoing edges are also the same. 
Having established Lemma~\ref{lemma:eo-gates}, we have the following reduction. 
% i don't think there is any need to cite this below.
%General statements of this form can be found in~\cite{jcbook} (e.g.,
%Lemma 1.3.)
\begin{lemma}
%\cite{jcbook}
%%% i don't know we should cite this lemma specifically to the book
If $f$ is realizable from a set $\mathcal{F}$ of {\rm EO} signatures, then ${\rm\#EO}(\{f\}\cup \mathcal{F})\equiv_T\rm{\#EO}(\mathcal{F})$.
\end{lemma}
%\begin{proof}
%Refer to Lemma 1.3 of \cite{jcbook}. 
%\end{proof}

A particular gadget construction often used in this paper is \emph{mating}. Given an EO signature $f$ of arity $n\geqslant 3$, we connect two copies of $f$ in the following manner:
Fix a set $S$ of $n-2$ variables among all $n$ variables of $f$. For each $x_k\in S$, connect $x_k$ of one copy of $f$ with $x_k$ of the other copy using $\neq_2$. 
%(See Figure \ref{} ).
The two variables $x_i, x_j$ that are not in $S$ are called dangling variables.
%of each copy of $f$ and connect each variable in one set to the same variable in the other set using $\neq_2$. 
%That is, connect $x_1$ to $x_1$, $\ldots$ and $x_l$ to $x_l$. 
%Suppose they are $x_i$ and $x_j$. 
Then, 
this mating construction  realizes a signature of arity $4$, denoted by $\mathfrak m_{ij}f$.
It can be represented by matrix multiplication. 
%that are not mated by a mating construction. 
%We use $\mathfrak m_{ij}f$ to denote the arity $4$ signature realized by this mating construction. 
We have 
\begin{equation}\label{m-form}
\begin{aligned}[b]
M(\mathfrak m_{ij}f)&=M_{(ij)}(f)N^{\otimes (n-2)}M^{\tt T}_{(ij)}(f)\\
&=\begin{bmatrix}
{\bf {f}}^{00}\\
{\bf {f}}^{01}\\
{\bf {f}}^{10}\\
{\bf {f}}^{11}\\
\end{bmatrix}
\left[\begin{matrix}
0 & 1 \\
 1 & 0
\end{matrix}\right]^{\otimes (n-2)} 
\left[\begin{matrix}
{{\bf {f}}^{00}}^{\tt T} &{{\bf {f}}^{01}}^{\tt T}
&{{\bf {f}}^{10}}^{\tt T}&{{\bf {f}}^{11}}^{\tt T}
\end{matrix}\right]\\
&=\begin{bmatrix}
{\bf {f}}^{00}\\
{\bf {f}}^{01}\\
{\bf {f}}^{10}\\
{\bf {f}}^{11}\\
\end{bmatrix}
\left[\begin{matrix}
{\overline{{\bf {f}}^{11}}}^{\tt T} &{\overline{{\bf {f}}^{10}}}^{\tt T}
&{\overline{{\bf {f}}^{01}}}^{\tt T}&{\overline{{\bf {f}}^{00}}}^{\tt T}
\end{matrix}\right]\\
&=\left[\begin{matrix}
0 & 0 & 0& |{\bf f}^{00}|^2\\
0 & \langle{\bf f}^{01}, {\bf f}^{10}\rangle & |{\bf f}^{01}|^2 & 0\\
0 & |{\bf f}^{10}|^2 & \langle{\bf f}^{10}, {\bf f}^{01}\rangle  & 0\\
|{\bf f}^{11}|^2 & 0 & 0 & 0
\end{matrix}\right],\\
\end{aligned}
\end{equation}
%Notice that $\left[\begin{matrix}
%0 & 1 \\
% 1 & 0
%\end{matrix}\right]^{\otimes (n-2)}$ reverses all columns of 
%$\left[\begin{matrix}
%\overrightarrow{f}^T_{00{y}} &\overrightarrow{f}^T_{01{y}}
%&\overrightarrow{f}^T_{10{y}}&\overrightarrow{f}^T_{11{y}}
%\end{matrix}\right]$; the third equality is by  arrow reversal symmetry;
where $\langle \cdot, \cdot\rangle$ denotes the (complex) inner product and $|\cdot|$ denotes the  norm defined by this inner product.
Note that $|\langle{\bf f}^{01}, {\bf f}^{10}\rangle|^2\leqslant|{\bf f}^{01}|^2|{\bf f}^{10}|^2$ by Cauchy-Schwarz inequality.

\subsection{Signature factorization}
Recall that in this paper
all signatures have positive arity.
A nonzero signature $g$ \emph{divides} $f$, denoted by $g \mid f$, if there is a signature $h$ such that  $f=g \otimes h$ (with possibly a permutation of variables)
%=h \otimes g$ 
%%% the permutation includes this other order
or there is a constant $\lambda$ such that $f= \lambda \cdot g$.
In the latter case, if $\lambda \neq 0$, then we also have $f \mid g$ since $g= \frac{1}{\lambda} \cdot f$.
%When $h$ is a nonzero signature of arity $0$, i.e. a  constant $c \not = 0$,
%it is a \emph{unit} and $g \otimes h$ is just $cg$. 
For nonzero signatures, if both $g\mid f$ and $f \mid g$, then they are nonzero constant multiples of
each other, and 
we say $g$ is an \emph{associate} of $f$, 
denoted by $g \sim f$.
In terms of  this division relation, we can define \emph{irreducible} signatures and \emph{prime} signatures. We will show that they are equivalent, and this gives us the \emph{unique prime factorization} of signatures\footnote{The factorization of signatures is synonymous with the decomposition of multipartite quantum states in quantum information theory. 
There, the uniqueness of decomposition is usually assumed as a common knowledge.
%in quantum information theory community. 
To our best knowledge, we are not aware of any formal proof.}.

\begin{definition}[Irreducible signatures]
A nonzero signature $f$ is irreducible if $g \mid f$ implies that $g \sim f$.
%On the contrary, 
We say a signature is reducible if it is not irreducible or it is a zero signature.
By definition, if a signature $f$ of arity greater than 1 is reducible, then there is a factorization
%We say a signature $f$ is reducible if 
$f = g \otimes h$,
%up to a permutation of variables, 
for some signatures $g$ and $h$ (of positive arities). %All zero signatures (of arity greater than 1) are reducible.
\end{definition}
%%% 0 is reducible
% unit , arity-0 nonzero constant is neither
% arity 1 nonzero is irre
% nonzero arity >= 2 : reducible if there is a decomp 
% of f = lower arity g \otimes lower arity  h
% and irre otherwise.

\begin{definition}[Prime signatures]
A nonzero signature $f$ is a prime signature, if
for any nonzero signatures $g$ and  $h$,
$f \mid g \otimes h$ 
implies that $f\mid g$ or $f \mid h$. 
\end{definition}

\begin{lemma}\label{prime=irreducible}
The notions of  irreducible signatures and prime signatures are equivalent.
\end{lemma}
\begin{proof}
Suppose $f$ is a prime signature. 
If  $f$ is not irreducible,  then there is
a nonzero  signature $g$ such that
$g \mid f$ but not $g \sim f$.
So there is a signature $h$ (of arity $\geqslant 1$) 
such that 
$f = g \otimes h$, 
up to a permutation
of variables ($h \not\equiv 0$ due to $f \not\equiv 0$).  Then $f \mid  g \otimes h$ and  by being a prime,
either $f \mid g$ or $f \mid h$. This is impossible because both $g$ and $h$ have lower arity  than  $f$.

Now, suppose $f$ is irreducible and let $f \mid g\otimes h$,
where $g$ and $h$ are nonzero 
 signatures (of arity $\geqslant 1$). 
 If $f \sim g\otimes h$, then   $f=(\lambda g)\otimes h$
 for some  constant $\lambda \not = 0$. This contradicts
  $f$ being irreducible.
 Thus, there is a nonzero signature $e$ 
 (of arity $\geqslant 1$)
 such that, up to a permutation
 of variables,
\begin{equation}\label{efgh}
    e\otimes f =  g\otimes h.
\end{equation}
Consider the  scope of $f$, i.e., its set of variables. 
Suppose it intersects with the scopes of both $g$ and $h$. Since $e\not \equiv 0$, we can pick an input $\beta$ of $e$ such that $e^\beta = \lambda_1 \neq 0$. %($\beta$ maybe empty when $e$ is a unit).
  By setting the variables in the scope of $e$ to $\beta$ on both sides of (\ref{efgh}),  we have $$\lambda_1 \cdot f
  = g' \otimes h',$$
  %=g^{\bullet\circ\beta_1} \otimes %h^{\bullet\circ\beta_2},$$
  where $g'$ and $h'$
  %$g^{\bullet\circ\beta_1}$ and $ %h^{\bullet\circ\beta_2}$ 
  denote the resulting signatures from $g$ and $h$ respectively, both of
  which have a non-empty scope, i.e., having arity $\geqslant 1$.
  %Thus both are not constants. 
  This is a contradiction to  $f$ being irreducible.
  %($\beta_1$ and $\beta_2$ maybe empty). Note that  the %scope of $f$ still intersects with the scopes of both %$g^{\cdot\circ\beta_1}$ and $h^{\cdot\circ\beta_2}$. %Thus, $g^{\cdot\circ\beta_1}$ and %$h^{\cdot\circ\beta_2}$ are not units. 

Hence the scope of $f$ is a subset of the scope of either  $g$ or $h$. Suppose it is $g$, then
the scope of $h$ is a subset of the scope of $e$.
 Since $h\not\equiv 0$, we can pick an input $\alpha$ of $h$ such that $h^\alpha = \lambda_2 \neq 0$. %($\alpha$ maybe empty when $h$ is an unit). 
  By setting the variables in the scope of $h$ to $\alpha$ on both sides of (\ref{efgh}),
 we have $$e'\otimes f=\lambda_2 \cdot g,$$ where $e'$ denotes the resulting signature by setting
 $\alpha$ in $e$. 
 %obtained by setting the variables of $e$ that are in the scope of $h$ to $\alpha$. %That is, $g=e^{\cdot\circ\beta}/h^\beta\otimes f$
 Thus, we have $f \mid g$. Similarly,
 if the scope of $f$ is a subset of the scope of $h$, then  we have $f\mid h$.
\end{proof}

A prime factorization of a signature $f$ is $f=g_1\otimes \ldots \otimes g_k$ up to a permutation of
variables, where each $g_i$ is a prime signature (irreducible). 
%When $k=1$, we know $f$ itself is irreducible. 
Start with any nonzero signature, if we 
keep factoring reducible signatures and induct
on arity, any nonzero $f$ has a factorization
into irreducible (prime)  signatures.  The following
important lemma says that the prime factorization of a nonzero  signature is unique up to the order of the tensor factors and constant scaling factors. 
It can be proved using Lemma~\ref{prime=irreducible} and a standard argument, 
which we omit.

\begin{lemma}[Unique prime factorization]\label{unique}
Every nonzero signature $f$ has a prime factorization.
If  $f$ has  prime factorizations
 $f=g_1\otimes \ldots \otimes g_k$ and $f=h_1\otimes \ldots \otimes h_\ell$,
both up to a permutation of variables,
then $k=\ell$ and after reordering the factors we have $g_i \sim h_i$  for all $i$. 
%where $\lambda_i$ are nonzero constants.
\end{lemma}

%\begin{definition}[reducible signatures]
%A signature $f$ is reducible if there are non-unit signatures $g$
%\end{definition}

%A nonzero non-unit signature is  reducible if there are non-unit $g$ and $h$ such that $f=g\otimes h$.
%In the following, when we write down $f=g\otimes h$, we assume $g$ and $h$ are not units. 
If a vertex $v$ in a signature grid 
is labeled by a reducible signature $f=g\otimes h$, we can replace the vertex $v$ 
by two vertices $v_1$ and $v_2$ and label $v_1$ with $g$ and $v_2$ with $h$, respectively.
The incident edges of $v$ become incident edges of $v_1$ and $v_2$ respectively
according to the partition of variables of $f$ in the tensor product of $g$ and $h$.  This does not
change the Holant value.
Clearly, $f=g\otimes h$ is realizable from $\{g, h\}$. On the other hand,  Lin and Wang proved in~\cite{beida} 
(Corollary 3.3) that, from a reducible signature $f=g\otimes h$ we can freely replace $f$ by $g$ and $h$ while preserving the complexity of a  Holant problem.  To apply their theorem to \#EO problems we need to take care of one subtlety, namely from the factorization $f=g\otimes h$ 
 we have to get EO signatures $g$ and $h$, so that we can discuss 
  \#EO problems involving
  signatures $g$ and $h$.
  This is true for  EO signatures $f$  
  satisfying  {\sc ars}.

%, provided that $g$ and $h$ are suitable for  \#EO problems.

%One may wonder if $f$ is a reducible EO signature satisfying the arrow reversal symmetry, whether it can be factorized into two EO signatures that also satisfy the arrow reversal symmetry. We have the follow two lemmas

\begin{lemma}\label{decom-eo}
Let $f$ be a nonzero reducible {\rm EO} signature  satisfying  {\sc ars}. Then,
\begin{enumerate}
    \item for any factorization  $f=g\otimes h$,
$g$ and $h$ are both {\rm EO} signatures;
\item there exists a factorization $f=g\otimes h$ such that $g$ and $h$ are both {\rm EO} signatures satisfying  {\sc ars}.
\end{enumerate}
\end{lemma}

\begin{proof}
\begin{enumerate}
    \item 

Since $f\not\equiv 0$, we know $g\not\equiv 0$ and $h \not\equiv 0$ for any factorization  $f=g\otimes h$.
For a contradiction,
suppose there is a factorization  $f=g\otimes h$ such that $g$ is not an EO signature. 
Then, there is an input $\alpha$ of $g$ such that $g^{\alpha}\neq 0$, and ${\rm wt}(\alpha)\neq {\rm wt}(\bar{\alpha})$. (This is true no matter whether $g$ has even  or odd arity.)
Since $h\not\equiv 0$, there is an input $\beta$ of $h$ such that $h^{\beta}\neq 0$.
Note that $\alpha\circ\beta$ is an input of $f$ and we  have 
$$f^{\alpha\circ\beta}=g^{\alpha}\cdot h^{\beta}\neq 0.$$
Moreover, since $f$ satisfies  {\sc ars}, we have
$$0 \not = f^{\bar\alpha\circ\bar\beta}=g^{\bar\alpha}\cdot h^{\bar\beta}.$$
 Then, we know $g^{\bar\alpha}\neq 0$, and hence
 $$f^{\bar\alpha\circ\beta}=g^{\bar\alpha}\cdot h^{\beta} \neq 0.$$
 However, notice that ${\rm wt}(\alpha\circ\beta)
 ={\rm wt}(\alpha)
 + {\rm wt}(\beta)
 \neq {\rm wt}(\bar\alpha)
 + {\rm wt}(\beta)
 =
 {\rm wt}(\bar\alpha\circ\beta)$.
 %because ${\rm wt}(\alpha)\neq {\rm wt}(\bar\alpha)$. 
 This implies that $\mathscr{S}(f)\not\subseteq \mathscr{H}_{{\rm arity}(f)}$,  contradicting   $f$ being an EO signature.
 Thus, for any $f=g\otimes h$, $f$
 and $g$ are both EO signatures.
 \item 
 By the first part of the proof, it suffices to show that there exists a factorization $f=g\otimes h$ such that $g$ and $h$ both satisfy {\sc ars}.
 Suppose $f=g\otimes h$.
Since $f\not\equiv 0$, there is $\alpha\circ\beta$ such that $f^{\alpha\circ\beta}=g^\alpha \cdot h^\beta\neq 0.$ Since $f$ satisfy {\sc ars}, we have 
$$\overline{g^\alpha} \cdot \overline{h^\beta}=\overline{f^{\alpha\circ\beta}}=f^{\bar \alpha\circ \bar\beta}=g^{\bar \alpha} \cdot h^{\bar \beta} \not = 0,$$ and also
$${g^\alpha} \cdot h^{\bar \beta}={f^{\alpha\circ\bar \beta}}=\overline{f^{\bar \alpha\circ \beta}}=\overline{g^{\bar \alpha}} \cdot \overline{h^{ \beta}} \not = 0.$$ 
Multiply these two equalities, and cancel a nonzero
common
factor,  we have 
$$|g^\alpha|^2=|g^{\bar \alpha}|^2.$$
Since $g^\alpha$ and $g^{\bar \alpha}$ have the same norm, we can pick a scalar
$\lambda=1/( g^\alpha g^{\bar \alpha})^{1/2}$
such that $\overline{\lambda g^\alpha}=\lambda g^{\bar \alpha}.$ We have $f=(\lambda g)\otimes (\frac{1}{\lambda} h)$ and we will show $\lambda g$ and $\frac{1}{\lambda} h$ satisfy the {\sc ars} condition.
We rename $\lambda g$ and $\frac{1}{\lambda} h$ by $g$ and $h$, and now we can assume there is an $\alpha$ such that $\overline{g^\alpha}=g^{\bar \alpha} \neq 0$.
For any input $\beta$ of $h$, we have 
$$\overline{g^\alpha} \cdot \overline{h^\beta}=\overline{f^{\alpha\circ\beta}}=f^{\bar \alpha\circ \bar\beta}=g^{\bar \alpha} \cdot h^{\bar \beta}=\overline{g^\alpha} \cdot h^{\bar \beta},$$
and hence, $\overline{h^\beta}=h^{\bar \beta}$.
Hence $h\not \equiv 0$ satisfies the {\sc ars}
condition. We can pick a particular $\beta$ such that $\overline{h^\beta}=h^{\bar \beta}\neq 0$. Then, for any input $\alpha'$ of $g$, since $f$ satisfies  the {\sc ars}
condition, we have $\overline{g^{\alpha'}} \cdot \overline{h^\beta}=g^{\overline {\alpha'}} \cdot h^{\overline \beta}=g^{\overline{\alpha'}}\cdot \overline{h^\beta}$, and hence $\overline{g^{\alpha'}}=g^{\overline {\alpha'}}$. That is, $g$ also satisfies  {\sc ars}. \qedhere
 \end{enumerate}
\end{proof}
\begin{remark}
\begin{enumerate}
    \item The first part does not hold without assuming {\sc ars}. For example, $f=(0, 0, 1, 0)=(0, 1)\otimes(1, 0)$, where $(0, 0, 1, 0)$ is an EO signature but $(0, 1)$ and $(1, 0)$ are not.
    \item  The second part does not hold when we restrict signatures to real valued signatures. For example, $f=
(0, 1, -1, 0)^{\otimes 2}=
\left[\begin{smallmatrix}
0 & 0 &0 & 0&\\
0 & 1 &-1 & 0&\\
0 & -1 &1 & 0&\\
0 & 0 &0 & 0&\\
\end{smallmatrix}\right]$ is a real valued reducible signature satisfying   {\sc ars}. (Note that while
$f$ satisfies  {\sc ars}, the binary factor $(0, 1, -1, 0)$ does not satisfy  {\sc ars}.) In fact it cannot be factorized into two real valued signatures that also satisfy   {\sc ars}.
To see that, first  $f$  cannot  
have a unary tensor factor by Lemma \ref{decom-eo}.  
So any factorization of $f$ must be with two binary signatures.
A real valued binary EO signature satisfying  {\sc ars} has the form $(0, a, a, 0)$ where $a \in \mathbb{R}$. 
Thus, if $f$ is a tensor product of  such two signatures,
we have $f=(0, a, a, 0)\otimes (0, b, b, 0)$ up to a permutation
of variables, which implies all nonzero entries of $f$ are the same, a contradiction.
% because, setting x_1 =0 and x_1 =1 gives different and nonzero support patterns,  
% and true for other x_i too.
% because support pattern is symmetric for all x_i
%A moment reflection shows that the only possible factorization has the form
%$f_1(x_1, x_2) \otimes f_2(x_3, x_4)$, and both must be a scalar multiple of
%$0, 1, -1, 0)$.  Thus arrow reversal symmetry is impossible among real valued factors.
%%% setting f_1(x_1, x_3) say. then x_1=x_3=0 gives different and nonzero support pattern 
However, $f=(0, \ii, -\ii, 0)\otimes (0, -\ii, \ii, 0)$, which is a tensor product of two complex valued signatures satisfying  {\sc ars}.  Thus, by going to the
complex field we restored a closure property for 
prime factorizations of signatures satisfying {\sc ars}.
\end{enumerate}
\end{remark}

In the following, when we say that a nonzero EO signature $f$ satisfying  {\sc ars} has a factorization $g \otimes h$,  we always assume $g$ and $h$ are EO signatures satisfying  {\sc ars}. 
The following %\emph{factorization} 
lemma follows from Corollary 3.3  of  Lin and Wang \cite{beida}. %We state it for our setting. 

\begin{lemma}\label{lem-decom}  %\footnote{no vanishing}
If a nonzero {\rm EO} signature $f$ satisfying  {\sc ars} has a factorization $g \otimes h$, then 
\[{\rm\#EO}(\{g, h\}\cup \mathcal{F})\equiv_T{\rm\#EO}(\{f\} \cup \mathcal{F})\] for any {\rm EO} signature set $\mathcal{F}$. In this case, we also say $g$ and $h$ are realizable from $f$.
\end{lemma}
%\begin{remark}
%If $f$ has a factorization $g \otimes h$, we also say $g$ and $h$ is realizable from $f$.
%\end{remark}

%In this paper, we particularly focus on signatures that are  tensor products of binary EO signatures satisfying the {\sc ars}. 

\begin{definition}
We use $\mathcal{B}$ to denote the set of signatures that are  tensor products 
of (one or more) binary {\rm EO} signatures satisfying {\sc ars}. 
\end{definition}

Note that a binary {\rm EO} signature $b(x_i, x_j)$ over variables $x_i, x_j$ satisfying  {\sc ars} has a particular form $(0, a, \bar{a}, 0)$ for 
some $a \in \mathbb{C}$, and it is irreducible when $a\neq 0$.
%%%JYC I seem to need this in lm 5.5
Signatures in $\mathcal{B}$ satisfy a closure property:
If $f \in \mathcal{B}$ and $b(y,z) \in \mathcal{B}$
where $y$ and $z$ are distinct from variables of $f$,
and we connect one variable $x$ of $f$ via $\neq_2$
with $y$, the new signature also belongs to
$\mathcal{B}$.  This  is easily verified by
$\left[\begin{smallmatrix}
0 & a \\
\bar{a} & 0 
\end{smallmatrix}\right]
\left[\begin{smallmatrix}
0 & 1 \\
1 & 0 
\end{smallmatrix}\right]
\left[\begin{smallmatrix}
0 & b\\
\bar{b} & 0 
\end{smallmatrix}\right]
=
\left[\begin{smallmatrix}
0 & ab \\
\overline{ab} & 0 
\end{smallmatrix}\right]$.
We use the notation
$b^\ii(\cdot, \cdot)$
to denote the binary 
signature $(0, \ii, -\ii, 0)$
in $\mathcal{B}$.
For signatures in $\mathcal{B}$, we have the following result.
%That is, $\mathcal{B}=\{f=b_1\otimes b_2 \ldots \otimes b_n\mid b_i is a \}$
%We say $b(x_i, x_j)|f$ if $f$ has a factorization $f=b(x_i, x_j)\otimes g$ where $g$ is an EO signature on variables except $x_i, x_j$ satisfying 
%When $|a|=1$, we say such a binary signature is normalized. 

% That is, the binary signature $b(x_i, x_j)=(0, i, -i, 0)$ on variables $x_i$, $x_j$ is a tensor factor of $f_{(i', j')}$. 
 
 % We use $\mathcal{WDEQ}^{\otimes}$ to denote signatures that 
 %are tensor products of one or more
 %signatures in $\mathcal{WDEQ}$.

 \begin{lemma}\label{0ii0}
Suppose $f\in \mathcal{B}$. Then $\partial_{(ij)}f\equiv 0$ iff the signature $b^\ii(x_i, x_j)=(0, \ii, -\ii, 0)$ divides $f$. 
\end{lemma}
\begin{proof}
If  $b^\ii(x_i, x_j)|f$, then $f=b^\ii(x_i, x_j)\otimes g$, where $g$ is a constant or
a signature on variables other than $x_i, x_j$. %Note that $f'$ maybe a zero signature, which means $f$ itself is a zero signature.
We have $\partial_{(ij)}f=(\ii-\ii) \cdot g\equiv0$.

Now, suppose $\partial_{(ij)}f\equiv 0$. If $f\equiv 0$, then it is trivial.
%clearly, there is $b(x_i, x_j)=(0, \ii, -\ii, 0)_{ij}|f$ since a zero signature can be divided by any binary signature. 
Otherwise, $f\not\equiv 0$. Consider the unique prime factorization of $f$. If $x_i$ and $x_j$ appear in one binary signature $b(x_i, x_j)=(0, a, \bar{a}, 0)$, then $a\neq 0$, and $f=b(x_i, x_j)\otimes g$, where $g$ is a constant
or a signature on  variables other than $x_i, x_j$ and $g\not\equiv 0$ due to $f\not\equiv 0$.
Then, we have $\partial_{(ij)}f=(a+\bar{a}) g\equiv0$, which means $a+\bar{a}=0$. That is, $a=\lambda \ii$ for some $\lambda \in \mathbb{R}$. So, we have $b^\ii(x_i, x_j)|f$.

Otherwise, $x_i$ and $x_j$ appear in separate binary signatures $b_1(x_i, x_{i'})=(0, a, \bar{a}, 0)$ and $b_2(x_j, x_{j'})=(0, b, \bar{b}, 0)$ in the unique prime factorization of $f$. That is, $f=b_1(x_i, x_{i'})\otimes b_2(x_j, x_{j'})\otimes g$, where $g$ is a constant or a signature on variables other than $x_i, x_{i'}, x_j, x_{j'}$ and $g\not\equiv 0$. 
%When merging variables $x_i$ and $x_j$ of $f$, we actually connect the 
%Look at $(0, a, \bar{a}, 0)_{ii'}$ and $(0, b, \bar{b}, 0)_{jj'}$. If we merge variables $x_i$, $x_j$, we actually connect two weighted disequalities using disequality. The result is still a  weighted disequality and can be computed by matrix multiplication. It's not a zero binary since it's a product of non-zero matrices. 
%That is, 
%By merging variables $x_i$, $x_j$ of $f$, we have
Then
 $\partial_{(ij)}f=b'(x_{i'}, x_{j'})\otimes  g$ where $b'(x_{i'}, x_{j'})=(0, \bar{a}b, a\bar{b}, 0)\not\equiv 0$. Hence, $\partial_{(ij)}f\not\equiv 0$. A contradiction.
\end{proof}

%!!!may need to define unique factorization. 

\subsection{Tractable signatures and \#P-hardness results}
We give some known signature sets 
that
 define polynomial time computable (tractable) counting problems.
 %to be computable in polynomial time, called tractable.
%These are called tractable signatures.
In this paper, we need two families: affine signatures and
product-type signatures.

\begin{definition}\label{definition-affine}%[Affine signatures]
 A signature $f(x_1, \ldots, x_n)$ of arity $n$
is \emph{affine} if it has the form
 \[
  \lambda \cdot \chi_{A X = 0} \cdot {\mathfrak i} ^{Q(X)},
 \]
 where $\lambda \in \mathbb{C}$,
 $X = (x_1, x_2, \dotsc, x_n, 1)$,
 $A$ is a matrix over $\mathbb{Z}_2$,
 $Q(x_1, x_2, \ldots, x_n)\in \mathbb{Z}_4[x_1, x_2, \ldots, x_n]$
is a quadratic (total degree at most 2) multilinear polynomial
 with the additional requirement that the coefficients of all
 cross terms are even, i.e., $Q$ has the form
 \[Q(x_1, x_2, \ldots, x_n)=a_0+\displaystyle\sum_{k=1}^na_kx_k+\displaystyle\sum_{1\leqslant i<j\leqslant n}2b_{ij}x_ix_j,\]
 and $\chi$ is a 0-1 indicator function
 such that $\chi_{AX = 0}$ is~$1$ iff $A X = 0$.
 We use $\mathscr{A}$ to denote the set of all affine signatures.
\end{definition}
If the support set $\mathscr{S}(f)$ is an affine linear subspace, then we say $f$ has affine support. Clearly, any affine signature has affine support.

\begin{definition}
\label{definition-product-2}%[Product-type signatures]
 A signature on a set of variables $X$
 is of \emph{product type} if it can be expressed as a
product of unary functions,
 binary equality functions $(1,0,0,1)$,
and binary disequality functions $(0,1,1,0)$, each on one or two
variables of $X$.
 We use $\mathscr{P}$ to denote the set of product-type functions.
\end{definition}

Note that the products of 
 unary,
 binary equality,
and binary disequality functions in the definition of
$\mathscr{P}$ can be on overlapping variables.
A simple and  important observation is that $\mathcal{B}\subseteq \mathscr{P}$. 
By  definition, it is easy to verify the following lemmas.
\begin{lemma}\label{product-affine-support}
Any signature of product type has affine support.
\end{lemma}
Please see Definition 2.22 in Section 2.5  of \cite{caifu16} for a proof of Lemma~\ref{product-affine-support}.
More information on $\mathscr{A}$
and $\mathscr{P}$ can be found in~\cite{jcbook}.

\begin{lemma}\label{4-prod}
Let $f$ be an {\rm EO} signature of arity $4$ satisfying {\sc ars}
with support size $4$, say
$f^\alpha$, $f^\beta$, $f^{\overline{\alpha}}$ and $f^{\overline{\beta}}\neq 0$
where $\alpha \neq \beta, \overline{\beta}$.
Then $f\in \mathscr{P}$ if and only if %$\overline{f^\alpha}f^{\overline{\alpha}}=\overline{f^\beta}f^{\overline{\beta}}$. That is,
$|f^\alpha|=|f^\beta|$.
%by {\sc ars}.
\end{lemma}
\begin{proof}
Suppose $f \in \mathscr{P}$.  Then by Lemma~\ref{product-affine-support}
it has affine support. Being an EO signature with support size 4, 
we can show that,
after renaming  its 4 variables
we may assume the support is defined by $(x_1 \neq x_2) \wedge (x_3 \neq x_4)$.
No binary equality is
used in its definition for
being in $\mathscr{P}$,
and  exactly these two binary disequalities are
used. 
%%% this actually requires a %proof: in defining support,no
% "=" can be used:
% if x_1=x_2, then 0011,
%% 1100 are the only possible support
% because it has supp size 4
% some = or \neq must be used
% so may assume x_1 \neq x_2
% then each 01 and 10 must have
% two (all weight 1) extensions
% so it's 0101 0110
% and 1001, 1010.
Then $f$ takes  values
$ac, ad, bc, bd$ on $0101, 0110, 1001, 1010$ for some
$a,b,c,d \neq 0$.
By  {\sc ars}, we have $bd = \overline{ac}$ and $ad = \overline{bc}$.
It follows that $|a|=|b|$. Similarly $|c|=|d|$. 
Therefore all nonzero values of $f$ have the same norm. Hence $|f^\alpha|=|f^\beta|$.

Conversely, suppose $f^\alpha= r e^{i \varphi}$ and
$f^\beta = r e^{i \psi}$, for some $r>0$ and $\varphi$, $\psi$.
By renaming  variables we may assume $\alpha = 0101$,
$\beta= 0110$.
Let $a = r e^{i \frac{\varphi+\psi}{2}}$ and
%$b = r e^{-i\frac{\varphi+ \psi}{2}}$,
% b is a-bar
$c = e^{i \frac{\varphi - \psi}{2}}$.
%$d = e^{-i\frac{\varphi-\psi}{2}}$.
%d is c-bar
Then the unary functions $(a,\bar a)$ on $x_1$ and $(c,\bar c)$ on $x_3$,
times $(x_1 \neq x_2) \wedge (x_3 \neq x_4)$ defines $f \in  \mathscr{P}$.
\end{proof}

The following tractable result is known~\cite{Cai-Lu-Xia-csp}, 
since $(\neq_2) \in \mathscr{A}
\cap \mathscr{P}$.

\begin{theorem}\label{aptractable}
  Let $\mathcal{F}$ be any set of complex-valued signatures on Boolean variables.  If $\mathcal{F} \subseteq \mathscr{A}$
 or   $\mathcal{F} \subseteq \mathscr{P}$,
 then $\Holant(\neq_2 \mid \mathcal{F})$ is tractable.
 %In particular, when $\mathcal{F}$ consisting of {\rm EO} signatures, $\#{\rm EO}(\mathcal{F})$ is tractable.
\end{theorem}
Problems defined by $\mathscr{A}$  and by
$\mathscr{P}$ are tractable~\cite{jcbook}.
%are tractable essentially by Gauss sums
%\cite{jcbook}.
% Problems defined by $\mathscr{P}$ are tractable by a propagation algorithm.
For an EO signature
$f$ of  arity $4$, the complexity classification of \#EO$(f)$ is known \cite{cfx}. (This is the 
six-vertex model for general 4-regular graphs.) For 
the special case when 
the $4$-ary EO signature $f$
further
satisfies {\sc ars}, 
the tractability criterion
in~\cite{cfx} simplifies because in this case $f\in\mathscr{A}$ implies that $f\in\mathscr{P}$.  We restate this
complexity classification of \#EO$(f)$ for our setting.
\begin{theorem}\label{six-vertex}
Let $f$ be an {\rm EO} signature of arity $4$ satisfying 
{\sc ars}. Then {\rm \#EO}$(f)$ is \#P-hard unless 
%$f \in \mathscr{A}$, or  
$f \in \mathscr{P}$.
%$f\in\mathcal{WDEQ}$ or $f$ is a tensor product of two signatures of arity $2$ that are in $\mathcal{WDEQ}$.
\end{theorem}

The complexity classification of \#CSP is also known \cite{Cai-Lu-Xia-csp}.
\begin{theorem}\label{csp-dic} 
Let $\mathcal{F}$ be any set of complex-valued signatures on Boolean variables. Then $\CSP(\mathcal{F})$ is $\SHARPP$-hard unless $\mathcal{F} \subseteq \mathscr{A}$
 or   $\mathcal{F} \subseteq \mathscr{P}$, in which cases the problem is tractable.
\end{theorem}
%On the other hand, problems defined by $\mathcal{WDEQ}^{\otimes}$, which is a  case of product-type signatures \cite{}, are tractable by a propagation algorithm. 

\section{Main Result and Proof Outline}\label{sec-main}

Starting from this section, a signature means a
complex-valued EO signature satisfying the {\sc ars} condition and $\mathcal{F}$ denotes a set of such signatures.
The main result of this paper is
%and a signature set means a set of such signatures.
\begin{theorem}\label{main}
Let $\mathcal{F}$ be a set of {\rm EO} signatures satisfying {\sc ars}. Then,
$\rm{\#EO}( \mathcal{F})$ is \#P-hard unless $\mathcal{F}\subseteq \mathscr{A}$  or $\mathcal{F}\subseteq \mathscr{P}$, in which cases it is tractable.
\end{theorem}

%
%\noindent
Since a complexity dichotomy of ${\rm\#EO}(f)$ for a single signature $f$ of arity $4$ is known (Theorem~\ref{six-vertex}), we wish to leverage this knowledge and realize
arity 4 signatures from a given set of signatures,
to which we can apply the known tractability 
criteria.
%The mating construction plays an important role here, since it naturally realizes a signature of arity $4$. We are able to analyze the signature realized by the mating construction using Cauchy-Schwarz inequality.
We will use the mating construction to 
realize signatures of arity $4$, then apply the Cauchy-Schwarz inequality.
Consider a nonzero signature $f\in \mathcal{F}$. We may assume that $f$ is irreducible. Otherwise we can replace $f$ by its irreducible factors without changing the complexity due to Lemma \ref{lem-decom}.
We have the following lemma.
\begin{lemma}\label{twocase}
Let $f\in \mathcal{F}$ be a nonzero irreducible signature of arity $n\geqslant 4$. Then 
one of the following alternatives holds:
\begin{itemize}
\item {\rm\#EO}$(\mathcal{F})$ is {\rm\#}P-hard,
     \item $\rm{\#EO}(\{\neq_4\}\cup \mathcal{F})
\leqslant_T \rm{\#EO}(\mathcal{F})$, or
    \item there exists a nonzero constant $\lambda$,
    such that for all pairs of distinct indices $\{i, j\}$,
    $M(\mathfrak{m}_{ij}f)= \lambda N^{\otimes 2}$, where $N^{\otimes 2}
    =\left[\begin{smallmatrix}
    0 & 0 & 0 & 1\\
     0 & 0 & 1 & 0\\
      0 & 1 & 0 & 0\\
       1 & 0 & 0 & 0\\
    \end{smallmatrix}\right]$.
    \end{itemize}
     We call the third statement the property of (mutual) orthogonality. 
\end{lemma}

\begin{proof}
 We consider the signature $\mathfrak m_{ij}f$ realized by mating two copies of $f$ for all pairs of distinct indices $\{i,j\}\subseteq[n]$.
 %with variables $x_i$ and $x_j$ as dangling variables. 
 If \#EO$(\mathfrak m_{ij}f)$ is already \#P-hard, then \#EO$(\mathcal{F})$ is also \#P-hard since ${\rm\#EO}(\mathfrak m_{ij}f)\leqslant_T\rm{\#EO}(\mathcal{F})$.
Since we already have a complexity dichotomy for arity 4 signatures, we may assume that $\mathfrak m_{ij}f$ satisfies the tractability condition and that \#EO$(\mathfrak m_{ij}f)$ is computable in polynomial time
for every pair $\{i,j\}$.
%is not \#P-hard for every pair $\{i,j\}$. 
%
%%% I added that, because we %need to % logically say, if %not #P-hard, then tractable.
%for any indices $i$ and $j$.
Recall the form (\ref{m-form}). If there exists some $\{i, j\}$,
such that $\mathfrak m_{ij}f\equiv 0$, then ${\bf f}_{ij}^{00}={\bf f}_{ij}^{01}={\bf f}_{ij}^{10}={\bf f}_{ij}^{11}\equiv 0$, which implies $f\equiv 0$. A contradiction. So we have
$\mathfrak m_{ij}f \not\equiv 0$, for all pairs $\{i, j\}$.
Then by Theorem \ref{six-vertex}, \#EO$(\mathfrak m_{ij}f)$ is tractable if and only if $\mathfrak m_{ij}f \in \mathscr{P}$. By Lemma \ref{product-affine-support}, we know $\mathfrak m_{ij}f$ has affine support, and being nonzero it has support size either $2$ or $4$ (by the form in  (\ref{m-form}), the support size is not 1).
There are two cases depending on the support size of $\mathfrak m_{ij}f$ for all pairs $\{i, j\}$.
%we only need to consider two cases:
%such a signature $\mathfrak m_{ij}f$ is tractable only if $\mathfrak m_{ij}f$ has support size of $2$ or $4$. We analyze the structure of $\mathfrak m_{ij}f$ by the size of its support.

\begin{enumerate}
     \item  There exists some pair $\{i, j\}$ such that
     $\mathfrak m_{ij}f$ has support  size $2$. Then, 
     \begin{itemize}
         \item Either  $M(\mathfrak m_{ij}f)$  has the form $\lambda_{ij}\left[\begin{smallmatrix}
    0 & 0 & 0 & 1\\
     0 & 0 & 0 & 0\\
      0 & 0 & 0 & 0\\
       1 & 0 & 0 & 0\\
    \end{smallmatrix}\right]$ where $\lambda_{ij}=|{\bf f}_{ij}^{00}|^2=|{\bf f}_{ij}^{11}|^2\neq 0$,
    \item {\em or}  $M(\mathfrak m_{ij}f)$ has the form $\lambda_{ij}\left[\begin{smallmatrix}
    0 & 0 & 0 & 0\\
     0 & 0 & 1 & 0\\
      0 & 1 & 0 & 0\\
       0 & 0 & 0 & 0\\
    \end{smallmatrix}\right]$ where $\lambda_{ij}=|{\bf f}_{ij}^{01}|^2=|{\bf f}_{ij}^{10}|^2\neq 0$. 
     \end{itemize} 
    The form that $\langle{\bf f}_{ij}^{01}, {\bf f}_{ij}^{10}\rangle\neq 0$ while $|{\bf f}_{ij}^{01}|^2=0$ cannot occur since $|\langle{\bf f}_{ij}^{01}, {\bf f}_{ij}^{10}\rangle|\leqslant |{\bf f}_{ij}^{01}||{\bf f}_{ij}^{10}|$.
    In both forms, $\neq_4$ is realizable since $\lambda_{ij} \not =0$. %due to $f \not\equiv 0$. 
    Thus,  {\#EO}$(\{\neq_4\}\cup \mathcal{F})\leqslant_T{\rm\#EO}( \mathcal{F})$.
    %That is, {\#EO}$(\{\neq_4\}\cup \mathcal{F})\leqslant_T{\rm\#EO}( \mathcal{F})$.
    
     \item For all pairs $\{i, j\}$, 
     %Otherwise 
     $\mathfrak m_{ij}f$ has support size $4$.
     %for every pair
     %$\{i, j\}$. 
     By Lemma \ref{4-prod},  
     \begin{itemize}
         \item Either $M(\mathfrak m_{ij}f)$  has the form $\lambda_{ij}\left[\begin{smallmatrix}
    0 & 0 & 0 & 1\\
     0 & 0 & 1 & 0\\
      0 & 1 & 0 & 0\\
       1 & 0 & 0 & 0\\
    \end{smallmatrix}\right]$ where $\lambda_{ij}=|{\bf f}_{ij}^{00}|^2=|{\bf f}_{ij}^{11}|^2=|{\bf f}_{ij}^{01}|^2=|{\bf f}_{ij}^{10}|^2\neq 0$,
    \item {\em or} $M(\mathfrak m_{ij}f)$  has the form $\left[\begin{smallmatrix}
    0 & 0 & 0 & 0\\
     0 & \langle{\bf f}_{ij}^{01}, {\bf f}_{ij}^{10}\rangle & |{\bf f}_{ij}^{01}|^2 & 0\\
      0 & |{\bf f}_{ij}^{10}|^2  & \langle{\bf f}_{ij}^{10}, {\bf f}_{ij}^{01}\rangle & 0\\
       0 & 0 & 0 & 0\\
    \end{smallmatrix}\right]$, where %$\langle{\bf f}^{01}, {\bf f}^{10}\rangle \langle{\bf f}^{10}, {\bf f}^{01}\rangle=
    $|\langle{\bf f}_{ij}^{01}, {\bf f}_{ij}^{10}\rangle|^2=|{\bf f}_{ij}^{01}|^2|{\bf f}_{ij}^{10}|^2
    \not = 0$. 
     \end{itemize} 
     Again, the form that $\langle{\bf f}_{ij}^{01}, {\bf f}_{ij}^{10}\rangle\neq 0$ while $|{\bf f}_{ij}^{01}|^2=0$ cannot occur.
    %with $|{\bf f}_{01y}| = |{\bf f}_{10y}|$.
    %In this case, by the form of $M(\mathfrak m_{ij}f)$,
%For every pair $\{i, j\}$, 
In the first form, four vectors form a set of mutually \emph{orthogonal} vectors of nonzero equal norm.  
   \iffalse
   The only other possible alternative form for $M(\mathfrak m_{ij}f)$ to have support  size 4
   is $\left[\begin{smallmatrix}
    0 & 0 & 0 & 0\\
     0 & \langle{\bf f}_{01y}, {\bf f}_{10y}\rangle & |{\bf f}_{01y}|^2 & 0\\
      0 & |{\bf f}_{10y}|^2  & \langle{\bf f}_{10y}, {\bf f}_{01y}\rangle & 0\\
       0 & 0 & 0 & 0\\
    \end{smallmatrix}\right]$.  But in order for ${\rm\#EO}( \mathfrak m_{ij}f)$ to be tractable,
    by \cite{}
    %%% cite dichotomy
    it can be shown that we must have 
    $\langle{\bf f}_{01y}, {\bf f}_{10y}\rangle \langle{\bf f}_{10y}, {\bf f}_{01y}\rangle=|\langle{\bf f}_{01y}, {\bf f}_{10y}\rangle|^2=|{\bf f}_{01y}|^2|{\bf f}_{10y}|^2
    \not = 0$, with $|{\bf f}_{01y}| = |{\bf f}_{10y}|$.
    \fi
    In the second form,
    by Cauchy-Schwarz, it means that ${\bf f}_{ij}^{01}= c {\bf f}_{ij}^{10}$
    for some  $c \in \mathbb{C}$. In addition, we know $|c|=1$ due to $|{\bf f}_{ij}^{01}|=|{\bf f}_{ij}^{10}|$ by {\sc ars}. Since $|{\bf f}_{ij}^{00}|^2=|{\bf f}_{ij}^{11}|^2=0$, we have ${\bf f}_{ij}^{00}={\bf f}_{ij}^{11}=\bf{0}$, the all-zero vector. Thus, $f$ is factorizable as a tensor product
    $f = b(x_i, x_j) \otimes g$, for some $g$ and some binary signature $b(x_i, x_j)=(0, a, \overline{a}, 0)$, a contradiction because $f$ is irreducible. % Indeed by {\sc ars} we may assume the
    %binary signature has the form $(0, a, \overline{a}, 0)$.
    %where by arrow reversal symmetry,
   % we have $\overline{\bar{\alpha} g_{\bar{y}}}  = \alpha g_y$, and hence $g_{\bar{y}} =
    %\overline g_{\bar{y}}$, i.e., $g$ has arrow  reversal symmetry. Also clearly $g$ is 
    %only supported on
    %half weights.

    Thus, in this case, $M(\mathfrak m_{ij}f)=\lambda_{ij}N^{\otimes 2}$ for all pairs $\{i, j\}$. 
    Then, we show that all $\lambda_{ij}$ have the same value.   If we mate further the four dangling variables of $\mathfrak m_{ij}f$, i.e., we connect the  remaining $x_i$ in one copy of $f$ with another copy using $(\neq_2)$,
    and similarly for $x_j$,  which totally mates two copies of $f$, we get a value $4\lambda_{ij}$. 
This value clearly does not depend on the particular indices $\{i, j\}$.
    We denote the common value $\lambda_{ij}$ by $\lambda$.
    Thus, there exists $\lambda\neq 0$ such that  for all pairs  $\{i, j\}$, $M(\mathfrak m_{ij}f)=\lambda N^{\otimes 2}.$ \qedhere
\end{enumerate}
   % Therefore, if there is a pair of indices $\{i, j\}$ such that $\mathfrak m_{ij}f$ has support size $2$, then we have {\#EO}$(\{\neq_4\}\cup \mathcal{F})\leqslant_T{\rm\#EO}( \mathcal{F})$. Otherwise, for every pair of indices $\{i, j\}$, $\mathfrak m_{ij}f$ has support size $4$, and $M(\mathfrak m_{ij}f)$ has the form 
    %$\lambda_{ij}N^{\otimes 2}$, where $\lambda_{ij}
    %\neq 0$. Now, we show that every $\lambda_{ij}$ has the same value.   If we mate further the four dangling variables of $\mathfrak m_{ij}f$, i.e., we connect the  remaining $x_i$ in one copy of $f$ with another copy using $(\neq_2)$,
    %and similarly for $x_j$,  which totally mates two copies of $f$, we get a value $4\lambda_{ij}$. 
%This value clearly does not depend on the particular indices $\{i, j\}$.
 %   We denote the common value $\lambda_{ij}$ by $\lambda$.
\end{proof}
Now we return to the proof of Theorem~\ref{main}.
By Lemma \ref{twocase}, we have two main cases depending on whether 
$\neq_4$ can be realized by  $\mathfrak m_{ij}f$ from $\mathcal{F}$. 
We give a proof outline to show how they will be handled. 
\begin{enumerate}

\item
The signature $\neq_4$
%or a \#P-hard signature 
cannot be realized by  $\mathfrak m_{ij}f$ from %the a signature set 
$\mathcal{F}$. That is, every irreducible signature (or factor of signatures) in $\mathcal{F}$ satisfies the orthogonality property stated in Lemma~\ref{twocase}.

We show that this case happens only if $\mathcal{F}\subseteq \mathcal{B}$ (Theorem \ref{5theo}).
We want to prove this by induction. 
The general strategy is to start with
any signature $f\in \mathcal{F}$ of arity $2n$ that is \emph{not} in $\mathcal{B}$,  we  realize a signature $g$ of arity $2n-2$ that is also \emph{not} in $\mathcal{B}$, i.e. {\rm \#EO}$(
\{g\} \cup \mathcal{F})\leqslant_T{\rm \#EO}(\mathcal{F})$ (Lemma \ref{induction}). 
If we can  reduce the arity down to 4 (this is by a sequence of reductions that is constant in length
independent of the problem instance size of the graph), then we can show it is impossible for such a signature to satisfy the orthogonality. 
Thus, we can use it to realize $\neq_4$ or a \#P-hard signature by Lemma \ref{twocase}. 
However, our induction proof only works when the arity $2n\geqslant 10$ (there is an intrinsic reason for this).
%one will see the reason later). 
Therefore we must establish
the base cases at arity $4, 6$ and $8$.
Fortunately, using the orthogonality of $f$, we can prove our theorem for signatures of arity $4, 6$ and $8$ separately (Lemma \ref{468}).

For the induction proof, we use merging to realize signatures of lower arity. It naturally reduces the arity by two. 
Given a signature $f$ $\notin \mathcal{B}$ of arity $2n\geqslant 10$, 
if $\partial_{(ij)}f \notin \mathcal{B}$ for some $\{i, j\}$, then we are done. 
So we may assume for every $\{i, j\}$, $\partial_{(ij)}f \in \mathcal{B}$.
%Then, assuming $\partial_{(ij)}f \in \mathcal{B}$ for all $\{i, j\}$, %we denote this property by $f\in \int\mathcal{B}$.
We further inquire whether for every $\{i, j\}$, $\partial_{(ij)}f \not \equiv 0$.  If for some
$\{i, j\}$, $\partial_{(ij)}f\equiv 0$, then  it turns out
to be relatively easy to handle (Lemma \ref{0ii0}).  So we may assume
 for every $\{i, j\}$, $\partial_{(ij)}f \not \equiv 0$.
 %we denote this property by $f\in \int\mathcal{B}_{\not \equiv 0}$.
We aim to show that  there is a binary signature $b(x_u, x_v)$ such that $b(x_u, x_v) \mid f$.
If so, the ``quotient'' gives us a signature not in 
$\mathcal{B}$, but of arity $2n-2$, by Lemma \ref{lem-decom}.
In some cases we have to replace $f$ by another $f'$ to accomplish that.

%or we can realize a signature $f'$ from $f$, where $f'\not\in \mathcal{B}$ has the same arity as $f$. If there exists $\{i, j\}$ such that $\partial_{ij}f' \notin \mathcal{B}$, we are done.
%Otherwise, we show that there is a binary signature $b(x_u, x_v)$ such that $b(x_u, x_v) \mid f'$
%$b(x_u, x_v)|f'$ where $f'$ is a signature realizable from $f$
%(Lemma \ref{}). 
%(In the latter case, we just replace $f$ by $f'$).
%Once we have $f=b(x_u, x_v)\otimes g$ where $g$ is signature of arity $2n-2$ and $g\notin \mathcal{B}$ due to $f\notin \mathcal{B}$, we can realize $g$ directly from the factorization of $f$ by Lemma \ref{lem-decom}, and we are done. We will finish the induction proof
%it leads to a contradiction.
%Pick a variable $x_{u'}$ in $g$. By merging variables $x_u$ and $x_{u'}$, we 
%connect $g$ with a binary EO signature. The result is not in $\mathcal{B}$ due to $g\notin \mathcal{B}$. 
%That is, $\partial_{(uu')}f\notin \mathcal{B}$. Contradiction.
%(Or one can apply Lemma \ref{lem-decom} and get $g$ directly from the factorization $f=b(x_u, x_v)\otimes g$.) 
Assuming $\partial_{(ij)}f \in \mathcal{B}$ for all $\{i, j\}$,
%$\partial_{(ij)}f \in \mathcal{B}$ for any $\{i, j\}$, 
we prove there is a $b(x_u, x_v)$ such that $b(x_u, x_v) \mid f$ or $b(x_u, x_v)\mid f'$ in the following steps:
%in the following steps: 
 
\begin{enumerate} 
    \item  If there is a binary signature $b(x_u, x_v)$ such that $b(x_u, x_v) \mid \partial_{(ij)}f$ for every $\{i, j\}$ disjoint with $\{u, v\}$, then $b(x_u, x_v) \mid f$ (Lemma \ref{factor}).
    \item We have assumed  $\partial_{(ij)}f \in \mathcal{B}$ for all $\{i, j\}$. 
    Suppose  there is one $\partial_{(uv)}f\equiv 0$. 
    We show that the binary signature $b^\ii(x_u, x_v)=(0, \ii, -\ii, 0)$ divides $\partial_{(ij)}f$ for every  $\{i, j\}$ disjoint with $\{u, v\}$ (Lemma \ref{onezero}).
    \item Now, we further assume $\partial_{(ij)}f \not\equiv 0$ for all $\{i, j\}$.
    %i.e., $f\in \int\mathcal{B}_{\not \equiv 0}$.
    We want to show that if a binary signature $b(x_u, x_v)$  divides a ``triangle'', i.e. $b(x_u, x_v)\mid \partial_{(rs)}f, \partial_{(st)}f, \partial_{(rt)}f$ (we say $f$ satisfies the $\Delta$-property), it  divides $\partial_{(ij)}f$ for every $\{i, j\}$ disjoint with $\{u, v\}$ (Lemma \ref{triangle}).
    To prove this, we need the following delicate lemma. %This is derived from the following lemma.
 \item If a binary signature $b(x_u, x_v)$ divides ``two pairs'', i.e. $b(x_u, x_v)\mid \partial_{(st)}f, \partial_{(s't')}f$, where $\{s, t\}$ and $\{s', t'\}$ are distinct but not necessarily disjoint, then it divides $\partial_{(ij)}f$ for any $\{i, j\}$ which is disjoint with $\{u, v\}\cup\{s, t\}\cup\{s', t'\}$ that satisfies $\partial_{(st)(ij)}f\not\equiv 0$ and $\partial_{(s't')(ij)}f\not\equiv 0$ (Lemma \ref{twononzero}). 
 %we have $b(x_u, x_v) \mid \partial_{(st)}f$.
%It happens only if $\mathcal{F}$ consists of binary signatures or tensor products of binary signatures. 
\item Finally, we show that either (\rmnum 1) $f$ satisfies the $\Delta$-property, or (\rmnum 2) we can realize a signature $f'$, where $f'\not\in \mathcal{B}$ has the same arity as $f$, such that
either $\partial_{(ij)}f \not \in \mathcal{B}$ for some $\{i, j\}$, 
%$f'\not\in \int \mathcal{B}$ 
or $f'$ satisfies the $\Delta$-property 
%there is a binary signature $b(x_u, x_v)$ and three indices $\{r, s, t\}$ such that $b(x_u, x_v)|\partial_{(rs)}f, \partial_{(st)}f, \partial_{(rt)}f$; or we can realize a signature $f'$ from $f$ such that $f'\not\in \mathcal{B}$ has the same arity as $f$, and $b(x_u, x_v)|\partial_{(rs)}f', \partial_{(st)}f', \partial_{(rt)}f'$ if $\partial_{(ij)}f' \in \mathcal{B}$ for any $\{i, j\}$ 
(Lemma \ref{3indices}).
\end{enumerate}
These steps will accomplish the arity reduction inductive
step.

This case is handled in Section \ref{no-neq_4}. We will see that the unique prime factorization plays an important role in the proof.

\item
Otherwise, we have $\rm{\#EO}(\{\neq_4\}\cup \mathcal{F})
\leqslant_T \rm{\#EO}(\mathcal{F})$.

The signature $\neq_4$ can be used to realize any $(\neq_{2k}) \in \mathcal{DEQ}$ (Lemma \ref{lem-4.1}), and then the  problem $\rm{\#EO}(\mathcal{DEQ}\cup \mathcal{F})$ can be expressed as $\holant{\mathcal{DEQ}}{\mathcal{F}}$ (Lemma \ref{lem-4.2}). 
%That is $\holant{\mathcal{DEQ}}{ \mathcal{F}} \leqslant_T \rm{\#EO}(\{\neq_4\}\cup \mathcal{F}).$
The next idea 
%in this section 
is to simulate \#CSP$(\mathcal{G})\equiv_T\holant{\mathcal{EQ}}{\mathcal{G}}$  using $\holant{\mathcal{DEQ}}{ \mathcal{F}}$ for some $\mathcal{G}$ closely related to $\mathcal{F}$, and we can apply the dichotomy of \#CSP (Theorem \ref{csp-dic}) to get hardness results.
The challenge is to simulate $\mathcal{EQ}$ using $\mathcal{DEQ}$ and  $\mathcal{F}$.
After some reflection one can observe that it is \emph{impossible} to realize $\mathcal{EQ}$ by direct gadget constructions. 
Since signatures in $\mathcal{DEQ}$ and  $\mathcal{F}$ are EO signatures  satisfying   {\sc ars}, %(the supports of which are on half-weight), 
by Lemma \ref{lemma:eo-gates} any gadget realizable from them is also an  EO signature. But clearly, %any signature
any $(=_k) \in\mathcal{EQ}$ is not an EO signature. 
However we  found an alternative  way to  simulate $\mathcal{EQ}$ globally, and this is achieved  depending
crucially on some special properties of $\mathcal{F}$,
as follows:
\begin{enumerate}

\item First, using {\sc ars} we show  that
{\rm \#CSP}$(|\mathcal{F}|^2)\leqslant_T\holant{\mathcal{DEQ}}{ \mathcal{F}}$
(Lemma \ref{lem-square}), where $|\mathcal{F}|^2$ denotes the set of signatures by taking norm squares
of  signatures in $\mathcal{F}$, namely
$|\mathcal{F}|^2 =\{ |f|^2 \mid f\in \mathcal{F}\}$. 
This directly implies that $\holant{\mathcal{DEQ}}{ \mathcal{F}}$ is \#P-hard unless every signature in $\mathcal{F}$ has affine support (Corollary \ref{cor-square}). 

\item Then, consider an EO signature with affine support. We show its support has a special structure called \emph{pairwise opposite} (Definition \ref{def-oppo} and Lemma \ref{lem-opposite}).

\item Finally, given the support of every signature $f\in \mathcal{F}$ is pairwise opposite, 
we show {\rm \#CSP}$(\mathcal{F})\leqslant_T\holant{\mathcal{DEQ}}{ \mathcal{F}}$ (Lemma \ref{lem-self})
by a global simulation, and hence the problem
$\holant{\mathcal{DEQ}}{ \mathcal{F}}$ is \#P-hard unless $\mathcal{F}\subseteq \mathscr{A}$  or $\mathcal{F}\subseteq \mathscr{P}$ (Corollary \ref{cor-self}). 
%First, note that any signature $f\in\mathcal{F}$ satisfies {\sc ars}. That is $f(\alpha)=\overline{f(\bar{\alpha})}$. Suppose it has arity $2n$. Then, consider the function  $f(\bar{x}_1, \bar{x}_2, \ldots, \bar{x}_{2n})$. 
%That is, we replace the input variables by there negations. Then we have $f(\bar{x}_1, \bar{x}_2, \ldots, \bar{x}_{2n})=\overline{f(x_1, x_2, \ldots, x_{2n})}$ due to {\sc ars}.
%Define function $|f|^2(x_1, \ldots, x_{2n})$ to be $|f(x_1, \ldots, x_{2n})|^2$. 
\end{enumerate}
It follows that, in this case, we have $\rm{\#EO}(\{\neq_4\}\cup \mathcal{F})$ is \#P-hard unless $\mathcal{F}\subseteq \mathscr{A}$  or $\mathcal{F}\subseteq \mathscr{P}$ (Theorem \ref{theo-4}). 
This case is handled in Section \ref{neq_4}.
We will introduce the pairwise opposite structure and show the global reductions from \#CSP to \#EO problems.
%This case is handled in Section \ref{neq_4}.
\end{enumerate}

As observed earlier
 $\mathcal{B}\subseteq \mathscr{P}$.
 Thus  \rm{\#EO}$(\mathcal{F})$ is tractable if
 $\mathcal{F}\subseteq \mathcal{B}$,
 by Lemma~\ref{eo=holant} and Theorem~\ref{aptractable}.
  In Section \ref{no-neq_4}, we  show that
if  $\mathcal{F}\not\subseteq \mathcal{B}$
then either \rm{\#EO}$(\mathcal{F})$ 
is \#P-hard, 
or we have
%we can derive $\neq_4$, i.e.,
\rm{\#EO}$(\{\neq_4\} \cup \mathcal{F})\leqslant_T$ 
\rm{\#EO}$(\mathcal{F})$.
In Section~\ref{neq_4}, we show that $\rm{\#EO}(\{\neq_4\}\cup \mathcal{F})$ is \#P-hard unless $\mathcal{F}\subseteq \mathscr{A}$ or $\mathcal{F}\subseteq \mathscr{P}$. This completes the proof of Theorem~\ref{main}.

\section{Interplay of Unique Prime Factorization and $\partial_{(ij)}$
 Operations
%``partial order''
}\label{no-neq_4}
%Recall that $\mathcal{B}$ denotes the set of signatures that are  tensor products of binary EO signatures satisfying the {\sc ars}. 
  In this section, we  show that
if  $\mathcal{F}\not\subseteq \mathcal{B}$
then either \rm{\#EO}$(\mathcal{F})$ is \#P-hard 
or we can realize $\neq_4$, i.e.,
\rm{\#EO}$(\{\neq_4\} \cup \mathcal{F})\leqslant_T$ 
\rm{\#EO}$(\mathcal{F})$, and then the results from Section~\ref{neq_4}
take over.
%unless $\mathcal{F}\subseteq \mathcal{B}$. 
%We only need to consider that $\mathcal{F}$ consisting of a single signature $f$ and 
Suppose $\mathcal{F}\not\subseteq \mathcal{B}$, then
 it contains some signature $f\notin \mathcal{B}$, and 
we prove the statement by induction on the arity of $f$.
The general strategy is that we  start with
any signature $f$ of arity $2n\geqslant 10$ that 
is \emph{not} in $\mathcal{B}$,  and
realize a signature $g$ of arity $2n-2$ that is also \emph{not} in $\mathcal{B}$. 
As stated in Section~\ref{sec:intro} this induction only works for
arity $2n\geqslant 10$.
We prove the base cases of the induction
 separately, when  $f$ has arity 4, 6 or 8.
%When $f$ has arity 4, 6 or 8, they are base cases and will be handled later. 

For the inductive step, we consider $\partial_{(ij)}f$ for all $\{i, j\}$. If there exists $\{i, j\}$ such that $\partial_{(ij)}f \notin \mathcal{B}$, 
then  we can realize
$g=\partial_{(ij)}f$ which has arity $2n-2$, and we are done.
Thus, we assume $\partial_{(ij)}f \in \mathcal{B}$ for all $\{i, j\}$. 
 We denote this property %that $\partial_{(ij)}f \in \mathcal{B}$ for all $\{i, j\}$ 
 by $f\in \int\mathcal{B}$. 
 Under the assumption that $f\in \int\mathcal{B}$, our goal is to show that there is a binary signature $b(x_u, x_v)$ such that
either $b(x_u, x_v) \mid f$ or there exists another  $f' \not\in \mathcal{B}$ 
 realizable from $f$, such that $f'$ has  the same arity as $f$, and 
$b(x_u, x_v) \mid f'$.
  %$b(x_u, x_v) \mid f'$ where  $f'$ is realizable from $f$ and $f'\not\in \mathcal{B}$ has the same arity as $f$.
In the second case we may again assume $f' \in \int\mathcal{B}$,
for otherwise we may take $\partial_{(ij)}f'$  for some $\{i, j\}$.
Now we may replace $f$ by $f'$ in the second case. 
 %(In the latter case, we just replace $f$ by $f'$).
From  the factorization $f = b(x_u, x_v)\otimes g$, 
it follows from the definition of $\mathcal{B}$ that $g \notin \mathcal{B}$
since $f \notin \mathcal{B}$.
From  the factorization of $f$,
we can realize $g$ from $f$ by Lemma \ref{lem-decom}, and we are done. We carry out our induction proof in the next five lemmas.

%In the following, we denote the property that $\partial_{(ij)}f \in \mathcal{B}$ for all $\{i, j\}$ by $f\in \int\mathcal{B}$. 
For convenience, we use the following notations.
\begin{itemize}
    \item  $\mathcal{B} = \{ \mbox{tensor products of one or more binary EO signatures satisfying {\sc ars}}\}$.
%$ denotes the set of signatures that are
  %tensor products of one or more binary EO signatures satisfying  {\sc ars}. 
    \item  $f\in \int\mathcal{B}$ denotes the property that $\partial_{(ij)}f \in \mathcal{B}$ for all $\{i, j\}$.
    \item $f\in \int\mathcal{B}_{\not\equiv 0}$ denotes the property that  $\partial_{(ij)}f\in \mathcal B$ and $\partial_{(ij)}f \not\equiv 0$ for all $\{i, j\}$.
    \item We say $f$ satisfies the 
 $\Delta$-property,  if there exist three distinct
 indices $\{r, s, t\}$
 and a binary signature $b(x_u, x_v)$ such that
$\{u, v\} \cap \{r, s, t\} = \emptyset$, and
 $b(x_u, x_v) \mid \partial_{(rs)}f, \partial_{(st)}f, \partial_{(rt)}f$.
\end{itemize}

\begin{lemma}\label{factor}
Let  $f\in \int \mathcal{B}$ be a signature of arity  $2n\geqslant 6$. If there exists a binary signature $b(x_u, x_v)$ such that $b(x_u, x_v)\mid \partial_{(ij)}f$ for all
$\{i, j\}$ disjoint with $\{u, v\}$, then $b(x_u, x_v)\mid f$.
\end{lemma}

%\begin{lemma}\label{factor}
%Let  $f$ be a signature of arity  $2n\geqslant 6$. If there exists a binary signature $b(x_u, x_v)=(0, a, \bar{a}, 0)$ with $a \neq 0$ such that $b(x_u, x_v)\mid \partial_{(ij)}f$ for all
% $\{i, j\}$ disjoint with $\{u, v\}$, then $b(x_u, x_v)\mid f$.
%\end{lemma}

\begin{proof}
%For simplicity, we use 
Recall that  $f_{uv}^{bc}$ denotes the signature obtained by setting variables $(x_u, x_v)$ of $f$ to $(b, c)\in \{0, 1\}^2$. These are called 
the pinning operations on $\{u ,v\}$. 
%Then, $(\partial_{(ij)}f)^{cc'}$ denotes the signature generated by pinning $(x_1, x_2)$ of $\partial_{(ij)}f$ to be $(c ,c')$.
%While $\partial_{(ij)}(f^{cc'})$ denotes the signature generated by merging variables $x_i$ and $x_j$ of $f^{cc'}$. 
Clearly, for any $\{i, j\}$ disjoint with $\{u ,v\}$,
the pinning operations on $\{u ,v\}$ commute with
the merging operation $\partial_{(ij)}$, and so
we have $(\partial_{(ij)}f)_{uv}^{bc}=\partial_{(ij)}(f_{uv}^{bc})$.
% since the merging operation and the pinning operation on disjoint pairs of variables commute. 

%$f_{(ij)}^{01}$ and $f_{(ij)}^{10}$ denotes the signatures generated by pinning  $(x_u, x_v)=(0,1)$ and $(1,0)$ of $f_{ij}$ respectively.
%(The definition may seem like ambiguous since $f_{(ij)}^{01}$ can also denote the signature generated by merging $x_i$ $x_j$ of the signature $f^{01}$. However, the merging operation and pinning operation are commutable, so they are identically the same signature.)
We may assume the binary signature has the form
$b(x_u, x_v)=(0, a, \bar{a}, 0)$, where $a\neq 0$.
Consider the signature $f'\vcentcolon=\bar{a}f_{uv}^{01}-af_{uv}^{10}$. It is a signature on variables of $f$ other than $x_u$ and $x_v$. For any $\{i, j\}$ disjoint with $\{u, v\}$, by merging variables $x_i$ and $x_j$ of $f'$, 
and recalling that $\partial_{(ij)}$ is a linear operator, we have 
$$\partial_{(ij)}f'=\partial_{(ij)}(\bar{a}f_{uv}^{01}-af_{uv}^{10})=\bar{a}\partial_{(ij)}(f_{uv}^{01})-a\partial_{(ij)}(f_{uv}^{10})=\bar{a}(\partial_{(ij)}f)_{uv}^{01}-a(\partial_{(ij)}f)_{uv}^{10}.$$ 
By assumption,
%$b(x_1, x_2)=(0, a, \bar{a}, 0)_{12}|\partial_{(ij)}f$. That is, 
$\partial_{(ij)}f=b(x_u, x_v)\otimes g$, where $g$ is 
a signature on variables other than $x_u, x_v, x_i, x_j$.
(Since $\partial_{(ij)}f$ has arity at least 4, $g$ is not a constant.)  Then we have $$
\partial_{(ij)}f' =\bar{a}(\partial_{(ij)}f)_{uv}^{01}-a(\partial_{(ij)}f)_{uv}^{10}=\bar{a}(ag)-a(\bar{a}g)\equiv 0.$$
Note that $f'$ is also an EO signature. By Lemma \ref{zero}, we have $f'\equiv 0$, and hence $\bar{a}f_{uv}^{01}\equiv a f_{uv}^{10}$. 
%Clearly, $f^{00}=f^{11}\equiv 0$ since 
Moreover, by the factorization of $\partial_{(ij)}f$,
we have $\partial_{(ij)}(f_{uv}^{00})=(\partial_{(ij)}f)_{uv}^{00}\equiv 0$ and $\partial_{(ij)}(f_{uv}^{11})=(\partial_{(ij)}f)_{uv}^{11}\equiv 0$ for any $\{i, j\}$ disjoint with $\{u, v\}$.
Also, since $2n\geqslant 6$, $f_{uv}^{00}(\alpha)=f_{uv}^{11}(\alpha)=0$ when ${\rm wt}(\alpha)=0$ or $2n-2$.
By Lemma \ref{zero} again, we have $f_{uv}^{00}=f_{uv}^{11}\equiv 0$.
Hence, $f=(f_{uv}^{00}, f_{uv}^{01}, f_{uv}^{10}, f_{uv}^{11})=(0, a, \bar{a}, 0)\otimes (\frac{1}{a}f_{uv}^{01})$, and we have $b(x_u, x_v)\mid f$.
\end{proof}

Notice that for arity $2n \geqslant 6$, if $b(x_u, x_v)\mid f$ and thus
$f = b(x_u, x_v)\otimes g$, then by the definition of $\mathcal{B}$,
from $f
\notin \mathcal{B}$
we obtain $g \notin \mathcal{B}$, which has arity $2n-2$,
 completing the induction step using  Lemma \ref{lem-decom}.
Therefore, to apply Lemma~\ref{factor} we want to show that there is a binary signature $b(x_u, x_v)$
%of the form $(0, a,  \bar{a}, 0)$ with $a \neq 0$  
such that  $b(x_u, x_v)\mid \partial_{(ij)}f$ for every $\{i, j\}$ disjoint with $\{u, v\}$.
We first consider the case that  $\partial_{(uv)}f\equiv 0$
for some $\{u, v\}$.
%In this case, we show the binary $b(x_u, x_v)=(0, \ii, -\ii, 0)_{uv}|\partial_{(ij)}f$ for any $\{i, j\}$ disjoint with $\{u, v\}$.
\begin{lemma}\label{onezero}
Suppose  $f$ has arity $\geqslant 4$ and $f\in \int \mathcal{B}$. If $\partial_{(uv)}f\equiv 0$
for some $\{u, v\}$,  then 
 the binary signature $b^\ii(x_u, x_v)=(0, \ii, -\ii, 0)$ 
satisfies $b^\ii(x_u, x_v)\mid \partial_{(ij)}f$ for all $\{i, j\}$ disjoint with $\{u, v\}$.
\end{lemma}
\begin{proof}
For any $\{i, j\}$ disjoint with $\{u, v\}$, 
the operations $\partial_{(ij)}$ and $\partial_{(uv)}$ commute.
%Hence 
%\[\partial_{(ij)} (\partial_{(uv)}f) =
%\partial_{(uv)} (\partial_{(ij)} f).\] 
%We denote  these iterative
%merging operations by $\partial_{(uv)(ij)}f$ and $\partial_{(ij)(uv)}f$
%respectively. 
Since $\partial_{(uv)}f\equiv 0$, we have
\[\partial_{(uv)} (\partial_{(ij)} f) = 
\partial_{(ij)} (\partial_{(uv)}f) \equiv 0.\] 
Since $\partial_{(ij)}f\in \mathcal{B}$, by Lemma \ref{0ii0}, we have $b^\ii(x_u, x_v)\mid \partial_{(ij)}f$. 
%That is, we have a binary signature $b(x_u, x_v)=(0, \ii, -\ii, 0)_{uv}$ such that $b(x_u, x_v)|\partial_{(ij)}f$ for any $\{i, j\}$ disjoint with $\{u, v\}$. By Lemma \ref{factor}, we have $b(x_u, x_v)=(0, \ii, -\ii, 0)_{uv}|f$.
\end{proof}

In the following, for convenience
we denote 
$\partial_{(ij)} (\partial_{(uv)}f)$ by $\partial_{(ij)(uv)}f$.

Now, we assume $\partial_{(ij)}f\in \mathcal B$ and $\partial_{(ij)}f \not\equiv 0$ for all $\{i, j\}$.
We denote this property by $f\in \int \mathcal{B}_{\not\equiv 0}$.
Each $\partial_{(ij)}f$ has a unique prime factorization. 
 We will show that once we can find some binary signature $b(x_u, x_v)$ 
 %appears in the factorizations of 
 that divides a ``triangle'', 
i.e. $b(x_u, x_v)|\partial_{(rs)}f, \partial_{(st)}f, \partial_{(rt)}f$ for three distinct $\{r,s,t\}$ disjoint with $\{u, v\}$, then it 
 %appears in   all 
 divides $\partial_{(ij)}f$ for all $\{i, j\}$ disjoint with $\{u, v\}$.
 We first consider the case that $b(x_u, x_v)$ divides ``two pairs''.
The statement of the following lemma is  delicate.

\begin{lemma}\label{twononzero}
Let $f$ be a signature of arity $2n \geqslant 8$ and  
 $f\in \int \mathcal{B}_{\not\equiv 0}$.
%$\partial_{(ij)}f \not\equiv 0$.
%such that $\partial_{(ij)}f\in \mathcal{B}$ and $\partial_{(ij)}f\not\equiv 0$ for any $\{i, j\}$. 
Suppose there exist two pairs of indices $\{s, t\}$ and $\{s', t'\}$ that are distinct but not necessarily disjoint, and a binary signature $b(x_u, x_v)$,
where $\{u, v\} \cap (\{s, t\} \cup \{s', t'\}) = \emptyset$,
 such that $b(x_u, x_v)\mid \partial_{(st)}f, \partial_{(s't')}f$.
Then for any $\{i, j\}$ disjoint with $\{u, v\}\cup\{s, t\}\cup\{s', t'\}$, if $\partial_{(st)(ij)}f\not\equiv 0$ and $\partial_{(s't')(ij)}f\not\equiv 0$, 
then $b(x_u, x_v) \mid \partial_{(ij)}f$.
\end{lemma}
\begin{proof}
By  hypothesis $f\in \int \mathcal{B}_{\not\equiv 0}$, so for any
$\{i, j\}$, we have $\partial_{(ij)}f \in \mathcal{B}$ and is 
nonzero, and thus it has a unique factorization with binary prime factors.
Let $\{i, j\}$ be disjoint with $\{u, v\}\cup\{s, t\}\cup\{s', t'\}$.
Suppose it satisfies the condition 
$\partial_{(st)(ij)}f\not\equiv 0$ and $\partial_{(s't')(ij)}f\not\equiv 0$.
We first prove that  $x_u$ and $x_v$ must appear in one
single binary prime factor $b'(x_u, x_v)$ in the factorization 
of $\partial_{(ij)}f$. 
That is, 
\begin{equation}\label{eqn-lm5.4-for-real}
\partial_{(ij)}f=b'(x_u, x_v) \otimes g,
\end{equation}
where $g\not\equiv 0$ is a signature on variables other than $x_u, x_v, x_{i}, x_j$.
For a contradiction,
suppose variables $x_u$ and $x_v$ appear in two
distinct binary prime factors $b_1(x_u, x_{u'})$ and $b_2(x_v, x_{v'})$ in the 
prime factorization of $\partial_{(ij)}f$.
Then, 
\begin{equation}\label{eqn-lm5.4-for-contradiction}
\partial_{(ij)}f=b_1(x_u, x_{u'})\otimes b_2(x_v, x_{v'})\otimes g',
\end{equation}
where $g'\not\equiv 0$ is a signature on variables other than $x_u, x_{u'}, x_v, x_{v'}, x_{i}, x_j$.
By hypothesis, $b(x_u, x_v) \mid \partial_{(st)}f$,
thus $\partial_{(st)}f = b(x_u, x_v) \otimes h$ for some
$h$ on variables other than $x_u,  x_v,  x_{s}, x_t$, which certainly 
include $x_{i}, x_j$.
Thus $\partial_{(ij)(st)}f = b(x_u, x_v) \otimes \partial_{(ij)}h$, and
 we have $b(x_u, x_v)|\partial_{(ij)(st)}f = \partial_{(st)(ij)}f$.
%Since $\partial_{(ij)(st)}f\not\equiv 0$, by Lemma \ref{}, we have $\{i, j\}=\{u', v'\}$.
By hypothesis for this $\{i, j\}$
we have $\partial_{(st)(ij)}f\not\equiv 0$. This implies that after merging
 variables $x_s$ and $x_t$ of $\partial_{(ij)}f$, $x_u$ and $x_v$ form a nonzero binary signature.
%Look at the structure of $\partial_{(st)}f$,
By the form (\ref{eqn-lm5.4-for-contradiction})
of $\partial_{(ij)}f$, the only way  $x_u$ and $x_v$ can form a nonzero binary
signature in $\partial_{(st)(ij)}f$ is that the merge
operation is actually merging $x_{u'}$ and $x_{v'}$.
We conclude that $\{s, t\}=\{u', v'\}$.
We can repeat the same proof replacing $\{s', t'\}$
for $\{s, t\}$,
and since $b(x_u, x_v) \mid \partial_{(s't')(ij)}f$ and $\partial_{(s't')(ij)}f\not\equiv 0$, we have $\{s', t'\}=\{u', v'\}$. Hence, we have $\{s, t\}=\{s', t'\}$. This is a contradiction,
and (\ref{eqn-lm5.4-for-contradiction}) does not hold.

Thus (\ref{eqn-lm5.4-for-real}) holds.
Since $\{s, t\}$ is disjoint with $\{u, v, i, j\}$, by the form
(\ref{eqn-lm5.4-for-real}) of $\partial_{(ij)}f$, when merging variables $x_s$ and $x_t$ of $\partial_{(ij)}f$, we actually merge variables $x_s$ and $x_t$ of $g$ and the binary signature $b'(x_u, x_v)$ is not affected.
%by such an operation. 
Thus,
$$\partial_{(st)(ij)}f=b'(x_u, x_v) \otimes \partial_{(st)}g.$$
That is, $b'(x_u, x_v) \mid \partial_{(st)(ij)}f.$ 
By hypothesis we also have  
 $b(x_u, x_v) \mid \partial_{(st)}f$. By the fact that $\{i, j\}$ is
 disjoint with $\{u, v, s, t\}$, 
%by fact \ref{} 
we have  $b(x_u, x_v)\mid \partial_{(ij)(st)}f=\partial_{(st)(ij)}f.$ 
Thus $b(x_u, x_v)$ and $b'(x_u, x_v)$ both divide 
$\partial_{(st)(ij)}f\not\equiv 0$.
By the unique factorization lemma (Lemma \ref{unique}),
%by fact \ref, 
we have $b(x_u, x_v)=\lambda b'(x_u, x_v)$ for some $\lambda\neq 0$. 
In particular,  by (\ref{eqn-lm5.4-for-real}),
 $b(x_u, x_v) \mid \partial_{(ij)}f$.
\end{proof}

Now we come to the pivotal  ``triangle'' lemma.
Recall that the  $\Delta$-property was defined just before Lemma~\ref{factor}.
Suppose  $f$ satisfies the $\Delta$-property, i.e.,
there is a binary $b(x_u, x_v)$ that divides a ``triangle'', 
 $b(x_u, x_v) \mid \partial_{(rs)}f, \partial_{(st)}f, \partial_{(rt)}f$. 
 A key step in the proof of Lemma~\ref{triangle} is to show that
for any $\{i, j\}$ disjoint with $\{u, v,r, s, t\}$, among
the three iterated ``derivatives''
  $\partial_{(rs)(ij)}f, \partial_{(st)(ij)}f$ and $\partial_{(rt)(ij)}f$, 
at most one of them can be
identically zero. Then Lemma \ref{twononzero} applies.

\begin{lemma}\label{triangle}
%Given an EO signature $f$ of arity $n\geqslant 10$ such that $\partial_{(ij)}f\in \mathcal{B}$ and $\partial_{(ij)}f\not\equiv 0$ for any $\{i, j\}$. 
Let $f\in \int \mathcal{B}_{\not\equiv 0}$ have arity $2n \geqslant 10$.
Suppose $f$ satisfies the $\Delta$-property.
%Suppose there exist three indices $\{r, s, t\}$ and a binary signature $b(x_u, x_v)$ such that $b(x_u, x_v) \mid \partial_{(rs)}f, \partial_{(st)}f, \partial_{(rt)}f$.
Then there is a binary signature $b(x_u, x_v)$ such that 
 for any $\{i, j\}$ disjoint with $\{u, v\}$, we have $b(x_u, x_v)\mid \partial_{(ij)}f$.
\end{lemma}

\begin{proof}
By the  $\Delta$-property,
there is a binary signature $b(x_u, x_v)$ and $\{r, s, t\}$ disjoint with $\{u, v\}$
such that  
$b(x_u, x_v) \mid \partial_{(rs)}f, \partial_{(st)}f, \partial_{(rt)}f$.
%where $\{u, v\} \cap \{r, s, t\} = \emptyset$.
For any $\{i, j\}$ disjoint with $\{u, v\}$,
we first consider the case that $\{i, j\}$ is also disjoint with $\{r, s, t\}$. 
Our idea is to show that among $\partial_{(rs)(ij)}f, \partial_{(st)(ij)}f$ and $\partial_{(rt)(ij)}f$, at most one of them can be a zero signature.
This implies that there are  two among these
 that are not identically zero.
 Then by Lemma \ref{twononzero}, we have $b(x_u, x_v) \mid \partial_{(ij)}f$.

By Lemma \ref{0ii0}, $\partial_{(rs)(ij)}f\equiv 0$ iff
the binary signature $b^{\ii}(x_r, x_s)=(0, \ii, -\ii, 0)$ 
%on variables $x_i$ and $x_j$ 
divides $\partial_{(ij)}f$. 
Similarly, $\partial_{(st)(ij)}f\equiv 0$ iff $b^{\ii}(x_s, x_t)\mid \partial_{(ij)}f$, and $\partial_{(rt)(ij)}f\equiv 0$ iff $b^{\ii}(x_r, x_t) \mid \partial_{(ij)}f$.
By hypothesis, $f\in \int \mathcal{B}_{\not\equiv 0}$, so
 $\partial_{(ij)}f \not\equiv 0$. 
The signature $\partial_{(ij)}f \in \mathcal{B}$ has a
unique prime factorization. 
By Lemma \ref{unique},
since the  three signatures
$b^{\ii}(x_r, x_s), b^{\ii}(x_s, x_t)$ and $b^{\ii}(x_r, x_t)$ are on pairwise overlapping sets of variables,
   at most one of them can be a tensor factor of $\partial_{(ij)}f$.
%Otherwise, in the factorization of $\partial_{(st)}f$, there is a variable appear in two binaries. Contradiction. 
Thus, among $\partial_{(rs)(ij)}f, \partial_{(st)(ij)}f$ and $\partial_{(rt)(ij)}f$, at most one of them can be a zero signature, which implies $b(x_u, x_v)\mid \partial_{(ij)}f$, by Lemma \ref{twononzero},
for all $\{i, j\}$  disjoint with $\{u, v, r, s, t\}$.

Now
suppose $\{i, j\}$  is 
disjoint with $\{u, v\}$,
but not disjoint with $\{r, s, t\}$. In the union $\{i, j\}\cup\{r, s, t\}\cup\{u, v\}$, there are at most $6$
distinct indices. Since the arity of $f$ is at least $10$, 
there are three indices $\{r', s', t'\}$ such that $\{r', s', t'\}$ is disjoint with $\{i, j\}\cup\{r, s, t\}\cup\{u, v\}$. 
Since $\{r', s'\}$ is disjoint with $\{u, v, r, s, t\}$, we can replace $\{i, j\}$ by  $\{r', s'\}$
in the proof above for the case when $\{i,j\}$ is disjoint with
$\{u,v, r,s,t\}$, 
and derive $b(x_u, x_v)\mid \partial_{(r's')}f$. 
By the same reason, we also have 
 $b(x_u, x_v)\mid \partial_{(s't')}f$,
 and  $b(x_u, x_v)\mid \partial_{(r't')}f$.
In other words we found a new ``triangle'', that is,
$f$ satisfies the  $\Delta$-property with the binary signature $b(x_u, x_v)$
and the triple $\{r', s', t'\}$ replacing $\{r, s, t\}$.
  Note that now $\{i, j\}$ is disjoint with $\{r', s', t'\}$. So,
 we can apply the proof  above with $\{r, s, t\}$ now replaced by $\{r', s', t'\}$, and we conclude that $b(x_u, x_v) \mid \partial_{(ij)}f$.
\end{proof}

\begin{remark}
This is the first place we require the arity of $f$ to be at least $10$.
\end{remark}

%We say $f$ satisfies the $\Delta$-property if there is a binary signature $b(x_u, x_v)$ that divides a ``triangle''. %i.e. $\partial_{(rs)}f, \partial_{(st)}f, \partial_{(rt)}f$. 
%Finally, we show that we can find such a binary signature %that divides a ``triangle''.
%we can find a binary signature $b(x_u, x_v)$ and three indices $\{r, s, t\}$ such that $b(x_u, x_v)|\partial_{(rs)}f, \partial_{(st)}f, \partial_{(rt)}f$; or we can realize a signature $f'$ from $f$ such that $f'\not\in \mathcal{B}$ has the same arity as $f$, and $b(x_u, x_v)|\partial_{(rs)}f', \partial_{(st)}f', \partial_{(rt)}f'$ if 
%$f'\in \int \mathcal{B}$. 
%This will finish the induction proof.

\iffalse \begin{lemma}
Given an EO signature $f$ of arity $2n\geqslant 10$, and  $f\notin \mathcal{B}$, then we can realize a signature $g$ of arity $2n-2$ and $g\notin \mathcal{B}$ such that \#EO$(g)\leqslant_T$\#EO$(f)$.
\end{lemma} \fi

We go for the kill in the next lemma.
\begin{lemma}\label{3indices}
Let $f \in \mathcal{F}$ be a signature of arity $2n \geqslant 10$,  $f\notin \mathcal{B}$ and $f\in \int \mathcal{B}_{\not\equiv 0}$. Then
%there exists a binary signature $b(x_u, x_v)$ and three indices $\{r, s, t\}$ such that 
\begin{itemize}
    \item either $f$ satisfies the $\Delta$-property;%$b(x_u, x_v) \mid \partial_{(rs)}f, \partial_{(st)}f, \partial_{(rt)}f$;
    \item  or there is a signature $f' \not\in \mathcal{B}$ 
that has  the same arity as $f$, such that
 ${\rm \#EO}(f' \cup \mathcal{F}) \leqslant_T {\rm \#EO}(\mathcal{F})$,
and the following hold:
either (1) $f'\not\in \int \mathcal{B}$ or  (2) $f'$ satisfies the $\Delta$-property.
    %and $b(x_u, x_v)\mid \partial_{(rs)}f'$, $\partial_{(st)}f',$ $\partial_{(rt)}f'$ if 
%$f'\in \int \mathcal{B}$.
\end{itemize}
\end{lemma}

\begin{proof}
%If there is $\{i, j\}$ such that $\partial_{(ij)}f \notin  \mathcal{B}$, let $g=\partial_{(ij)}f$. Done.
%Otherwise, we have $\partial_{ij}f \in  \mathcal{B}$ for any $\{i, j\}$.
%We show this will lead to a contradiction.

%If there is $\{i, j\}$ such that $\partial_{(ij)}f \equiv 0$, by Lemma \ref{onezero}, we have $b(x_i, x_j)=(0, \ii, -\ii, 0)|f$. That is, $f=b(x_i, x_j)\otimes g$ where $g$ is signature of arity $2n-2$ and $g\notin \mathcal{B}$. Pick a variable $x_{i'}$ in $g$. By merging variables $x_i$ and $x_{i'}$, we 
%connect $g$ with a weighted binary disequality. The result is not in $\mathcal{B}$ due to $g\notin \mathcal{B}$. 
%That is, $\partial_{(ii')}f\notin \mathcal{B}$. Contradiction.
 
%Then, by Lemma \ref{}, we have \#EO$(g)\leqslant_T$\#EO$(f)$. Otherwise, we have  $\partial_{ij}f \not\equiv 0$ for any $\{i, j\}$.

Consider $\partial_{(12)}f$. Since $\partial_{(12)}f\in  \mathcal{B}$ and $\partial_{(12)}f\not\equiv 0$,
without loss of generality, we may assume in the unique prime factorization of $\partial_{(12)}f$, variables $x_3$ and $x_4$ appear in one binary prime factor, $x_5$ and $x_6$ appear in one binary prime factor, and so on. 
That is,
\begin{equation}\label{eqn:lm5-partial12-f}
\partial_{(12)}f= b_1(x_3, x_4)\otimes b_2(x_5, x_6)\otimes b_3(x_7, x_8)\otimes b_4(x_9, x_{10}) \otimes \ldots \otimes b_{n-1}(x_{2n-1}, x_{2n}).
\end{equation}
%Here, we normalize each binary to be $()$
\begin{description}
\item{Case 1.}
For all $1\leqslant k \leqslant n-1$, $b_k(x_{2k+1}, x_{2k+2})\neq $
a scalar multiple of $(0, \ii, -\ii, 0)$. 

Then by Lemma \ref{0ii0},  $\partial_{(34)(12)}f\not\equiv 0$, and clearly, $b_k(x_{2k+1}, x_{2k+2})|\partial_{(34)(12)}f$  for $k\geqslant 2$.
In particular, we have $$b_2(x_5, x_6), b_3(x_7, x_8), b_4(x_9, x_{10})|\partial_{(34)(12)}f,$$ since $f$ has arity at least $10$.
% Since $\partial_{(34)(12)}f=\partial_{(12)(34)}f$, we know 
%$$b_2(x_5, x_6), b_3(x_7, x_8), b_4(x_9, x_{10})|\partial_{(12)(34)}f.$$
%Consider $\partial_{(34)}f$.

Now
consider $\partial_{(34)}f$.
We have $\partial_{(34)}f\in \mathcal{B}$,
$\partial_{(34)}f\not\equiv 0$, 
 and $\partial_{(12)(34)}f =
\partial_{(34)(12)}f\not\equiv 0$.
\begin{itemize}
    \item 
If $x_1$ and $x_2$ appear in one binary prime factor $b'_1(x_1, x_2)$ in the unique prime factorization of $\partial_{(34)}f$, then after merging variables $x_1$ and $x_2$, the binary signature $b'_1(x_1, x_2)$ becomes a nonzero constant, but all other binary prime factors
of $\partial_{(34)}f$ are unchanged and appear in the  prime factorization of $\partial_{(12)(34)}f$.
By commutativity
$\partial_{(12)(34)}f$ = $\partial_{(34)(12)}f$, and
by (\ref{eqn:lm5-partial12-f})  the prime factors of $\partial_{(12)(34)}f$
are
precisely  $b_k(x_{2k+1}, x_{2k+2})$, for  $2 \leqslant k \leqslant n-1$,  we conclude that the unique prime factorization of $\partial_{(34)}f$  has the following form (up to a nonzero constant)
$$\partial_{(34)}f=b'_1(x_1, x_2)\otimes b_2(x_5, x_6)\otimes b_3(x_7, x_8)\otimes b_4(x_9, x_{10}) \otimes \ldots \otimes b_{n-1}(x_{2n-1}, x_{2n}).$$
\item
If $x_1$ and $x_2$ appear in two
distinct binary prime factors $b''_1(x_1, x_i)$ and $b''_2(x_2, x_j)$ in the unique prime factorization of $\partial_{(34)}f$, then after merging variables $x_1$ and $x_2$, from (\ref{eqn:lm5-partial12-f})
we have
\[\partial_{(12)(34)}f = \partial_{(34)(12)}f
=  c \cdot
b_2(x_5, x_6)\otimes b_3(x_7, x_8)\otimes b_4(x_9, x_{10}) \otimes 
\ldots \otimes b_{n-1}(x_{2n-1}, x_{2n})\]
for some nonzero constant $c$.
On the other hand,
from the form of $\partial_{(34)}f$,
the two variables $x_i$ and $x_j$ form a new nonzero binary $b''(x_i, x_j)$.
Thus the pair $\{i, j\}$ is either $\{5,6\}$, or $\{7,8\}$, etc.
and we may assume $(i, j)=(5, 6)$ by renaming the variables.
Thus, we have 
$$\partial_{(34)}f=b''_1(x_1, x_5)\otimes b''_2(x_2, x_6)\otimes b_3(x_7, x_8)\otimes b_4(x_9, x_{10}) \otimes \ldots \otimes b_{n-1}(x_{2n-1}, x_{2n}).$$
\end{itemize}
(In the following proof we can use any $b_j$, for $4 \leqslant j \leqslant n-1$;
for definiteness we  set $j=4$, and since $n \geqslant 5$ this choice $b_4$ is permissible.) 
In both cases above, we have $b_4(x_9, x_{10})|\partial_{(34)}f$, and $\partial_{(78)(34)}f\not\equiv 0$ since $b_3(x_7, x_8)\neq(0, \ii, -\ii, 0)$ by assumption. 
Moreover, note that in both cases, $x_6$ and $x_7$ do not appear
as the two variables of a single binary signature tensor factor of
$\partial_{(34)}f$. The same is true for 
$x_6$ and $x_8$. This implies that $\partial_{(67)(34)}f\not\equiv 0$ and $\partial_{(68)(34)}f\not\equiv 0$.
So we have derived
\[b_4(x_9, x_{10}) \mid \partial_{(34)}f,~~
\partial_{(78)(34)}f\not\equiv 0,~~
\partial_{(67)(34)}f\not\equiv 0,~~
\mbox{and}~~ \partial_{(68)(34)}f\not\equiv 0.\]
Clearly, by
(\ref{eqn:lm5-partial12-f}), we also have 
\[b_4(x_9, x_{10}) \mid \partial_{(12)}f,~~ \partial_{(78)(12)}f\not\equiv 0,~~ 
\partial_{(67)(12)}f\not\equiv 0,~~\mbox{and}~~ \partial_{(68)(12)}f\not\equiv 0.\]
Apply Lemma \ref{twononzero} three times
(with $\{u, v\}=\{9, 10\}, \{s, t\}=\{1, 2\}, \{s', t'\}=\{3, 4\},$ and
taking $\{i, j\}=\{6, 7\}, \{7, 8\}, \{6, 8\}$ separately), we have $$b_4(x_9, x_{10})\mid \partial_{(67)}f, \partial_{(78)}f, \partial_{(68)}f.$$
Thus $f$ satisfies the $\Delta$-property 
($\{u, v\}=\{9, 10\}$ and $\{r, s, t\}=\{6, 7, 8\}$)
and we are done.
%Then, by Corollary \ref{triangle} and Lemma \ref{factor}, we have $b_4(x_9, x_{10})|f$. 
%That is, $f=b_4(x_9, x_{10})\otimes g$ where $g$ is signature of arity $2n-2$ and $g\notin \mathcal{B}$. Done.

\item{Case 2.}
There is a binary signature $b_{k-1}(x_{2k-1}, x_{2k})$ in the factorization of  $\partial_{(12)}f$ such that $b_{k-1}(x_{2k-1}, x_{2k})= $
a scalar multiple of $(0, \ii, -\ii, 0)$. 
Then by Lemma \ref{lem-decom}, we have 
the reduction \#EO$((0, \ii, -\ii, 0), f)\leqslant_T$\#EO$(f)$. Connecting the variable $x_{2k-1}$ of $f$ with $(0, \ii, -\ii, 0)$, we can realize a signature $f'$. 
Consider $\partial_{(12)}f'$. Again the operations commute:
 it is the same as connecting the variable $x_{2k-1}$ of $\partial_{(12)}f$ with $(0, \ii, -\ii, 0)$. Since $\partial_{(12)}f$ is a tensor product of binary signatures, connecting the variable $x_{2k-1}$ of $\partial_{(12)}f$ with $(0, \ii, -\ii, 0)$ is just connecting the variable $x_{2k-1}$ of the binary $b_{k-1}(x_{2k-1}, x_{2k})$ with $(0, \ii, -\ii, 0)$, which gives a binary $(0, 1, 1, 0)$. 
That is, $\partial_{(12)}f'$ is still a tensor product of the same binary signatures as in  $\partial_{(12)}f$ except that $b_{k-1}(x_{2k-1}, x_{2k})=(0, \ii, -\ii, 0)$ is replaced by $b'_{k-1}(x_{2k-1}, x_{2k})=(0, 1, 1, 0)$. Similarly, for any binary signature $b_{\ell-1}(x_{2\ell-1}, x_{2\ell})=(0, \ii, -\ii, 0)$ in $\partial_{(12)}f$, we modify it in this way (together all at once). Thus, we can realize a signature $f'$ by connecting some variables with $(0, \ii, -\ii, 0)$ such that 
%%% JYC "all at once", to avoid hving to talk about if the reduction is
% is in constant # of steps.
$$\partial_{(12)}f'= b'_1(x_3, x_4)\otimes b'_2(x_5, x_6)\otimes b'_3(x_7, x_8)\otimes b'_4(x_9, x_{10}) \otimes \ldots \otimes b'_{n-1}(x_{2n-1}, x_{2n}),$$
where $b'_k(x_{2k+1}, x_{2k+2})\neq
$ a scalar multiple of $(0, \ii, -\ii, 0)$ for any $1 \leqslant k \leqslant n-1$.
Moreover, %by fact \ref{}, 
we know $f' \notin \mathcal{B}$ since $f\notin \mathcal{B}$;
this follows from the closure property of $\mathcal{B}$
under the operation of connecting
a variable by $(0, \ii, -\ii, 0)$  via
$\neq_2$,
and the fact that if we connect three times $(0, \ii, -\ii, 0)$ via
$\neq_2$ in a chain from $f'$, we get $f$ back:  
$\left(N\left[\begin{smallmatrix}
0 & \ii \\
-\ii & 0
\end{smallmatrix}\right] \right)^4 =I$.

If $f'\notin \int \mathcal{B}$, we are done.
Otherwise, $f'\in \int \mathcal{B}$. If there is $\{u, v\}$ such that $\partial_{(uv)}f' \equiv 0$, then by Lemma \ref{onezero}, we have $b^\ii(x_u, x_v)\mid \partial_{(ij)}f'$ for any $\{i, j\}$ disjoint with $\{u, v\}$ where $b^\ii(x_u, x_v)=(0, \ii, -\ii, 0)$. Then clearly  $f'$ satisfies the
$\Delta$-property.
%let $g=\partial_{(ij)}f'$. Note that $g$ is realizable from $f$. Done.
 Otherwise, $f'\in \int \mathcal{B}_{\not\equiv0}$.
 %we have  $\partial_{(ij)}f'\in\mathcal{B}$ for any $\{i, j\}$. 
 As we just proved in Case 1, now replacing $f$ by $f'$, we have 
 $b'_4(x_9, x_{10})\mid \partial_{(67)}f', \partial_{(78)}f', \partial_{(68)}f'.$ This completes the proof.
\end{description}
\vspace{-3ex}
\end{proof}
\vspace{-1ex}
\begin{remark}
This proof also requires the arity of $f$ to be at least $10$.
\end{remark}

\begin{lemma}[Induction]\label{induction}
If $\mathcal{F}$ contains a signature $f\notin \mathcal{B}$ of arity $2n\geqslant 10$, then there is a signature $g\notin \mathcal{B}$ of arity $2n-2$ such that {\rm \#EO}$(\{g\}\cup \mathcal{F})\leqslant_T{\rm \#EO}(\mathcal{F})$.
\end{lemma}

\begin{proof}
% We consider $\partial_{(ij)}f$ for all $\{i, j\}$. 
If $f  \not  \in \int \mathcal{B}$, then 
there exists $\{i, j\}$ such that $\partial_{(ij)}f \notin \mathcal{B}$, 
and  we are done by choosing $g=\partial_{(ij)}f$.
Thus, we assume $f \in \int \mathcal{B}$. %$\partial_{(ij)}f \in \mathcal{B}$ for all $\{i, j\}$.
If   $\partial_{(uv)}f\equiv 0$ for some indices $\{u, v\}$,
 then by Lemmas \ref{onezero} and \ref{factor}, 
the binary signature $b^\ii(x_u, x_v)=(0, \ii, -\ii, 0)$ divides $f$.
That is,  $f=b^\ii(x_u, x_v)\otimes g$ where $g$ is a signature of arity $2n-2$,
 and $g\notin \mathcal{B}$ since $f\notin \mathcal{B}$. 
By Lemma \ref{lem-decom},  we have {\rm \#EO}$(\{g\} \cup \mathcal{F})\leqslant_T{\rm \#EO}(\mathcal{F})$.
So we may assume $f\in \int \mathcal{B}_{\not\equiv 0}$. 
Now we apply Lemma \ref{3indices}.
If the first alternative of  Lemma \ref{3indices}
holds, then $f$ satisfies the $\Delta$-property. Then by Lemmas \ref{triangle} and \ref{factor}, there is  a binary signature $b(x_u, x_v)$ such that $b(x_u, x_v)\mid f$. This divisibility of $f$ produces
a signature not in $\mathcal{B}$ of arity $2n-2$
similar to what we have just proved, and we are done.
If the second alternative of  Lemma \ref{3indices} 
holds, then 
 we have a signature $f'\not\in \mathcal{B}$ having the same arity as $f$. 
We have ${\rm \#EO}(\{f' \}\cup \mathcal{F}) \leqslant_T {\rm \#EO}(\mathcal{F})$.  
%If $f'\not\in \mathcal{B}$, 
If $f' \not\in \int \mathcal{B}$,
then there exists $\{i, j\}$ such that $\partial_{(ij)}f' \notin \mathcal{B}$, and we can take $\partial_{(ij)}f'$
as $g$,
and so we are done.
Otherwise, by the conclusion of Lemma \ref{3indices},
 $f'$ satisfies the $\Delta$-property. Similar to the proof
above for $f$, there is  a binary signature $b(x_u, x_v)$ such that $b(x_u, x_v)\mid f'$. 
This divisibility of $f'$ produces
a signature not in $\mathcal{B}$ of arity $2n-2$.
This completes the inductive step.
%By Lemma \ref{3indices}, 
%By the following Lemmas from \ref{factor}  to \ref{3indices}, 
% there is a binary signature $b(x_u, x_v)$ such that $b(x_u, x_v) \mid f$;
%or we can realize a signature $f'$ from $f$, where $f'\not\in \mathcal{B}$ has the same arity as $f$. If there exists $\{i, j\}$ such that $\partial_{ij}f' \notin \mathcal{B}$, and we are done. Otherwise, there is a binary signature $b(x_u, x_v)$ such that $b(x_u, x_v) \mid f'$. 
%(In the latter case, we just replace $f$ by $f'$ since  {\rm \#EO}$(f'\cup \mathcal{F})\leqslant_T{\rm \#EO}(\mathcal{F})$.)
%Once we have $f=b(x_u, x_v)\otimes g$ where $g$ is signature of arity $2n-2$ and $g\notin \mathcal{B}$ due to $f\notin \mathcal{B}$, we can realize $g$ directly from the factorization of $f$ by Lemma \ref{lem-decom}, and hence  {\rm \#EO}$(g\cup \mathcal{F})\leqslant_T{\rm \#EO}(\mathcal{F})$. 
\end{proof}

Now, we use the orthogonality property
to prove the base cases. 
\begin{lemma}[Base cases]\label{468}
If $\mathcal{F}$ contains a signature $f\notin \mathcal{B}$ of arity $4$, $6$ or $8$, then either
  \rm{\#EO}$( \mathcal{F})$ is \#P-hard or \rm{\#EO}$(\{\neq_4 \} \cup \mathcal{F})\leqslant_T$ \rm{\#EO}$(\mathcal{F})$. 
%unless $f\in \mathcal{B}$.

%A \#P-hard signature or $\neq_4$ is realizable from $f$ unless $f\in\mathcal{B}$. 
\end{lemma}

\begin{proof}
%The idea to prove this lemma is using orthogonality. 
Again by Lemma \ref{lem-decom}, we may assume $f$ is irreducible. Otherwise, we just need to analyze each irreducible factor of $f$. 
More specifically, if $f\notin\mathcal{B}$ and $f$ is reducible, then there exists an irreducible factor $g$ of $f$ such that $g\notin\mathcal{B}$, and $g$ has arity $4$ or $6$. If 
%$\rm{\#EO}$(f)$ is \#P-hard or \rm{\#EO}$(\{\neq_4, f\})\leqslant_T$ $\rm{\#EO}$(f)$. 
we can use $g$ to realize a \#P-hard signature or $\neq_4$, 
we can also use $f$  to do so. %realize a \#P-hard signature or $\neq_4$.
% is realizable from $f$. 
%by the factorization lemma.

By Lemma \ref{twocase}, we may assume that $f$ satisfies the orthogonality. Otherwise, we are done.

\iffalse
Consider the signature $\mathfrak m_{ij}f$ obtained by mating two copies of $f$ with variables $x_i$ and $x_j$ as dangling variables. 
If  $\mathfrak m_{ij}f$ is  a \#P-hard signature or $\neq_4$, then we are done.
So suppose $\mathfrak m_{ij}f$ is neither. Since $f$ is irreducible, by form (\ref{m-form}) and Theorem \ref{six-vertex},  $\mathfrak m_{ij}f$ has the
following form $\lambda_{ij}\left[\begin{smallmatrix}
    0 & 0 & 0 & 1\\
     0 & 0 & 1 & 0\\
      0 & 1 & 0 & 0\\
       1 & 0 & 0 & 0\\
    \end{smallmatrix}\right]$, where $\lambda_{ij}=|{\bf f}^{ 00}_{ij}|^2=|{\bf f}^{01}_{ij}|^2=|{\bf f}_{ij}^{10}|^2=|{\bf f}_{ij}^{11}|^2$. 
    We first show that for every $\mathfrak m_{ij}f$, $\lambda_{ij}$ has the same value. 
    %We use $\lambda_{ij}$ to denote the scalar in $\mathfrak m_{ij}f$. 
    If we mate further the four dangling variables of $\mathfrak m_{ij}f$, which totally mates two copies of $f$, we get a value $4\lambda_{ij}$. 
This value clearly does not depend on the particular indices $\{i, j\}$.
    We denote the value $\lambda_{ij}$ by $\lambda$.
    \fi
    Therefore, we have $$|{\bf f}_{ij}^{ab}|^2=\lambda$$ for any 
    $(a,b) \in \{0, 1\}^2$, and any pair $\{i, j\}$.
    %This is already enough to rule out signatures of arity $4$.
    This readily leads to a contradiction for signatures of arity $4$
    as follows.
    Suppose $f$ is an irreducible signature on four variables $x_1, x_2, x_3, x_4$. Let $(i, j, k, \ell)$ be an arbitrary permutation of $\{1, 2, 3, 4\}$.
    Consider the vector ${\bf f}_{ij}^{00}$. It has only one possible nonzero entry ${f}_{ijk\ell}^{0011}$ since the support of $f$ is on half weight.
   Thus, $$|{\bf f}_{ij}^{00}|^2=|f_{ijk\ell}^{0011}|^2=\lambda$$ for any $(x_i, x_j, x_k, x_\ell)=(0, 0, 1, 1)$.
   Since $(i, j, k, \ell)$ is an arbitrary permutation of $\{1, 2, 3, 4\}$,
   $f_{ijk\ell}^{0011}$ is an arbitrary entry of $f$ at half weight, and since $f$ is nonzero, every
   weight two entry of $f$ has the same nonzero norm 
   $\sqrt{\lambda}$.
   % That is, any nonzero entry of $f$ has the same norm $\sqrt{\lambda}$. 
   However,
    consider the vector ${\bf f}_{ij}^{01}$, it has two  nonzero entries $f_{ijk\ell}^{0101}$ and $f_{ijk\ell}^{0110}$.
    Hence,
    $$\lambda=|{\bf f}_{ij}^{01}|^2=|f_{ijk\ell}^{0101}|^2+|f_{ijk\ell}^{0110}|^2=2\lambda,$$ which means $\lambda=0$. %That is, $f$ is a zero signature. 
    This is a
contradiction. 

Before we go into the technical details of the proof for signatures of arity $6$ and $8$, we first give some intuitions. 
By considering the norm-squares of entries in $f$ as unknowns,
the orthogonality property of $f$ actually gives a linear system. Our proof is to show that when $f$ has small arity $4, 6, 8$, the solution region of such a system only has the trivial zero point. We illustrate this perspective by the arity $4$ case.
Suppose $f$ has arity $4$. It has ${4\choose 2}=6$ possible nonzero entries. These entries satisfy the orthogonality condition. We have 
$$|{\bf f}^{00}_{ij}|^2-\lambda=0, ~~~~|{\bf f}^{01}_{ij}|^2-\lambda=0, ~~~~|{\bf f}^{10}_{ij}|^2-\lambda=0,~~~~ |{\bf f}^{11}_{ij}|^2-\lambda=0$$ for any $\{i, j\}\subseteq\{1, 2, 3, 4\}$.  
There are ${4\choose 2}\times 4=24$ many equations in total. 
If we view these norm-squares of entries $|f^{0011}|^2$, $|f^{0101}|^2$, $|f^{0110}|^2$, $|f^{1001}|^2$, $|f^{1010}|^2$, $|f^{1100}|^2$ (we omit subscripts here) and the value $\lambda$ as variables, those equations are linear equations on these variables. 
By {\sc ars}, we have $|f^{0011}|^2=|f^{1100}|^2$, $|f^{0101}|^2=|f^{1010}|^2$, and $|f^{0110}|^2=|f^{1001}|^2$. So there are only
four 
%$3+1=4$  
variables.
Our idea is to show that
the matrix of this linear system which has $24$ many rows but only $4$ columns has full rank. 
We only need $4$ rows to prove this. In our proof for arity 4, we picked the following 4 rows and showed that the induced linear system has full rank:
$$\begin{bmatrix}
1 & 0& 0 &-1\\
0 & 1 & 0 & -1\\
0 & 0& 1 & -1\\
0 & 1& 1 & -1\\
\end{bmatrix}
\begin{bmatrix}
|f^{0011}|^2\\
|f^{0101}|^2\\
|f^{0110}|^2\\
\lambda
\end{bmatrix}=
\begin{bmatrix}
0\\
0\\
0\\
0\\
\end{bmatrix}.
$$

For the arity $6$ case, we will basically show the same thing (i.e., the linear system  has only the trivial zero solution) with some carefully chosen rows. For arity $8$ case, we will use the fact that the variables take nonnegative values and we show the linear system has no nonnegative solution except the zero solution. 

An intuitive reason why this proof could
 succeed for signatures of small arity is that in these cases, we have more equations than variables, which leads to an over-determined linear system. 
For the general case of arity $n$,
%note that given a signature of arity $n$, 
there are $4{n\choose 2}$ many equations but ${n\choose n/2}/2+1$ many variables. 
Since $4{n\choose 2}\ll{n\choose n/2}/2+1$ when $n$ is large, this
method will not work for large $n$.
% the linear system is under-determined. 
This is why we cannot hope to apply this proof to signatures of large arity. 
    
    %To rule out 
    Now, we give the formal proof for signatures of arity 6 and 8.
    %For signatures of arity $6$, we need to go further. 
In what follows we assume  $f$ has arity $\geqslant 6$.
    Given a vector ${\bf f}_{ij}^{ab}$, we can pick a third variable $x_k$ and separate ${\bf f}_{ij}^{ab}$ into two vectors ${\bf f}_{ijk}^{ab0}$ and ${\bf f}_{ijk}^{ab1}$ according to $x_k=0$ or $1$.
    By setting $(a, b)=(0, 0)$, we have
    \begin{equation}\label{e1}
        |{\bf f}_{ij}^{00}|^2=|{\bf f}_{ijk}^{000}|^2+|{\bf f}_{ijk}^{001}|^2=\lambda.
    \end{equation}
    Similarly, we consider the vector ${\bf f}_{ik}^{00}$ and separate it according to $x_j=0$ or $1$. We have
    \begin{equation}\label{e2}
        |{\bf f}_{ik}^{00}|^2=|{\bf f}_{ijk}^{000}|^2+|{\bf f}_{ijk}^{010}|^2=\lambda.
    \end{equation}
    Comparing equations (\ref{e1}) and (\ref{e2}), we have $|{\bf f}_{ijk}^{001}|^2=|{\bf f}_{ijk}^{010}|^2$. 
    Moreover, by {\sc ars}, we have $|{\bf f}_{ijk}^{010}|^2=|{\bf f}_{ijk}^{101}|^2.$ Thus, we have $|{\bf f}_{ijk}^{001}|^2=|{\bf f}_{ijk}^{101}|^2$.
    Note that the vector ${\bf f}_{jk}^{01}$ can be separated into two vectors ${\bf f}_{ijk}^{001}$ and ${\bf f}_{ijk}^{101}$ 
    %are two vectors divided by 
    according to $x_i=0$ or $1$.
    Therefore, 
    $$|{\bf f}_{jk}^{01}|^2=|{\bf f}_{ijk}^{001}|^2+|{\bf f}_{ijk}^{101}|^2=\lambda.$$ Thus, we have $|{\bf f}_{ijk}^{001}|^2=|{\bf f}_{ijk}^{101}|^2=\lambda/2$. 
     Then, by equation (\ref{e1}), we have $|{\bf f}_{ijk}^{000}|^2 =\lambda/2$, and again by {\sc ars}, we also have $|{\bf f}_{ijk}^{111}|^2 =|{\bf f}_{ijk}^{000}|^2=\lambda/2$.
     Note that the indices $i, j, k$ can be   arbitrary
     three distinct indices, by symmetry we have 
\begin{equation}\label{eqn:arity>=6-basecase}
 |{\bf f}_{ijk}^{abc}|^2=\lambda/2
\end{equation} for  $f$ of arity $\geqslant 6$,
and for all  $(x_i, x_j, x_k)=(a, b, c)\in\{0, 1\}^3.$ 
     %It suffices to rule out signatures of arity $6$. 

      This leads to a contradiction for signatures of arity $6$.
     Suppose $f$ is an irreducible signature on $6$ variables 
     $x_1, x_2, \ldots, x_6$.
     %indexed by $\{i, j, k, i', j', k'\}$.
     Let $(i, j, k, i', j', k')$ be an arbitrary permutation of $\{1, 2, \ldots, 6\}$.
     Note that the vector ${\bf f}_{ijk}^{000}$ has only one possible nonzero entry $f_{ijki'j'k'}^{000111}$. Thus,
by (\ref{eqn:arity>=6-basecase})
 we have $$|{\bf f}_{ijk}^{000}|^2=|f_{ijki'j'k'}^{000111}|^2=\lambda/2$$
     for any $(x_i, x_j, x_k, x_{i'}, x_{j'}, x_{k'})=(0, 0, 0, 1, 1, 1)$.
     That is, any  entry of $f$ at half weight has the same nonzero norm $\sqrt{\lambda/2}$. However,
     the vector ${\bf f}_{ijk}^{001}$ has ${3\choose2} =3$  nonzero entries. But,
     $$\lambda/2=|{\bf f}_{ijk}^{001}|^2=|f_{ijki'j'k'}^{001011}|^2+|f_{ijki'j'k'}^{001101}|^2+|f_{ijki'j'k'}^{001110}|^2=3\lambda/2,$$ which means $\lambda=0$. 
     This is  a
     %That is, $f$ is a zero signature. A 
     contradiction. 
    
    For signatures of arity $8$, we need to go further and use the fact that the norm-square is nonnegative. 
     Given a vector ${\bf f}_{ijk}^{abc}$, we can continue to pick a fourth variable $x_\ell$ and separate ${\bf f}_{ijk}^{abc}$ into two vectors ${\bf f}_{ijk\ell}^{abc0}$ and ${\bf f}_{ijk\ell}^{abc1}$ according to $x_\ell=0$ or $1$.
     By setting $(a, b, c)=(0, 0, 0)$, 
we have from (\ref{eqn:arity>=6-basecase})
     \begin{equation}\label{e3}
     |{\bf f}_{ijk}^{000}|^2=|{\bf f}_{ijk\ell}^{0000}|^2+|{\bf f}_{ijk\ell}^{0001}|^2=\lambda/2.
     \end{equation}
   Similarly, we consider the vector ${\bf f}_{ij\ell}^{001}$ and separate it according to $x_k=0$ or $1$. We have
   \begin{equation}\label{e4}
     |{\bf f}_{ij\ell}^{001}|^2=|{\bf f}_{ijk\ell}^{0001}|^2+|{\bf f}_{ijk\ell}^{0011}|^2=\lambda/2.
     \end{equation}
     Comparing equations (\ref{e3}) and (\ref{e4}), we have $|{\bf f}_{ijk\ell}^{0000}|^2=|{\bf f}_{ijk\ell}^{0011}|^2$. 
     %It suffices to rule out signatures of arity $8$. 
     This leads to a contradiction for signatures of arity $8$. 
     
     Suppose $f$ is an irreducible signature on $8$ variables 
     $x_1, x_2, \ldots, x_8$.
     Let $(i, j, k, \ell, i', j', k', \ell')$ be an arbitrary permutation of $\{1, 2, \ldots, 8\}$.
     %Suppose $f$ is a signature on $8$ variables $x_i, x_j, x_k, x_\ell, x_{i'}, x_{j'}, x_{k'}, x_{\ell'}$. 
     The vector ${\bf f}_{ijk\ell}^{0000}$ has only one possible nonzero entry $f_{ijk\ell i'j'k'\ell'}^{00001111}$. 
     Thus, 
     \begin{equation}\label{e5}
         |{\bf f}_{ijk\ell}^{0000}|^2=|f_{ijk\ell i'j'k'\ell'}^{00001111}|^2.
         \end{equation}
     The vector ${\bf f}_{ijk\ell}^{0011}$ has ${4\choose 2}=6$ possible nonzero entries including $f_{ijk\ell i'j'k'\ell'}^{00110011}$.
     Thus, \begin{equation}\label{e6}
     |{\bf f}_{ijk\ell}^{0011}|^2=|f_{ijk\ell i'j'k'\ell'}^{00110011}|^2+\Delta,
     \end{equation}
     where $\Delta$ denotes the sum of norm-squares of the other $5$ entries in ${\bf f}_{ijk\ell}^{0011}$ and we know $\Delta\geqslant 0$. 
     Since the left-hand sides of equations (\ref{e5}) and (\ref{e6}) are equal, we have 
     \begin{equation}\label{e7}
         |f_{ijk\ell i'j'k'\ell'}^{00001111}|^2=|f_{ijk\ell i'j'k'\ell'}^{00110011}|^2+\Delta.
     \end{equation}
     Similarly, consider vectors ${\bf f}_{iji'j'}^{0000}$ and ${\bf f}_{iji'j'}^{0011}$.
     We have $|{\bf f}_{iji'j'}^{0000}|^2=|{\bf f}_{iji'j'}^{0011}|^2$.
     The vector ${\bf f}_{iji'j'}^{0000}$ has only one possible nonzero entry. Thus, $$|{\bf f}_{iji'j'}^{0000}|^2=|f_{ijk\ell i'j'k'\ell'}^{00110011}|^2.$$
     The vector ${\bf f}_{iji'j'}^{0011}$ has $6$ possible nonzero entries. Thus, 
     $$|{\bf f}_{iji'j'}^{0011}|^2=|f_{ijk\ell i'j'k'\ell'}^{00001111}|^2+\Delta',$$ where $\Delta'$ denotes the sum of norm-squares of the other $5$ entries in ${\bf f}_{iji'j'}^{0011}$ and we know $\Delta' \geqslant 0$. 
     Thus, we have 
     \begin{equation}\label{e8}
         |f_{ijk\ell i'j'k'\ell'}^{00110011}|^2=|f_{ijk\ell i'j'k'\ell'}^{00001111}|^2+\Delta'
     \end{equation}
     Comparing equations (\ref{e7}) and (\ref{e8}), we have $\Delta=-\Delta'$, which means $\Delta=\Delta'=0$ due to $\Delta\geqslant0$ and $\Delta'\geqslant 0$.
     Since $\Delta$ is the sum of $5$ norm-squares, each of which is nonnegative, $\Delta=0$ means each norm-square in the sum $\Delta$ is $0$.
     In particular, $|f_{ijk\ell i'j'k'\ell'}^{00111100}|^2$ is a term in the sum $\Delta$. We have $|f_{ijk\ell i'j'k'\ell'}^{00111100}|^2=0$.
     Since the order of indices is picked arbitrarily, 
%we have  $f_{ijk\ell i'j'k'\ell'}^{00111100}=0$ for any $(i, j, k, \ell, i', j', k', \ell')$ which is a permutation of $\{1, 2, \ldots, 8\}$. That is, 
all entries of $f$ are zero. Thus, $f$ is a zero signature. A contradiction.
\end{proof}

\begin{theorem}\label{5theo}
If $\mathcal{F}\not \subseteq \mathcal{B}$, then either
 \rm{\#EO}$(\mathcal{F})$ is \#P-hard or \rm{\#EO}$(\{\neq_4\} \cup \mathcal{F})\leqslant_T$ \rm{\#EO}$(\mathcal{F})$.
\end{theorem}
\begin{proof}
The base case is 
Lemma~\ref{468} and  the inductive step is 
Lemma~\ref{induction}. Done by induction.
\end{proof}

\section{Reduction from \#CSP to \#EO Problems}\label{neq_4}
In this section, we will show $\rm{\#EO}(\{\neq_4\}\cup \mathcal{F})$ is \#P-hard unless $\mathcal{F}\subseteq \mathscr{A}$ or $\mathcal{F}\subseteq \mathscr{P}$.
The first steps are simple;
the availability  of $\neq_4$ allows us to realize any $(\neq_{2k})$ and therefore all of $\mathcal{DEQ}$.
%The reason why $\neq_4$ plays an important role here is that we can realize any $\neq_{2k}\in \mathcal{DEQ}$ once we have $\neq_4$. 

\begin{lemma}\label{lem-4.1}
$\rm{\#EO}(\mathcal{DEQ}\cup \mathcal{F}) \leqslant_T \rm{\#EO}(\{\neq_4\}\cup \mathcal{F})$.
\end{lemma}

\begin{proof}
Connecting $\neq_{2k}$ ($k\geqslant 2$) and $\neq_4$ using $\neq_2$ we get 
$\neq_{2k+2}$. 
Every occurrence of signatures in $\mathcal{DEQ}$
can be realized by a linear size gadget.
%%% JYC i don't want to say induction. this suggests
% a non constant # of reduction steps
% Then, by induction, 
Then we have $\rm{\#EO}(\mathcal{DEQ}\cup \mathcal{F}) \leqslant_T \rm{\#EO}(\{\neq_4\}\cup \mathcal{F})$.
\end{proof}

Recall that $\rm{\#EO}(\mathcal{DEQ}\cup \mathcal{F})$ is just $\holant{\neq_2}{\mathcal{DEQ}\cup \mathcal{F}}$ 
expressed in the Holant framework. We show that after we
get  $\mathcal{DEQ}$ on the right hand side (RHS)
in the above Holant problem, %$\holant{\neq_2}{\mathcal{DEQ}\cup \mathcal{F}}$, 
we can also realize $\mathcal{DEQ}$ on the left-hand side (LHS).
\begin{lemma}\label{lem-4.2}
$\holant{\mathcal{DEQ}}{ \mathcal{F}}\leqslant_T \holant{\neq_2}{\mathcal{DEQ}\cup \mathcal{F}}$,
which is equivalent to  
$\rm{\#EO}(\mathcal{DEQ}\cup \mathcal{F})$.
%
%= \rm{\#EO}(\mathcal{DEQ}\cup \mathcal{F}).$ 
\end{lemma}
\begin{proof}
%Notice that if we connect each edge of $\neq_2k$ with $\neq_2$
In $\holant{\neq_2}{\mathcal{DEQ}\cup \mathcal{F}}$
we take $2k$ copies of $\neq_2$ on the LHS 
and connect one variable of each copy of $\neq_2$   to all $2k$ variables of
one copy of  $\neq_{2k}$ on the RHS. This gives us the constraint function $\neq_{2k}$ on the LHS.
\end{proof}
Now, consider 
an arbitrary instance of $\holant{\mathcal{DEQ}}{ \mathcal{F}}$;  it is given by a bipartite graph. 
Similar to how we express \#CSP$(\mathcal{F})$ using %the Holant framework
$\holant{\mathcal{EQ}}{\mathcal{F}}$,
in $\holant{\mathcal{DEQ}}{ \mathcal{F}}$ we can view vertices on the LHS (labeled by $(\neq_{2k})\in \mathcal{DEQ}$) as variables, and vertices on the RHS (labeled by $f\in \mathcal{F}$) as constraints.
However, the difference here is that in this setting, both a variable itself and its negation appear as input variables of constraints, and %more importantly, 
they always appear the same number of times. 
More specifically, for a variable vertex $x$ labeled by $\neq_{2k}$, the entire set of $2k$ edges incident to $x$ can be divided into two subsets, each of which consisting of $k$ edges.
In each subset, every edge takes the same value, while two edges in different sets always take opposite values.
Then, we can view the $k$ edges in one subset as the variable $x$ appearing $k$ times, while another $k$ edges in the other subset as its negation $\overline{x}$ appearing $k$ times.

%The main idea in this section is to simulate \#CSP$(\mathcal{F}')$  using $\holant{\mathcal{DEQ}}{ \mathcal{F}}$ for some $\mathcal{F}'$ related to $\mathcal{F}$. 
%Then, we can apply the dichotomy of \#CSP to get hardness results.
%The challenge is to simulate $\mathcal{EQ}$ using $\mathcal{DEQ}$ and  $\mathcal{F}$.
%One may already observe that it is impossible to realize $\mathcal{EQ}$ directly by gadget construction. 
%Since signatures in $\mathcal{DEQ}$ and  $\mathcal{F}$ are EO signatures, the supports of which are on half-weight (), any gadget realizable from them is also an  EO signature. But clearly, any signature $=_k\in\mathcal{EQ}$ is not an EO signature. So, we have to simulate $\mathcal{EQ}$ globally, and this is achievable crucially because of the properties of $\mathcal{F}$.
 Recall that  signatures $f\in\mathcal{F}$ satisfy {\sc ars}. 
 %That is $\overline{f(\alpha)}=f(\bar{\alpha})$. 
 Suppose $f\in\mathcal{F}$ has arity $2n$. Then, consider the function  $f(\overline{x_1}, \overline{x_2}, \ldots, \overline{x_{2n}})$. 
That is, we replace the input variables by their negations. Then we have $f(\overline{x_1}, \overline{x_2}, \ldots, \overline{x_{2n}})=\overline{f(x_1, x_2, \ldots, x_{2n})}$ by   {\sc ars}.
Define the norm square function $|f|^2$, which takes value $|f(x_1, \ldots, x_{2n})|^2$ on input $(x_1, \ldots, x_{2n})$. Then, we have 
$$|f|^2(x_1, \ldots, x_{2n})=f(x_1, \ldots, x_{2n})\overline{f(x_1, \ldots, x_{2n})}=f(x_1, \ldots, x_{2n})f(\overline{x_1}, \ldots, \overline{x_{2n}}),$$
and this gives the following reduction.
\begin{lemma}\label{lem-square}
Let $|\mathcal{F}|^2=\{ |f|^2 \mid f\in \mathcal{F}\}.$ Then {\rm \#CSP}$(|\mathcal{F}|^2)\leqslant_T\holant{\mathcal{DEQ}}{ \mathcal{F}}.$
\end{lemma}

\begin{proof}
Given an instance $I$ of  {\rm \#CSP}$(|\mathcal{F}|^2)$ over $m$ variables. Suppose it
contains $\ell$ occurrences of constraints $|f_i|^2\in |\mathcal{F}|^2$  $(i\in [\ell])$ of arity $2n_i$,
and $f_i$ is applied to the variables
$x_{i_1}, \ldots, x_{i_{2n_i}}$.
%which are not necessarily distinct. 
Then
\begin{equation}\label{equ-square}
{\rm \#CSP}(I)= \sum_{x_1, \ldots, x_m \in \mathbb{Z}_2}\prod^{\ell}_{i=1}|f_i|^2(x_{i_1}, \ldots, x_{i_{2n_i}}) 
= \sum_{x_1, \ldots, x_m \in \mathbb{Z}_2}\prod^{\ell}_{i=1}f_i(x_{i_1}, \ldots, x_{i_{2n_i}})f_i(\overline{x_{i_1}}, \ldots, \overline{x_{i_{2n_i}}}).
\end{equation}
Notice that in the final form of (\ref{equ-square}), for each variable $x\in \{x_1, \ldots, x_m\}$, both itself and its negation appear as input variables to various  constraints $f_i\in \mathcal{F}$. Moreover, 
there is a one-to-one correspondence between
each occurrence of $x$ and that of $\bar{x}$. Thus, $x$ and $\bar{x}$ appear the same number of times. %As we showed above, 
Thus the partition function
${\rm \#CSP}(I)$ for the {\rm \#CSP}$(|\mathcal{F}|^2)$
problem
can be expressed as the partition function
of an instance of $\holant{\mathcal{DEQ}}{ \mathcal{F}}$
of polynomially bounded size.
\end{proof}

Directly by this reduction, we have the following hardness result.
Corollary~\ref{cor-square} follows from
Theorem~\ref{csp-dic}.

\begin{corollary}\label{cor-square}
$\holant{\mathcal{DEQ}}{\mathcal{F}}$ is \#P-hard if there is some $f\in\mathcal{F}$ such that $\mathscr{S}(f)$ is not affine.
\end{corollary}

\begin{proof}
By the definition of $|f|^2$, we know $\mathscr{S}(|f|^2)=\mathscr{S}(f)$. Thus, there is some $|f|^2\in |\mathcal{F}|^2$ such that $\mathscr{S}(|f|^2)$ is not affine. This implies that $|\mathcal{F}|^2 \not\subseteq \mathscr{A}$. Moreover, by Lemma \ref{product-affine-support}, we also have $|\mathcal{F}|^2 \not\subseteq \mathscr{P}$. 
%Note that $|\mathcal{F}|^2 \not\subseteq \mathcal{WDEQ}^{\otimes}$ implies $|\mathcal{F}|^2 \not\subseteq \mathscr{P}$ since $|\mathcal{F}|^2$ is also a set of EO signatures satisfying {\sc ars}. 
By Theorem \ref{csp-dic},  {\rm \#CSP}$(|\mathcal{F}|^2)$ is \#P-hard and hence, by Lemma \ref{lem-square},  $\holant{\mathcal{DEQ}}{\mathcal{F}}$ is \#P-hard.
\end{proof}

Now, we may assume every signature $f\in \mathcal{F}$ has affine support. 
Quite amazingly,
 if an EO signature has affine support, then its support
 must have a special structure,  called \emph{pairwise opposite}.

\begin{definition}[Pairwise opposite]\label{def-oppo}
Let $\mathscr S\subseteq \mathbb{Z}_2^{2n}$  be an affine linear subspace.
%$=\{(x_1, x_2, \ldots, x_{2n})\mid x_1, \ldots, x_{2n} \in \mathbb{Z}_2\}$, 
We say $\mathscr S$ is pairwise opposite if we can partition the $2n$ variables into $n$ pairs such that on $\mathscr S$, two variables of each pair always take opposite values. %in $S$.
If $\mathscr S$ is pairwise opposite, we fix a pairing. Then each pair under this paring is called an opposite pair.
\end{definition}

\begin{example}
Let $\mathscr S=\{(x_1, x_2, \ldots, x_{2n}) \mid x_1, \ldots, x_{2n} \in \mathbb{Z}_2, \overline{x_i}=  x_{n+i} \ ( i\in [n])\} $. Then $\mathscr S$ is pairwise opposite. 
Moreover, any affine linear subspace of $\mathscr{S}$ is pairwise opposite. 

For instance, let $C$ be the Hamming  $(7,4)$-code. We consider its dual Hamming code $C^{\perp}$. $C^{\perp}$ is a linear subspace of $\mathbb{Z}_2^7$ of dimension 3.
Let
\[\mathscr S_C = \{{\bf \alpha} \circ \overline{{\bf \alpha}} \in \mathbb{Z}_2^{14}
 \mid {\bf \alpha} \in C^{\perp}\}.\]
Then $\mathscr S_C$ is pairwise opposite. This $\mathscr S_C$ is introduced in \cite{Cai-Lu-Xia-holant-c} related
to a certain
%newly
%discovered 
tractable family of signatures for a class of Holant problems.
%support of a local affine signature. 
%Given an affine space $S'\subseteq \mathbb{Z}_2^{2n}$, if $S'\subseteq S$, then we know $S'$ is pairwise opposite. 
\end{example}

Note that if an affine linear subspace $\mathscr S\subseteq \mathbb{Z}_2^{2n}$ is pairwise opposite, then $\mathscr S\subseteq \mathscr{H}_{2n}$. Now, we show the other direction is also true.
%using the  M\"{o}bius inversion formula on partially ordered sets.  
This result 
%(Lemma~\ref{lem-opposite}) 
should be of  independent interest.

\begin{lemma}\label{lem-opposite}
Let $\mathscr{S}\subseteq \mathbb{Z}_2^{2n}$ be an affine linear subspace. If $\mathscr S\subseteq \mathscr{H}_{2n}$, 
%(for any $\alpha \in \mathscr{S}$, ${\rm wt}(\alpha)=n$), 
then $\mathscr S$ is pairwise opposite.
%Let  $S\subseteq \mathbb{Z}_2^{2n}$ be an affine space and $S\subseteq \mathscr{H}_n$. That is, for any $\alpha \in S$, ${\rm wt}(\alpha)=n$. 
\end{lemma}

\begin{proof}
%Let $\{x_1, x_2, \ldots, x_{2n}\}$ denote the variables that generate $\mathbb{Z}_2^{2n}$. 
%That is, $\mathbb{Z}_2^{2n}=\{(x_1, x_2, \ldots, x_{2n})| x_i \in \mathbb{Z}_2 (1\leqslant i \leqslant 2n)\}$.
%Since $\mathscr{S}$ is an affine space, 
The lemma is trivially true if $|\mathscr{S}| = 0, 1$.
Suppose $\dim(\mathscr{S})=k \geqslant 1$.
We can pick a set of free variables $F=\{x_1,  \ldots, x_k\}$,  then on $\mathscr{S}$, every variable $x$ is expressible as
%in $D=\{x_{k+1}, \ldots, x_{2n}\}$ 
a unique affine linear combination over $\mathbb{Z}_2$ of these free variables, 
 %for $x\in D$, %$k+1 \leqslant j \leqslant 2n$, 
$x=\lambda_1 x_1+\ldots +\lambda_k x_k +\lambda_{k+1}$, where $\lambda_1, \ldots, \lambda_{k+1} \in \mathbb{Z}_2$.
(If $x$ takes a  constant value
on $\mathscr{S}$, it is still an affine linear combination of these free variables.)
%Moreover, any variable in $F$ is clearly a linear combination of free variables.) 

We separate out all $2n$ variables into two types, those with $\lambda_{k+1} =
0$ (linear form) and those with $\lambda_{k+1} =
1$ (affine, \emph{but not} linear form).
If we set all free variables $x_1, \ldots, x_k$ to 0, we get 
a vector $\alpha \in \mathscr{S}$ with ${\rm wt}(\alpha)=n$.
Each $x$ of the first type contributes zero and each $x$ of the second type
contributes one.
Hence among all $2n$ variables, there are exactly $n$ variables 
of each type, and the chosen free variables are
among the first type. Without loss of generality, we may assume
variables of the first and second type are
 $U=\{x_1, \ldots, x_{n}\}$ 
 and $V=\{x_{n+1}, \ldots, x_{2n}\}$.
 
%that are affine, \emph{but not} linear, combinations of the free variables,
%i.e., $x=\lambda_1 x_1+\ldots \lambda_k x_k +1$, 
%equivalently,  $\bar{x}=\lambda_1 x_1+\ldots \lambda_k x_k$.
%Let $\alpha \in \mathscr{S}$ be the vector determined by setting all free variables to be $0$. 
%Consider any variable $x$. 
%In the vector $\alpha$, we have $x=0$ if $x$ is a linear combination of free variables, and $x=1$ otherwise. 
%Since %$f$ is an EO signature, we know 
%%%such that $x=1$ in $\alpha$. 
%That is, there are exactly $n$ variables 
%%among $\{y_{k+1}, \ldots, y_{2n}\}$ 
%that are affine linear but not linear combinations of free variables. Without loss of generality, we assume these variables are $x_{n+1}, \ldots, x_{2n}$. 
%In other words, for $n+1\leqslant i \leqslant 2n$, $\bar{x_i}$ is a linear combination of free variables.
%Note that the above proof already implies $k\leqslant n$.
%We divide all variables into two sets $U=\{x_1, \ldots, x_{n}\}$ and $V=\{x_{n+1}, \ldots, x_{2n}\}$ such that all variables in $V_0$ are linear combinations of free variables and all variables in $V_1$ are affine linear but not linear combinations of free variables.

%%%%%%%%

%%%%%%%%%%

For any  variable $x=\lambda_1 x_1+\ldots +\lambda_k x_k +\lambda_{k+1}$, with
respect to this unique affine linear expression, 
let $\Lambda(x)=\{i\in [k] \mid \lambda_i=1\}$, the
set of free variables that do appear in the %linear or affine linear 
expression of $x$. 
We have, $$x=\sum_{i\in \Lambda(x)}x_i ~~\text{ if }x\in U, ~~~~~\text{ and }~~~~~ x=1+\sum_{i\in \Lambda(x)}x_i ~~\text{ if }x\in V.$$
Clearly, for $i\in [k]$, $\Lambda(x_i)=\{i\}$. 
For any subset $I \subseteq [k]$, we let
\[U^{\subseteq} (I) = \{x \in U \mid I \subseteq \Lambda(x)\}, ~~~~~\text{ and }~~~~~
U^{=} (I) = \{x \in U \mid \Lambda(x) = I\}.\]
Define $V^{\subseteq} (I)$ and $V^{=}(I)$ analogously, with $V$ in place of $U$.
For any subset $I \subseteq [k]$, 
%let $F_{I}=\{x_i|i\in I\}\subseteq F$ and 
let $\alpha^I \in \mathscr{S}$ be the vector determined by setting free variables $x_i=1$ for  $i\in I$ and  $x_i=0$ for $i\in [k] - I$. 
% We consider the  vector $\alpha^I$.
  Within the $2n$ bit positions in the vector $\alpha^I$,
  for any variable $x \in U$, 
    $$x=1 ~~\text{ if } |I \cap \Lambda(x)| \text{ is odd, ~~~~and~~~~ } x=0 \text{  ~~otherwise.}$$ 
    Symmetrically for any variable $x \in V$,  we have $$x=0 ~~\text{ if } |I \cap \Lambda(x)| \text{ is odd,
    ~~~~and~~~~ } x=1 \text{  ~~otherwise.}$$
    Let $U^{\rm odd}(I)=\{x\in U\mid |I \cap \Lambda(x)| \text{ is odd}\}$ and $V^{\rm odd}(I)=\{x\in V\mid |I \cap \Lambda(x)| \text{ is odd}\}$. %and $\#^{\rm odd}_2(I)=|\{x_i\in V_2\mid |F_I \cap L(x_i)| \text{ is odd}\}$.
    Since $$n = {\rm wt}(\alpha^I)=|U^{\rm odd}(I)|+(n-|V^{\rm odd}(I)|),$$ we have $|U^{\rm odd}(I)|=|V^{\rm odd}(I)|$, for all $I \subseteq [k]$.

\vspace{.1in}
\noindent
{\bf Claim 1.} For all  $I \subseteq [k]$,
\[|U^{\rm odd}(I)|
= \sum_{J \subseteq I: J\neq \emptyset} (-2)^{|J|-1} |U^{\subseteq}(J)|.\]

%%% ok i am done here. :-)
%good enough. Thanks.

To prove this Claim, we count the contributions of every $x \in U$
to both sides of the equation. 
For  $x \in U$,  let $m(x) = |I \cap \Lambda(x)|$. 
%First
%suppose  $m(x)$  is odd.
This $x$ contributes one or zero to the LHS,  according to 
whether $m(x)$ is odd or even  respectively.
On the RHS, its contribution is
\[\sum_{j=1}^{m(x)} (-2)^{j-1} \sum_{J \subseteq I \cap \Lambda(x): |J|=j} 1
= \sum_{j=1}^{m(x)} (-2)^{j-1} {m(x) \choose j}=(-2)^{-1}  \left[ (1-2)^{m(x)} - 1 \right],\]
which is also precisely one or zero according to 
whether $m(x)$ is odd or even  respectively.

%which is also precisely one or zero according to 
%whether $m(x)$ is odd or even  respectively.
%This proves the Claim.

\vspace{.1in}

The same statement is true for $V^{\rm odd}(I)$ replacing $U$ by $V$, with the same proof.

\noindent
{\bf Claim 2.} For all  $I \subseteq [k]$,
\[|V^{\rm odd}(I)|
= \sum_{J \subseteq I: J\neq \emptyset} (-2)^{|J|-1} |V^{\subseteq}(J)|.\]

\vspace{.1in}

We show next that  $|U^{\rm \subseteq}(I)|=|V^{\rm \subseteq}(I)|$
for all $I \subseteq [k]$.
If $I = \emptyset$, then $U^{\subseteq}(I) = U$ and $V^{\subseteq}(I) = V$,
and so they have the same cardinality, both being $n$.
Inductively, for any $I \subseteq [k]$, suppose we already know
that $|U^{\subseteq}(J)| = |V^{\subseteq}(J)|$, for all
proper subsets $J \subsetneq I$, then since
$|U^{\rm odd}(I)|=|V^{\rm odd}(I)|$, by the two Claims
we have $|U^{\subseteq}(I)| = |V^{\subseteq}(I)|$
as well, since the coefficient $(-2)^{|I|-1} \not =0$.
%Thus, $|U^{\rm \subseteq}(I)|=|V^{\rm \subseteq}(I)|$
%for all $I \subseteq [k]$.

Then, by definition
$$|U^\subseteq(I)|=\sum_{I\subseteq J \subseteq [k]}|U^=(J)|.$$
By the M\"{o}bius inversion formula,
%in Lemma~\ref{lem-mobius}, 
we have
$$|U^=(I)|=\sum_{I\subseteq J \subseteq [k]}(-1)^{|J|-|I|}|U^\subseteq(J)|.$$
Indeed, $$\sum_{I\subseteq J \subseteq [k]}(-1)^{|J|-|I|}\sum_{J\subseteq J' \subseteq [k]}|U^=(J')|
= \sum_{I\subseteq J' \subseteq [k]}  \sum_{I\subseteq J \subseteq J'}(-1)^{|J|-|I|}
|U^=(J')|,$$
and for a proper containment $I \subsetneq  J'$ the coefficient of $|U^=(J')|$ is $(1-1)^{|J'|-|I|} =0$,
and it is 1 for $I = J'$.

The same statement is true for $V$.
Thus, we have $|U^=(I)|=|V^=(I)|$  for all $I \subseteq [k]$.

This allows us to set up
a pairing between $U$ and $V$ such that for each
pair of paired
variables $(x, y)\in U \times V$, 
we have $\Lambda(x)=\Lambda(y)$. For any $I \subseteq [k]$, we arbitrarily pick a pairing between  $U^=(I)$ and $V^=(I)$. This is achievable because they have the same cardinality. 
%Moreover, for each pair $(x, y)\in U^=(I) \times V^=(I)$, we have $\Lambda(x)=I=I(y)$.
%Note that in particular, we have 
Since the following decompositions for both $U$ and $V$ are disjoint unions 
$$U=\bigcup\limits_{I \subseteq [k]}U^=(I)
~~~~~~\mbox{and}~~~~~~
V=\bigcup\limits_{I \subseteq [k]}V^=(I),
$$
%where each corresponding pair of sets have the same cardinality,
we get  a global pairing between $U$ and $V$, 
such that for  each pair of paired
variables  $(x, y)\in U \times V$, we have $\Lambda(x)=\Lambda(y)$.
Recall that on $\mathscr{S}$,
since $x \in U$, we have $x=\sum_{i\in \Lambda(x)}x_i$;
meanwhile since $y \in V$ we have
$y=1 + \sum_{i\in \Lambda(y)}x_i$. It follows that
$\overline x = y$ on $\mathscr{S}$.
\end{proof}

Now, we are going to simulate \#CSP$(\mathcal{F})$  using $\holant{\mathcal{DEQ}}{ \mathcal{F}}$ when $\mathcal{F}$ consists of signatures with affine support. Suppose $f(x_1, \ldots, x_{2n})\in\mathcal{F}$ has affine support, by Lemma \ref{lem-opposite}, we know $\mathscr{S}(f)$ is pairwise opposite.
By permuting variables, we may assume for $i\in [n]$, $(x_i, x_{n+i})$ is paired as an opposite pair. %That is, $\bar x_i=x_{n+i}$.
Then, 
%we have $$\sum_{x_1, \ldots, x_{2n}\in \mathbb{Z}_2}f(x_1, \ldots, x_n, x_{n+1}, \ldots x_{2n})=\sum_{x_1, \ldots, x_{n}\in \mathbb{Z}_2}f(x_1, \ldots, x_n, \bar x_{1}, \ldots \bar x_{n}),$$
%and this gives 
we have the following reduction.

\begin{lemma}\label{lem-self}
%Let $|\mathcal{F}|^2=\{ |f|^2 \mid f\in \mathcal{F}\}.$ Then 
Suppose $\mathcal{F}$ is a set of {\rm EO}
signatures. If every signature $f\in \mathcal{F}$ has affine support, then {\rm \#CSP}$(\mathcal{F})\leqslant_T\holant{\mathcal{DEQ}}{ \mathcal{F}}.$
\end{lemma}

\begin{proof}
Given an instance $I$ of  {\rm \#CSP}$(\mathcal{F})$ over $m$ variables $V=\{x_1, \ldots, x_m\}$. Suppose it
contains $\ell$ constraints $f_i$  $(i\in [\ell])$ of arity $2n_i$, and $f_i$ is applied to the variables
$x_{i_1}, \ldots, x_{i_{2n_i}}$.
%which are not necessarily distinct.
We define a graph $G=(V, E)$, where $V$ is the variable set and $(x, y)\in E$ if variables $x, y$ appear as an opposite pair in some $\mathscr{S}(f_i)$.
%there is $f_i$ such that $(x, y)$ is an opposite pair in $\mathscr{S}(f_i)$. %That is, $\bar x = y$ otherwise 
%If there is $x$ such that $x$ is related to itself, then we know \#CSP$(I)\equiv 0$ since  $x=\bar x$ is impossible. In this case, we define $V'=\emptyset$.
%Otherwise, 
Consider all connected components of $G$.
%the reflective and transitive closure of this binary relation. 
We get a partition of $V$. 
Pick a representative variable in each connected component and define $V^{\tt r}$ to be the set of representative variables. 
Without loss of generality, we assume $V^{\tt r}=\{x_1, \ldots, x_{m^{\tt r}}\}$. For each variable $x\in V$, we use $x^{\tt r} \in V^{\tt r}$ to denote its representative variable. 
By the definition of opposite pairs, 
for any assignment with a nonzero contribution,
we have $x=\overline{x^{\tt r}}$ if there is a path of odd length from $x$ to $x^{\tt r}$ and $x=x^{\tt r}$ if there is a path of even length from $x$ to $x^{\tt r}$ (if $x^{\tt r}$ is $x$ itself, we say there is a path of length $0$ from $x^{\tt r}$ to $x$). 
If for some $x$, we have both $x=\overline{x^{\tt r}}$ and $x=x^{\tt r}$, (that is, the connected component containing $x$ is not a bipartite graph), then we know \#CSP$(I)\equiv 0$ since  $x=\bar x$ is impossible.
Otherwise, for each variable $x\in V$ we have either $x=\overline{x^{\tt r}}$ or $x=x^{\tt r}$, but not both. 

Then, for any nonzero term in the sum
$${\rm \#CSP}(I)= \sum_{x_1, \ldots, x_m \in \mathbb{Z}_2}\prod^{\ell}_{i=1}f_i(x_{i_1}, \ldots, x_{i_{2n_i}}),$$
the assignment of all variables in $V$ can be uniquely extended from its restriction on representative variables $V^{\tt r}$.
%This is because 
%We use $x^\prime$ to denote $x^\prime$ 
%For each $f_i$, we also replace the input variable $x_{i_t}$ by its representative variable or the negation of its representative variable. \
Moreover, since $\mathscr{S}(f_i)$ is pairwise opposite, for each opposite pair $(x_{i_s}, x_{i_{n+s}})$, we know exactly one variable is equal to $x^{\tt r}_{i_s}$ while the other one is equal to $\overline {x^{\tt r}_{i_s}}$. 
%% JYC: i disagree this is Without loss of generality,
% there is no uniform way to do that.
%Without loss of generality, 
Thus each pair  $(x_{i_s}, x_{i_{n+s}})$
is either $(x^{\tt r}_{i_s}, \overline {x^{\tt r}_{i_s}})$
or $(\overline {x^{\tt r}_{i_s}}, x^{\tt r}_{i_s})$.
We will write this as
$(\widehat{x^{\tt r}_{i_s}},
\overline{\widehat {x^{\tt r}_{i_s}}})$.
Then, we have
\begin{equation}\label{equ-self}
{\rm \#CSP}(I)= \sum_{x_1, \ldots, x_{m^{\tt r}} \in \mathbb{Z}_2}\prod^{\ell}_{i=1}f_i(x_{i_1}, \ldots,  x_{i_{2n_i}}) 
= \sum_{x_1, \ldots, x_{m^{\tt r}} \in \mathbb{Z}_2}\prod^{\ell}_{i=1}f_i(
\widehat{x^{\tt r}_{i_1}}, \ldots, \widehat{x^{\tt r}_{i_{n_i}}}, \overline{\widehat{x^{\tt r}_{i_1}}}, \ldots, \overline{\widehat{x^{\tt r}_{i_{n_i}}}}).
\end{equation}
%Same as the proof of Lemma \ref{lem-square}, 
The final form of (\ref{equ-self}) is an instance of $\holant{\mathcal{DEQ}}{ \mathcal{F}}$.
\end{proof}

By this reduction, we have the following hardness result.
\begin{corollary}\label{cor-self}
If every signature $f\in \mathcal{F}$ has affine support, then $\holant{\mathcal{DEQ}}{\mathcal{F}}$ is \#P-hard unless $\mathcal{F}\subseteq \mathscr{A}$ or $\mathcal{F}\subseteq \mathscr{P}$.
\end{corollary}

\begin{theorem}\label{theo-4}
$\rm{\#EO}(\{\neq_4\}\cup \mathcal{F})$ is \#P-hard unless $\mathcal{F}\subseteq \mathscr{A}$ or $\mathcal{F}\subseteq \mathscr{P}$.
\end{theorem}

\begin{proof}
It follows from Lemmas \ref{lem-4.1}, \ref{lem-4.2}, Corollaries \ref{cor-square} and \ref{cor-self}.
\end{proof}

Combining Theorems \ref{aptractable},  \ref{5theo} and \ref{theo-4},
we can finish the proof of the main Theorem \ref{main}.

\begin{proof}
  (of Theorem \ref{main})
If $\mathcal{F}\subseteq \mathscr{A}$  or $\mathcal{F}\subseteq \mathscr{P}$, then by Theorems \ref{aptractable}, 
$\rm{\#EO}( \mathcal{F})$ is tractable.
Suppose $\mathcal{F}\not \subseteq \mathscr{A}$  and
$\mathcal{F}\not \subseteq \mathscr{P}$, then certainly
 $\mathcal{F}\not \subseteq \mathcal{B}$ as $\mathcal{B} \subseteq  \mathscr{P}$. Then Theorems  \ref{5theo} and \ref{theo-4} 
 complete the proof.
\end{proof}

\section{\#EO  Encompasses \#CSP}
In this section, we will show that the \#EO  framework encompasses \#CSP.

Let $g$ be an arbitrary signature of arity $n>0$ (with no assumption on satisfying {\sc ars} or to be EO). 
We associate $g$ with an EO signature $\widetilde{g}$ of arity $2n$ in the following way. 
We define 
\[
\widetilde{g}(x_1, \ldots, x_n, x_{n+1}, \ldots, x_{2n})=
\begin{cases}
g(x_1, \ldots x_n) & \text{ if } x_i\neq x_{i+n}  ~~~~ (i \in [n]),\\ ~~~~~~~0 & \text{ otherwise.}
\end{cases}
\]
Clearly, $\widetilde{g}$ is an EO signature. Moreover, its support is pairwise opposite, i.e., $x_i$ and $x_{n+i}$ form an opposite pair. 
We say $x_i$ is in the first half of the inputs  of $\widetilde{g}$, while $x_{n+i}$ is in the second half. 
We define 
$\widetilde{\mathcal{G}}=\{\widetilde{g}\mid g \in \mathcal{G}\}$ for an arbitrary signature set $\mathcal{G}$.
We show that \#CSP is expressible in the \#EO framework by the following theorem.
\begin{theorem}\label{thm:csp-by-eo}
For every signature set $\mathcal{G}$ and
the {\rm EO}  signature set $\widetilde{\mathcal{G}}$
defined above, we have
%has the following property:
%Every instance in one
%framework of   $\CSP(\mathcal{G})$ and ${\rm \#EO}(\widetilde{\mathcal{G}})$
% corresponds
% to 
%an instance of the other,  
%such that they have the same value.
%In particular
\[\CSP(\mathcal{G})\equiv_T{\rm \#EO}(\widetilde{\mathcal{G}}).\]
\end{theorem}
\begin{remark}
Before we give the proof, we remark that
this theorem is not
 merely stating that for an
 arbitrary  $\CSP(\mathcal{G})$ problem,
 one can reduce 
 every instance  of $\CSP(\mathcal{G})$ to  an instance of  a suitable ${\rm \#EO}(\widetilde{\mathcal{G}})$ problem.
 Theorem~\ref{thm:csp-by-eo} is stronger 
 and categorical: For every signature set $\mathcal{G}$
 in the $\CSP$ framework,
 there is a (uniformly constructible)  EO  signature set $\widetilde{\mathcal{G}}$
 such that  $\CSP(\mathcal{G})$ is the same as the \#EO problem
 ${\rm \#EO}(\widetilde{\mathcal{G}})$.  In particular, a complexity dichotomy
 for  \#EO problems would generalize
% automatically give a 
the complexity dichotomy
 for $\CSP$ problems on Boolean variables (which is already known).  This is similar to the fact
 that the  $\CSP$ framework is encompassed by the Holant framework (but not vice versa,
 thus the latter is the more expressive one).
 In contrast, an instance dependent reduction in one direction would not have the
 consequence that a complexity dichotomy for the latter generalizes
 %automatically gives 
 a  complexity dichotomy for the former.
\end{remark}

\begin{proof}
We first show that
every instance of
$\CSP(\mathcal{G})$ is expressible
canonically as an instance  of ${\rm \#EO}(\widetilde{\mathcal{G}}),$ thus, $\CSP(\mathcal{G})\leqslant_T{\rm \#EO}(\widetilde{\mathcal{G}}).$
Let $G=(U, V, E)$ be a bipartite graph representing an instance $I$ of $\CSP(\mathcal{G})$,
where each $u\in U$ is a variable and each $v\in V$ is labeled by a constraint function $g\in \mathcal{G}$.
We will modify the instance $I$ to an instance $\widetilde{I}$ of ${\rm \#EO}(\widetilde{\mathcal{G}})$ that evaluates to
the same  value, as follows.

\begin{enumerate}
    \item 
For every $u\in U$,
 we create  $k = \deg_G(u)$ vertices denoted by $u^i$ 
 ($ 1 \leqslant i \leqslant k$). 
(For example, in Figure \ref{fig:csp-to-eo}, vertices $u_1$, $u_2$ and $u_3$ are decomposed into $3$, $2$ and $1$ vertices respectively.)
Then we connect the $k$ edges originally incident to $u$ to these $k$ new vertices,
so
that each new vertex is incident to exactly one edge.
(To be specific
%and to make it  a one-to-one correspondence
we assume the edges at $u$ in $I$ are ordered
from 1 to $k$, and  we connect the $i$-th edge to $u^i$.
These are edges drawn by solid lines in Figure \ref{fig:csp-to-eo}(b).)
We denote this graph by $G'$.
Each $u^i$ in  $G'$ has degree $1$ and the degree of each $v\in V$ does not change. 
%Suppose $\deg(v)=n$.

\item
For each edge $e^i=(u^i, v)$ in the graph $G'$, we add an edge $\bar e^i=(u^{i+1}, v)$ to $G'$ and we call them a pair.
(Here if $\deg_G(u)=k$ then we use $u^{k+1}$ to denote $u^1$;  we will add a multiple edge if $e^{i+1}=(u^{i+1}, v)$ is already in $G'$.
These edges $\bar e^i$ are  drawn by dashed lines in Figure \ref{fig:csp-to-eo}(b).)
%and when $k=1$ we have $u^2=u^1$.
%We call these two edges $(u^i, v)$ and $(u^{i+1}, v)$ incident to the same $v$ an opposite pair. 
%In Figure \ref{fig:csp-to-eo}(b), for example, we add 
This defines a graph $\widetilde{G}$. 
Each $u^i$ in  $\widetilde{G}$ has degree $2$ and we label it by $\neq_2$. 
If $\deg_G(v)=n$ and is labeled by the constraint function $g \in \mathcal{G}$, 
then $v$  in $\widetilde{G}$ has degree $2n$
and we label it by the corresponding $\widetilde{g}\in \widetilde{\mathcal{G}}$.
We place the signature $\widetilde{g}$ in a way such that 
every  pair of edges $e^i=(u^i, v)$ and $\overline{e}^i=(u^{i+1}, v)$ incident to the same $v$ 
appears as an opposite pair in the inputs of the function $\widetilde{g}$, 
and $e^i$ appears in the first half of the inputs of $\widetilde{g}$ while $\overline{e}^i$ appears in the second half.  
Recall that $\widetilde{g}$ is defined to be pairwise opposite such that
its $j$-th variable in the first half
is paired with its $(n+j)$-th variable
in the second half.
This defines an instance $\widetilde{I}$ of ${\rm \#EO}(\widetilde{\mathcal{G}})$.
	 %%%% \CSP_I was not defined in this paper.
\end{enumerate}
\begin{figure}[!htbp]
\centering
		\includegraphics[height=2.1in]{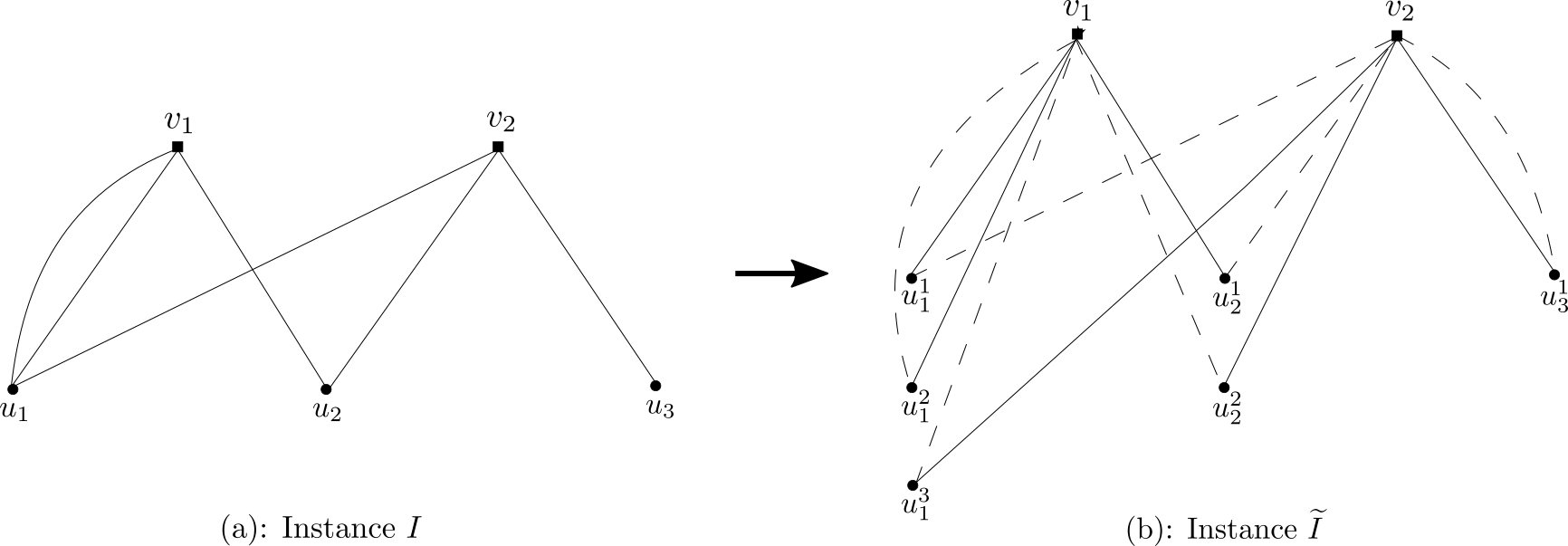}
	\caption{The reduction from \#CSP to \#EO}
	%\caption{Valid configurations of the six-vertex model.}
	\label{fig:csp-to-eo}
	\end{figure}
	
	We  show that  ${\rm \#EO}_{\widetilde{I}}$ has the same
	 value as the instance $I$ for $\CSP(\mathcal{G})$.
Consider each variable $u\in U$. Suppose it has $\deg_G(u)=k$ in the instance $I$. 
It corresponds to $k$ vertices $u^1, \ldots, u^k$ and $2k$ edges $e^1, \overline{e}^1, \ldots, e^i, \overline{e}^i, \ldots, e^k$ and $\overline{e}^k$. 
%$e^1=(u^1, v), \overline{e}^1=(u^2, v), \ldots, e^i=(u^i, v'), \overline{e}^i=(u^, v'), \ldots, e^k=(u^k, v''),  \overline{e}^k=(u^1, v'').$ 
These $2k$ edges form a circuit $C_u$.
For example, in Figure\ref{fig:csp-to-eo}, 
$u_1^1, v_1, u_1^2, v_1, u_1^3, v_2$ back to $u_1^1$ is such a circuit 
where the edges are successively
$e_1^1, \overline{e}_1^1, e_1^2,  \overline{e}_1^2, e_1^3, \overline{e}_1^3$
(edges drawn by solid lines and dashed lines alternate).
Note that, for every pair of edges $e^i$ and $\overline{e}^i$, 
we placed the signature $\widetilde{g}$
such that $e^i$ and $\overline{e}^i$ appear as an opposite pair.
Thus, we may assume $e^i$ and
$\overline{e}^i$ take opposite values
in the evaluation of ${\rm \#EO}_{\widetilde{I}}$.
Also, since each $u^i$ is labeled by $\neq_2$, we may also assume $\overline{e}^i$ and  $e^{i+1}$
take opposite values
in the evaluation. (This is really a consequence of
the definition of \#EO problems.)
%Thus e^i$ and $e^{i+1}$
%take equal  values
%in the evaluation of ${\rm \#EO}_{\widetilde{I}}$.
Thus, for any (possible) nonzero term in the sum ${\rm \#EO}_{\widetilde{I}}$, 
as a consequence of the support of signatures in  $\widetilde{G}$ and $\neq_2$,
we know on each circuit $C_u$  all edges  must take values $(0, 1, 0, 1, \cdots, 0, 1)$ or
 $(1, 0, 1, 0, \cdots, 1, 0)$, i.e., the values of $0, 1$ alternate.
% Thus, 
 % By the placement of $\widetilde{g}$, 
 Therefore, on the circuit $C_u$, we have $e^1, e^2, \ldots, e^k$ all take the same 0-1 value, 
 which corresponds to the 0-1 assignment on the variable $u$ in the \#CSP instance $I$.
 Recall in the definition of $\widetilde{g}$, its value can be determined by the first half of its inputs. 
 By the placement of $\widetilde{g}$, the first half of its inputs are edges in the graph $G'$ (drawn by solid lines).
 Therefore, the contribution of $\widetilde{g}$
 to ${\rm \#EO}_{\widetilde{I}}$ is exactly the same as  the  contribution of  $g$ 
 in the \#CSP instance $I$.
 Thus, these two instances have the same value.
 %${\rm \#EO}_{\widetilde{I}}=\CSP_I$.
 
%The above proof can be reversed. 
For the other direction, we first note that $\CSP( \mathcal{G} \cup \{\neq_2\})\leqslant_T \CSP(\mathcal{G}) $. 
If $\CSP(\mathcal{G})$ is \#P-hard, the reduction holds trivially since every \#CSP problem can be reduced in P-time to a \#P-hard problem.
Otherwise, by Theorem~\ref{csp-dic}, $\CSP(\mathcal{G})$ is tractable and $\mathcal{G}\subseteq \mathscr{A}$ or $ \mathscr{P}$.
Since  $(\neq_2) \in 
\mathscr{A} \cap  \mathscr{P}$,  we have 
$\mathcal{G} \cup \{\neq_2\}\subseteq \mathscr{A}$ or $ \mathscr{P}$.
Thus, $\CSP( \mathcal{G} \cup \{\neq_2\})$ is tractable. Then, again the  reduction holds trivially. 
%a P-time reduction from $\CSP( \mathcal{G} \cup \{\neq_2\})$ to $\CSP(\mathcal{G}) $. 
%by Theorem~\ref{csp-dic}, when $\CSP( \mathcal{G}$ 
%This follows
%from Theorem~\ref{csp-dic}
%because $(\neq_2) \in 
%\mathscr{A} \cap  \mathscr{P}$.
Then, we will show that 
${\rm \#EO}(\widetilde{\mathcal{G}})\leqslant_T
\CSP( \mathcal{G} \cup \{\neq_2\})$.

Consider an arbitrary instance $I'$ of ${\rm \#EO}(\widetilde{\mathcal{G}})$.
Because every signature in $\widetilde{\mathcal{G}}$ has the pairing structure
among its variables, we can decompose the graph of $I'$ into edge disjoint
circuits, by always following the paired variables at each constraint vertex.
For each 
 edge disjoint
circuit, we choose
an arbitrary default orientation.
The circuit visits constraint
vertices in some order according to
the default orientation.
The visit follows successive pairs of edges.
Recall that as a consequence of the support of constraint functions,
on each circuit, all these pairs of edges in the successive order
must take the same ordered pair of values $(x, \bar x)$,
where $x\in \{0, 1\}$.
%$(0, 1, 0, 1, \cdots, 0, 1)$ or
% $(1, 0, 1, 0, \cdots, 1, 0)$.
 %In particular, at all constraint vertices of a circuit, 
Thus, we can define a Boolean variable $x$ from
the edges on each such circuit. 
%that corresponds to an orientation of this circuit.  
From this
 a corresponding instance $I$ for $\CSP( \mathcal{G} \cup \{\neq_2\})$ can be obtained
 that has the same value as $I'$ in ${\rm \#EO}(\widetilde{\mathcal{G}})$.
 
 More specifically,
 suppose $g(x_1, \ldots, x_n) \in \mathcal{G}$ and let $g'(x_1, \ldots, x_n)
 = g(x^{\epsilon_1}_1, \ldots, x^{\epsilon_n}_n)$, where
 each $x^{\epsilon_i}_i$ is either
 $x_i$ or $\overline{x_i}$.
 To discuss the complexity of
 $\CSP( \mathcal{G} \cup \{\neq_2\})$,
 using $(\neq_2)$
 we may assume every function
 obtained by flipping any number of 
 variables in a function $g \in \mathcal{G}$ is also in $\mathcal{G}$.

%Now for each 
 %edge disjoint
%circuit defined above for $I'$, we can first choose
%an arbitrary default orientation.
%The circuit visits constraint
%vertices in some order according to
%the default orientation.
Now, consider the default orientation of each circuit.
At constraint vertices, the default orientation visits successive pairs
of edges corresponding to
 paired inputs of constraint functions,
say, $\{x_j, x_{n+j}\}$.  
If the default orientation 
always visits in the order
$x_j$  followed by $x_{n+j}$,
then this is exactly how
the canonical construction given above
and we can recover an instance
$I$ for  $\CSP(\mathcal{G})$ with the same value.
If at some  constraint  $\widetilde{g}$ of arity $2n$ the default orientation 
happens to visit 
in the order  $x_{n+j}$ followed by 
$x_j$, we can use one copy of
$\neq_2$ to modify the original
function $g$ to get another 
constraint $g'$, so that
the corresponding $\widetilde{g'}$ is just $\widetilde{g}$ with a
flip between its variables
$x_{n+j}$  and
$x_j$.  
Then 
according to the  default orientation 
the visit is
in the order  $x_j$
followed by $x_{n+j}$. 
%
% For the other direction ${\rm \#EO}(\widetilde{\mathcal{G}})\leqslant_T\CSP(\mathcal{G}),$ 
% we will first show that ${\rm \#EO}(\widetilde{\mathcal{G}})\leqslant_T\CSP(\neq_2, \mathcal{G}).$
%Then, by the dichotomy of \#CSP, we know  $\CSP(\mathcal{G})\equiv_T\CSP(\neq_2,\mathcal{G})$.
%This will finish the proof. 
%
%Consider an instance $\widetilde{I}$ of ${\rm \#EO}(\widetilde{\mathcal{G}})$. 
%For any edge $e=(u, v)$, one of its endpoint $u$ is labeled by $\neq_2$ and
%one of its endpoint $v$ is labeled by some $\widetilde{g_v}$.
%Consider any vertex $v$ that is labeled by some $\widetilde{g_v}\in \widetilde{\mathcal{G}}$.
%Suppose $v$ has degree $2n$.
%Since the support of $\widetilde{g_v}$ is pairwise opposite,
%its $2n$ inputs can be uniquely divided into $n$ many opposite pairs.
%Thus, these $2n$ edges incident to $v$ can be uniquely divided into $n$ many pairs.
%
 %By the placement of $\widetilde{g}$, 
%Thus, the values of edges on $C_u$ can be uniquely extended by the value of the first edge $e^1=(u^1, v)$ on $C_u$.
\end{proof}

\section{Holographic Transformation to Real-Valued Holant Problems}\label{apen-holo}
We show that the {\sc ars} condition corresponds to real-valued signatures under a holographic transformation without restricting on EO signatures. First we introduce the idea of holographic transformations. It is convenient to consider bipartite graphs.
For a general graph,
we can always transform it into a bipartite graph while preserving the Holant value,
as follows.
For each edge in the graph,
we replace it by a path of length two.
(This operation is called a \emph{2-stretch} of the graph and yields the edge-vertex incidence graph.)
Each new vertex is assigned the binary \textsc{Equality} signature $(=_2) = (f^{00}, f^{01}, f^{10}, f^{11})
= (1,0,0,1)$.
For an invertible $2$-by-$2$ matrix $T \in {\rm GL}_2({\mathbb{C}})$
 and a signature $f$ of arity $n$, written as
a column vector (contravariant tensor) $f \in \mathbb{C}^{2^n}$, we denote by
$T^{-1}f = (T^{-1})^{\otimes n} f$ the transformed signature.
  For a signature set $\mathcal{F}$,
define $T^{-1} \mathcal{F} = \{T^{-1}f \mid  f \in \mathcal{F}\}$ the set of
transformed signatures.
For signatures written as
 row vectors (covariant tensors) we define
$f T$ and  $\mathcal{F} T$ similarly.
%Whenever we write $T^{-1} f$ or $T^{-1} \mathcal{F}$,
%we view the signatures as column vectors;
%similarly for $f T$ or $\mathcal{F} T$ as row vectors.
In the special case of the  matrix
$Z = \frac{1}{\sqrt{2}} \left[\begin{smallmatrix} 1 & 1 \\ \ii & -\ii \end{smallmatrix}\right]$,
we also define $\widehat{\mathcal{F}} = Z  \mathcal{F}$ and $\widehat{f} = Z  f$.
Note that $Z^{-1} = \frac{1}{\sqrt{2}} \left[\begin{smallmatrix} 1 & -\ii \\ 1 & \ii \end{smallmatrix}\right]$, and $(=_2) Z^{\otimes 2} = (\neq_2)$, i.e.,  $Z$ transforms  
 binary \textsc{Equality}  to
 binary \textsc{Disequality}.
%Since constant factors are immaterial, for convenience we sometime
%drop the factor $\frac{1}{\sqrt{2}}$ when using $Z$. 

Let $T \in {\rm GL}_2({\mathbb{C}})$.
The holographic transformation defined by $T$ is the following operation:
given a signature grid $\Omega = (H, \pi)$ of $\holant{\mathcal{F}}{\mathcal{G}}$,
for the same bipartite graph $H$,
we get a new signature grid $\Omega' = (H, \pi')$ of $\holant{\mathcal{F} T}{T^{-1} \mathcal{G}}$ by replacing each signature in
$\mathcal{F}$ or $\mathcal{G}$ with the corresponding signature in $\mathcal{F} T$ or $T^{-1} \mathcal{G}$.

\begin{theorem}[Valiant's Holant Theorem~\cite{Val08}]
 For any $T \in {\rm GL}_2({\mathbb{C}})$,
  \[\Holant(\Omega; \mathcal{F} \mid \mathcal{G}) = \Holant(\Omega'; \mathcal{F} T \mid T^{-1} \mathcal{G}).\]
\end{theorem}

Therefore,
a holographic transformation does not change the value and thus the complexity of a Holant problem in the bipartite setting.
Note that \#EO$(\mathcal{F})$ can be expressed as $\holant{\neq_2}{\mathcal{F}}$ and $(\neq_2) (Z^{-1})^{\otimes 2}= (=_2)$. Thus, we have
$$\#{\rm EO}(\mathcal{F})\equiv_T \holant{\neq_2}{\mathcal{F}}\equiv_T\holant{(\neq_2)Z^{-1}} {Z\mathcal{F}}\equiv_T \Holant(=_2 \mid \widehat{\mathcal{F}})\equiv_T\Holant(\widehat{\mathcal{F}}).$$
We can determine
   $\widehat{\mathcal{F}}$, for $\mathcal{F}$ consisting of (not necessarily EO) signatures satisfying  {\sc ars}, as follows.
\begin{theorem}
A complex-valued signature $f$ satisfies {\sc ars} if and only if $\widehat{f} = Z f$ is real-valued.
\end{theorem}
\begin{proof}
We first prove that if $f$ satisfies {\sc ars} then $\widehat{f}$ is real.

We have
$2^{n/2}\widehat{f}=
 \left[\begin{smallmatrix} 1 & 1 \\ \ii & -\ii \end{smallmatrix}\right]^{\otimes n}
f$, and thus for  all $(a_1, \ldots, a_n) \in \{0, 1\}^n$,
\[2^{n/2}\widehat{f}^{a_1 \ldots a_n}
=\sum_{(b_1, \ldots, b_n) \in \{0, 1\}^n }
f^{b_1, \ldots, b_n}
\prod_{1 \leqslant j \leqslant n}
\left\{(-1)^{ a_j  b_j}\ii^{a_j}
\right\}.\]
Then,
\begin{eqnarray*}
2^{n/2}\overline{\widehat{f}^{a_1 \ldots a_n}}
&=&
\sum_{(b_1, \ldots, b_n) \in \{0, 1\}^n }
\overline{f^{b_1 \ldots b_n}}
\prod_{1 \leqslant j \leqslant n}\left\{
(-1)^{ a_j  b_j}
(-\ii)^{ a_j } \right\}\\
&=&
\sum_{(c_1, \ldots, c_n) \in \{0, 1\}^n }
f^{c_1 \ldots c_n}
\prod_{1 \leqslant j \leqslant n}\left\{
(-1)^{a_j  (1-c_j)}
(-\ii)^{a_j}\right\}
\\
&=& 2^{n/2}\widehat{f}^{a_1 \ldots a_n}.
\end{eqnarray*}
Hence, $\widehat{f}$ is real.

Now in the opposite direction,
suppose 
$\widehat{f}$ is real.
We have
$2^{n/2}f=
 \left[\begin{smallmatrix} 1 & -\ii \\ 1 & \ii \end{smallmatrix}\right]^{\otimes n}
\widehat{f}$, and thus for  all $(a_1, \ldots, a_n) \in \{0, 1\}^n$,
\[2^{n/2}f^{a_1 \ldots a_n}
=\sum_{(b_1, \ldots, b_n) \in \{0, 1\}^n }
\widehat{f}^{b_1, \ldots, b_n}
\prod_{1 \leqslant j \leqslant n}
\left\{(-1)^{ a_j  b_j}(-\ii)^{b_j}
\right\}.\]
So
\[2^{n/2}f^{\overline{a_1} \ldots \overline{a_n}}
=\sum_{(b_1, \ldots, b_n) \in \{0, 1\}^n }
\widehat{f}^{b_1, \ldots, b_n}
\prod_{1 \leqslant j \leqslant n}
\left\{(-1)^{ (1-a_j)  b_j}(-\ii)^{b_j}
\right\}.\]
Then,
\begin{eqnarray*}
2^{n/2}\overline{f^{a_1 \ldots a_n}}
&=&
\sum_{(b_1, \ldots, b_n) \in \{0, 1\}^n }
\overline{\widehat{f}^{b_1 \ldots b_n}}
\prod_{1 \leqslant j \leqslant n}\left\{
(-1)^{ a_j  b_j}
\ii^{ b_j } \right\}\\
&=&
\sum_{(b_1, \ldots, b_n) \in \{0, 1\}^n }
\widehat{f}^{b_1 \ldots b_n}
\prod_{1 \leqslant j \leqslant n}\left\{
(-1)^{a_j  b_j}
\ii^{b_j}\right\}
\\
&=& 2^{n/2}f^{\overline{a_1} \ldots \overline{a_n}}. 
\end{eqnarray*}
Hence, $f$ satisfies {\sc ars}.
\end{proof}

%%Clearly, real numbers are closed under gadget constructions. Given a gadget realized by signatures with {\sc ars}, we can first transform each signature in the gadget to its corresponding real-valued signature, and then the gadget a also a real-valued signature. Now, we transform it back, and we know the original gadget is a signature with {\sc ars}. 

Due to this  holographic transformation, 
the classification of \#EO problems with {\sc ars} is not only interesting
in its own right, but also it serves as a basic building block in  the
classification program for real-valued Holant problems over signature sets that are not necessarily symmetric.
%For symmetric signatures, a dichotomy for complex-valued Holant problems is known \cite{cai2016complete}. For asymmetric signatures, a dichotomy for non-negative-valued Holant problems is proved recently \cite{beida}.
%A full classification for (even real-valued) asymmetric signatures is open.
%While for symmetric  signatures
%cite{caiguowilliams13},
The signatures %in  \#EO problems 
with {\sc ars} are precisely those signatures that can be transformed into real-valued signatures in the Holant setting.

\section{Conclusion and Outlook}\label{sec-conclusion}
The main technical contribution of this paper is to introduce
unique prime factorization for signatures as a method to prove complexity
dichotomies. 
Combined with merging operations, it is a powerful technique to analyze
the complexity of signatures, and
build inductive proofs. This method should be more widely
applicable in the study of Holant problems. 
The results of this paper can serve as building blocks towards a classification of real-valued Holant problems (with no symmetry assumptions on the signatures). 
Under a suitable holographic transformation, these \#EO problems with {\sc ars} correspond
to precisely a class of real valued Holant problems. 
It also seems that
the EO signatures, which have support on half Hamming weight, are where  some 
of the most intricate cases 
for the final classification  lie.
%The techniques of this paper can handle all, not
%necessarily symmetric, local constraint functions on half weight support with {\sc ars}.
%where
%it seems that some most intricate cases of a full real valued Holant dichotomy lie. 

We give the following Figure~\ref{fig:structure} as a %blueprint 
partial map for the complexity classification program for Holant problems on the Boolean domain.
We indicate how the framework of \#EO problems considered in this paper fits in this program.

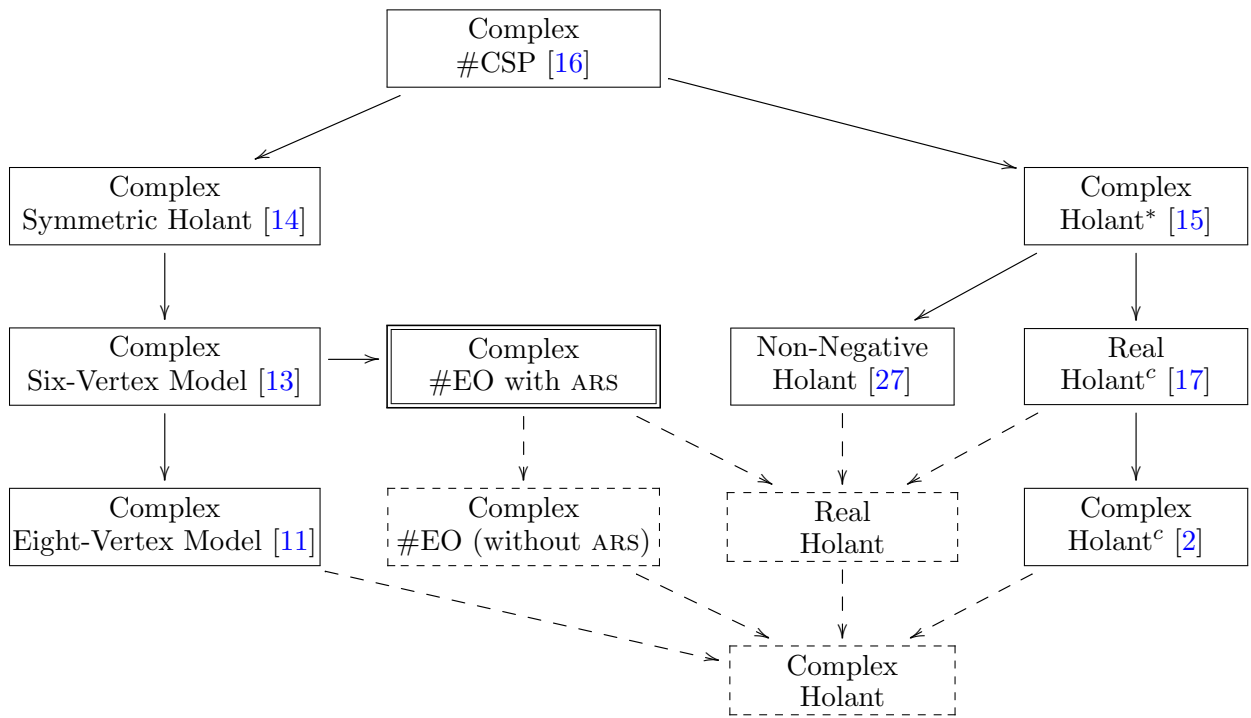
\begin{figure}[!hb]
\centering
\vspace{-1ex}
$$\xymatrix@C-=4ex{
& \framebox[22ex]{~~~~~
\txt{Complex \\ \#CSP~\cite{Cai-Lu-Xia-csp}}
~~~~~}
\ar[dl]\ar[drr] &    \\
 \framebox[25ex]{\txt{Complex\\
 Symmetric Holant~\cite{cai2016complete}}}\ar[d]
&   & 
  & \framebox[18ex]{\txt{Complex\\
 Holant$^\ast$~\cite{cai2011dichotomy}}}\ar[dl]\ar[d]  \\
  \framebox[25ex]{\txt{Complex\\
 Six-Vertex Model~\cite{cfx}}}\ar[d]\ar[r]
&  \doublebox{~~~\txt{Complex\\
 {\rm \#EO} with {\sc ars}}~~~}\ar@{-->}[d]\ar@{-->}[dr]  & \framebox[18ex]{\txt{Non-Negative\\
 Holant~\cite{beida}}}\ar@{-->}[d]
  & \framebox[18ex]{\txt{Real\\
 Holant$^c$~\cite{Cai-Lu-Xia-holant-c}}}\ar[d]\ar@{-->}[dl]  \\
  \framebox[25ex]{\txt{Complex\\
 Eight-Vertex Model~\cite{cai2017complexity}}}\ar@{-->}[drr]<-1.2ex>
&  \dboxed{\txt{Complex\\
{\#EO} 
%(\emph{without} 
(without {\sc ars})}}\ar@{-->}[dr]  & \dboxed{~~~~~~\hspace{0.5ex}\txt{Real\\
 Holant}~~~~~~\hspace{0.5ex}}\ar@{-->}[d]
  & \framebox[18ex]{\txt{Complex\\
 Holant$^c$~\cite{Backens-Holant-c}}}\ar@{-->}[dl]  \\
  &&\dboxed{~~~~~\txt{Complex\\
 Holant}~~~~~} &
}$$
\caption{A partial map 
%The blueprint 
of the complexity classification program for Holant problems}
\label{fig:structure}
 \end{figure}

The complexity classification for Holant problems is built on the dichotomy for \#CSP on the Boolean domain \cite{Cai-Lu-Xia-csp}. 
Based on this result, a dichotomy for complex-valued Holant problems with symmetric signatures (the function value depends only on the Hamming weight of the input) was established \cite{cai2016complete}.
For asymmetric signatures, 
the first result is a dichotomy for a restricted class called $\Holant^\ast$ problems where all unary signatures are assumed to be available \cite{cai2011dichotomy}.
Later, it was generalized to (first real-valued \cite{Cai-Lu-Xia-holant-c} and then complex-valued \cite{Backens-holant-plus,Backens-Holant-c}) $\Holant^c$ problems where two pinning unary signatures $(1, 0)$ and $(0, 1)$ are available.
%Naturally, the next is to consider asymmetric signatures.
%There are two tracks.
%The first track is to prove dichotomies for special families of Holant problems by assuming certain auxiliary unary signatures are available. 
%First, a dichotomy for complex-valued $\Holant^\ast$ problems where all unary signatures are available was established. 
In addition, based on the dichotomy for $\Holant^\ast$ problems, a  dichotomy for non-negative  $\Holant$ problems was proved \cite{beida} without assuming any auxiliary signatures.
 Simultaneously, progress has been made for  Holant problems parameterized by complex-valued signatures of even arities. 
The base case is a single 4-ary signature. (The case that all signatures are binary is known to be tractable.)
A dichotomy is proved for complex-valued six-vertex models \cite{cfx} and later it was generalized to complex-valued eight-vertex models \cite{cai2017complexity}.
The framework of \#EO problems generalizes six-vertex models to higher arities. 
In this paper, we prove a dichotomy for complex-valued \#EO problems with {\sc ars}, which can be transformed to a class of real-valued Holant problems.
Combining with other results, we hope this result will lead to a classification for all real-valued Holant problems. 
The ultimate goal is definitely a full classification for all complex-valued Holant problems.
In order to achieve this, we think a  classification for complex-valued \#EO problems \emph{without} assuming {\sc ars} may serve as a building block.

 \section*{Recent Development}
Most recently, after this paper was accepted, 
%a full dichotomy for real-valued Holant problems was proved \cite{real-holant} based on this result.
the vision that is enunciated in this paper
has been realized. A full complexity dichotomy
for all real-valued Holant problems on the Boolean domain has been achieved~\cite{real-holant}, and the
results in the present paper indeed played a crucial role in its proof.

\section*{Acknowledgement}
We   thank the anonymous referees for their very helpful comments. We also thank the editor Prof. Lane Hemaspaandra for his very detailed editorial comments.
Their many suggestions have helped us greatly to improve the presentation of the material.

\appendix

\section{An $8$-ary Signature  not Satisfying the $\Delta$-Property}
%\subsection{}

Some lemmas in Section~\ref{no-neq_4} (e.g., Lemmas~\ref{twononzero}, \ref{triangle}, and~\ref{3indices}) are subtle both in their precise statements and proofs. In particular,
they require high arities.
Here, we give a signature of arity $8$ that cannot be handled by using ``commutativity of merging operations'', which is the underlying ``principle'' that
 is essential for our induction proof (Lemma \ref{induction}, which requires arity at least 10). 
 This illustrates some of the intricacies and  why we have to deal with signatures of
 lower arities separately (Lemma \ref{468}).
 For convenience, we use $(i j)$ to denote the binary disequality signature on variables $x_i$ and $x_j$,
 and in the table (\ref{merging-form}) below we omit the tensor product symbol $\otimes$ between the binary 
  disequality signatures.
Let $f$ be an  EO signature satisfying {\sc ars} 
on eight variables $x_1, x_2, \ldots, x_8$.  Suppose $f$ has the following structural decompositions in $\mathcal{B}$ under various merging operations.
(The precise forms in the following 28 tensor product factorizations in $\mathcal{B}$ 
are meticulously chosen; embedded in it is a nontrivial symmetry subgroup of $S_8$.
It is not clear at this point that there are any function that can satisfy
all these requirements.)
\begin{equation}\label{merging-form}
    \begin{aligned}
    \partial_{(12)}f=(34)(56)(78), \ \partial_{(34)}f=(12)(56)(78), \  \partial_{(56)}f=(12)(34)(78), \ \partial_{(78)}f=(12)(34)(56),\\
    \partial_{(13)}f=(24)(57)(68), \ \partial_{(24)}f=(13)(57)(68), \ 
    \partial_{(57)}f=(13)(24)(68), \ \partial_{(68)}f=(13)(24)(57), \\
     \partial_{(14)}f=(23)(58)(67), \  \partial_{(23)}f=(14)(58)(67), \
    \partial_{(58)}f=(14)(23)(67),  \ \partial_{(67)}f=(14)(23)(58), \\
     \partial_{(15)}f=(26)(37)(48), \ \partial_{(26)}f=(15)(37)(48), \
    \partial_{(37)}f=(15)(26)(48), \ \partial_{(48)}f=(15)(26)(37), \\
     \partial_{(16)}f=(25)(38)(47), \ \partial_{(25)}f=(16)(38)(47), \
    \partial_{(38)}f=(16)(25)(47),  \ \partial_{(47)}f=(16)(25)(38), \\
     \partial_{(17)}f=(28)(35)(46), \ \partial_{(28)}f=(17)(35)(46), \
    \partial_{(35)}f=(17)(28)(46), \ \partial_{(46)}f=(17)(28)(35), \\
     \partial_{(18)}f=(27)(36)(45), \ \partial_{(27)}f=(18)(36)(45), \
    \partial_{(36)}f=(27)(18)(45), \ \partial_{(45)}f=(27)(36)(18). \\
    \end{aligned}
\end{equation}

 There are a total of ${8 \choose 2} = 28$ binary disequalities $(ij)$,
 each appears exactly three times in (\ref{merging-form}) (each disequality $(ij)$ appears in a single row in
 table (\ref{merging-form})).
One can verify that this $f$  does not satisfy the $\Delta$-property. We carefully design these particular factorizations so  that for any binary disequality signature $(ij)$, it divides exactly three totally disjoint pairs $\partial_{(s_1 t_1)}f$, $\partial_{(s_2 t_2)}f$ and $\partial_{(s_3 t_3)}f$, which do not form a ``triangle''. 
However (it is tedious but one can   verify that)
%More importantly, it is possible that $f$ has these ``partial orders'' by only considering commutativity.
by taking any two disjoint pairs $\{i, j\}$ and $\{u, v\}$, the above 
tensor factorizations
always satisfy $\partial_{(ij)(uv)}f=\partial_{(uv)(ij)}f$. 
(Of course any actual signature $f$ must satisfy this.)
%There is no contradiction. 
It follows that this case cannot be handled by our induction proof.

The above structural decompositions (\ref{merging-form}) under various merging operations may look like artificially designed 
and one may wonder whether there exists such a signature that  satisfies them all. 
While not obvious, the answer is yes. 
We give a signature that does satisfy  (\ref{merging-form}). %This signature is actually unique up to a nonzero scalar. 

Consider the following signature $f_8$ which has its
support set defined by
%. $f_8(\alpha)=1$ when $\alpha \in \mathscr{S}(f)$ and 
\begin{equation*}
\begin{aligned}
\mathscr{S}(f_8)=\{(x_1, x_2, \ldots, x_8)\in \{0, 1\}^{8} \mid~~
&x_1+x_2+x_3+x_4\equiv 0 \bmod2, ~ x_5+x_6+x_7+x_8\equiv 0 \bmod2,\\
&x_1+x_2+x_5+x_6\equiv 0 \bmod2, ~ x_1+x_3+x_5+x_7\equiv 0 \bmod2,\\
&\text{ and } 
\sum\nolimits_{i=1}^8x_i = 4
\}.
\end{aligned}
\end{equation*}
The function is defined to be
$f_8(\alpha)=1$ when $\alpha \in \mathscr{S}(f)$ and $f_8(\alpha)= 0$ elsewhere.
%%% is it really clearly?? 
%Clearly, 
It can be verified that $f_8$ is an EO signature with {\sc ars}. Its support has the following structure: 
(1) it is on half-weight;
(2) the sum of the first four variables is even;
(3) the last four variables are either 
identical to the first four variables, or 
are the complement of the first four variables.

%, and the last four variables are both even; the assignment of the first four variables are either identical to, or complement of the assignment of the last four variables.

One can verify that the signature $f_8$ satisfies forms (\ref{merging-form}).
Thus, there does exist a case that cannot be handled by our induction proof alone.
This partially explains some of the intricacies of our proof.
%, and we have to use the orthogonality property to handle it.

\bibliography{eo}{}

\end{document}